\documentclass[12pt]{article}

\pdfoutput=1

\usepackage{latexsym,amsfonts,bm,epsfig,amsmath,natbib,authblk,amsthm,thmtools,amssymb,mathtools}
\usepackage{upgreek}
\usepackage[toc,page,header]{appendix}
\usepackage{minitoc}

\usepackage{ulem}

\usepackage[utf8]{inputenc}

\usepackage{IEEEtrantools}
\usepackage[shortlabels, inline]{enumitem}
\usepackage{float}
\usepackage{booktabs}
\usepackage{graphicx}
\usepackage{subcaption}
\usepackage[margin=.5cm]{caption}
\usepackage{xr-hyper}
\usepackage{color}
\usepackage[colorlinks,linkcolor=black,citecolor=black,filecolor=black,bookmarks=false,pagebackref]{hyperref}
\usepackage[algoruled, noend]{algorithm2e}
\usepackage{mathrsfs}
\usepackage{bigints}
\usepackage{acronym}

\usepackage{pgfplots}
\usepackage{tikz}
\newlength\figureheight
\newlength\figurewidth

\floatstyle{plain}

\makeatletter
\newcommand{\myitem}[1]{%
\item[#1]\protected@edef\@currentlabel{#1}%
}
\makeatother

\acrodef{pro}[{\scshape PrO}]{predictively oriented}

\definecolor{procolour}{HTML}{4477AA}
\definecolor{gibbscolour}{HTML}{EE6677}
\definecolor{bnncolour}{HTML}{CCBB44}
\definecolor{mgpcolour}{HTML}{228833}

\makeatletter
\renewcommand{\algocf@captiontext}[2]{#1\algocf@typo. \AlCapFnt{}#2} %
\def\@algocf@capt@plain{top}
\renewcommand{\algocf@makecaption}[2]{%
	\addtolength{\hsize}{\algomargin}%
	\sbox\@tempboxa{\algocf@captiontext{#1}{#2}}%
	\ifdim\wd\@tempboxa >\hsize%
	\hskip .5\algomargin%
	\parbox[t]{\hsize}{\algocf@captiontext{#1}{#2}}%
	\else%
	\global\@minipagefalse%
	\hbox to\hsize{\box\@tempboxa}%
	\fi%
	\addtolength{\hsize}{-\algomargin}%
}
\makeatother

\newcommand{\ignore}[1]{}

\newcommand{\x}{{x}_{1:n}}

\newcommand{\wh}{\widehat}

\newcommand{\normal}{\ensuremath{\mathsf{N}}}

\newcommand{\dt}{\mathsf{d}}
 
\newcommand{\E}{\mathbb{E}}
\newcommand{\KL}{\mathrm{KL}}

\newcommand{\argmin}{\operatornamewithlimits{argmin}} 
\newcommand{\arginf}{\operatornamewithlimits{arginf}}

\newtheorem{assumption}{Assumption}
\newtheorem{theorem}{Theorem}

\newtheorem{corollary}{Corollary}
\newtheorem{lemma}{Lemma}
\theoremstyle{definition}
\newtheorem{definition}{Definition}
\theoremstyle{remark}

\theoremstyle{example}

\newtheorem{assumptionalt}{Assumption}[assumption]

\newenvironment{assumptionp}[1]{
  \renewcommand\theassumptionalt{#1}
  \assumptionalt
}{\endassumptionalt}

\usepackage{enumitem}
\usepackage{hyperref}

\begin{document}

	\def\spacingset#1{\renewcommand{\baselinestretch}%
		{#1}\small\normalsize} \spacingset{1}
	
	\title{Predictively Oriented Posteriors}
	\date{\empty}
	
\author[1]{Yann McLatchie$^{\star}$}
\author[2]{Badr-Eddine Ch\'{e}rief-Abdellatif$^{\star}$}
\author[3]{David T. Frazier$^{\star}$}
\author[1]{Jeremias Knoblauch$^{\star}$}
\affil[1]{Department of Statistical Science, University College London, UK}
\affil[2]{CNRS, Sorbonne Universit\'{e}, France}
\affil[3]{Department of Econometrics \& Business Statistics, Monash University, Australia}

	\maketitle
	\def\thefootnote{$\star$}\footnotetext{All authors contributed equally. Corresponding author: \url{j.knoblauch@ucl.ac.uk}}\def\thefootnote{\arabic{footnote}}
	
	\begin{abstract}
            We advocate for a new statistical principle that combines the most desirable aspects of both parameter inference and density estimation.
            This leads us to the \ac{pro} posterior, which expresses uncertainty as a consequence of predictive ability.
            We show that these posteriors converge to the predictively optimal model average and predictively dominate both classical and generalised Bayes posterior predictive distributions.
            Further, \ac{pro} posteriors {adapt} to the level of model misspecification: while they concentrate around the true model in the same way as classical and generalised Bayesian strategies if the model can recover the data-generating distribution, they do \textit{not} concentrate around a single model in the presence of non-trivial forms of model misspecification.
            Instead, they stabilise towards a non-degenerate predictively optimal posterior distribution that represents a form of irreducible uncertainty due to model misspecification.
            We put forward a sampling algorithm for \ac{pro} posteriors based on mean field Langevin dynamics, and verify the practical significance of our theoretical  developments on a number of numerical examples. 
	\end{abstract}
	\spacingset{1.5}

Often motivated through its favourable decision-theoretic properties, the primary goal of Bayesian statistical inference is to quantify uncertainty about a model parameter of interest given a set of prior beliefs \citep[e.g.][]{robert2007bayesian}. 
All other inferential tasks  are derivatives of this primary directive. 
In the current era of black-box prediction models and the wide-spread adoption of machine learning techniques across science and society, this principle leads to a fundamental tension with Bayesian statistical methods: in many modern problems, prediction is the primary goal of interest, with parameter inference being only a secondary concern.
Because Bayesian methodology is 
centered on parameter inference, its predictive performance quickly deteriorates  outside of the narrow conditions that make Bayesian belief updates optimal  \citep[see][]{zellner_optimal_1988, aitchison1975goodness}.

One of the clearest ways in which the Bayesian focus on parameter inference manifests is \textit{posterior concentration}: as data accumulate, the Bayes posterior generally concentrates around a single parameter value \citep[see e.g][Section 10.2]{vdv} representing the best ``single model'' that explains the observed data.
In misspecified settings however, this behaviour may be undesirable for down-stream tasks like model criticism or prediction: in large samples, Bayesian uncertainty tightly concentrates around the best-fitting model, so that the associated
predictive distribution amounts to predicting from this single best-fitting model.
Further, this behaviour persists even for generalised Bayesian methods that replace a likelihood-based approach with a loss-based belief update \citep{bissiri:etal:2016,jiang_gibbs_2008,zhang_e-entropy_2006,knoblauch2019generalized,alquier2024user,syring2019calibrating, martin_direct_2022} and resampling-based approaches such as the Martingale posterior \citep{fortini20quasibayes, fortini23predictionbased, fong2023martingale}.
As a result, posterior predictive distributions  obtained from both approaches often resemble plug-in predictives based on point estimators \citep{mclatchie_predictive_2025}, and incorporate no meaningful degree of model uncertainty.

Posterior concentration around the best-fitting model parameterisation is desirable if the data has been generated from the assumed model.
In practice, however, models are frequently structurally misspecified so that there are aspects of the data which the model is fundamentally incapable of capturing.
In this situation, posterior concentration  is often at odds with scientists'  intuitive understanding of what uncertainty quantification \textit{should} be.
In particular, if a model is known to be an imperfect description of reality, scientists are most interested in  uncertainty that communicates whether there are certain aspects of the data which \textit{no} parameterisation of their scientific models can  explain.

The \textit{muon  anomaly} in particle physics is a popular case of this type of uncertainty, and owes its name to the fact that the widely-agreed `Standard Model' cannot correctly predict the magnetic moment for muons \citep[][]{aliberti2025anomalous}.
The \textit{Hubble Tension} in cosmology exposes a similar situation: the  $\Lambda$CDM model predicts a rate of expansion for the universe---the so-called Hubble constant---that is contradicted by actual measurements \citep[][]{di2021realm}.
More generally speaking, seeing a larger amount of the same (unexplainable) data will \textit{not} make scientists more certain because the predictions made by \textit{any} of the corresponding models or parameterisations would not match the data they observe.
This is a stark contrast to the notion of uncertainty enshrined by posterior concentration: %
as more data accumulates,  the posterior becomes increasingly certain about the best-fitting model parameterisation---even when the assumed model is a poor description of the data.

In this paper, we conceive of an approach for \textit{predictively oriented (\ac{pro})} uncertainty that provides a mathematically rigorous template for the predictively oriented understanding of uncertainty articulated across many types of scientific analysis.
Unlike Bayes posteriors and their generalisations, \ac{pro} posteriors treat prediction as the primary inferential goal, and generate parameter uncertainty as a consequence of a parameter's impact on the accuracy of the resulting predictions.
Critically, this approach to uncertainty quantification is \textit{adaptive}: 
while the \ac{pro} posterior collapses to a point mass \textit{only} when the model is well-specified, it will  provide \textit{irreducible} { uncertainty}---in the form of non-singular and predictively optimal posteriors---whenever the model is misspecified.
Consequently, in predictive terms, the resulting \ac{pro} posterior predictive distributions dominate both Bayesian and generalised Bayesian methods {in finite samples}.

Conceptually, the \ac{pro} posterior communicates uncertainty by attaching to each possible model parameterisation a weight  proportional to the fraction of data which is best described by that particular parameterisation.
When the model is well-specified, \textit{all} data is explained by a \textit{single} parameter value, and the \ac{pro} posterior concentrates onto this value.
If the model is misspecified however, good explanations for different facets of the data will generally correspond to \textit{different} parameter values, and the weights assigned to them by the \ac{pro} posterior will reflect their relative frequency.

\section{Background}
\label{sec:motivation}

Throughout, we consider  data $\x=(x_1,\dots,x_n)$ for $x_i\in\mathcal{X}$ and $i=1,\dots,n$.
We assume that $\x$ was generated from an unknown distribution $P_0$.
With this, we denote expectations of a function $f:\mathcal{X}\to\mathbb{R}$ with respect to $P_0$ by $\mathbb{E}[f(X)]$ or by $\mathbb{E}_{X\sim P_0}[f(X)]$ if we wish to emphasise the sampling distribution.
To learn about $P_0$ from  $x_{1:n}$, we follow the tenets of classical Bayesian analysis, and posit a model class $\mathcal{M}_{\Theta} := \{P_\theta:\theta\in\Theta\}$ and a prior $\Pi \in \mathcal{P}(\Theta)$ over elements of $\mathcal{M}_{\Theta}$, where $\mathcal{P}(\Theta)$ denotes the space of probability measures on $\Theta$.
Instead of attempting to learn the best possible \textit{parameter}, as in classical and generalised Bayesian inference, we focus on learning the best possible \textit{predictive} distribution.
To rigorously formulate this optimality, we use scoring rules.

\begin{definition}[Scoring Rule]
\label{def:scoring-rule}
Let $S:\mathcal{P}(\mathcal{X}) \times \mathcal{X} \to \mathbb{R}\cup\{\infty\}$ be an extended real-valued function. Then $S$ is a proper scoring rule, relative to a class $\mathcal{P}$, if:
\vspace*{-0.3cm}
\begin{itemize}
    \item[(i)]  the expected score $\mathcal{S}(P, P') := \E_{X\sim P'}\left\{S(P,X)\right\}$, is finite for some $P,P'\in\mathcal{P}$;
    \vspace*{-0.3cm}
    \item[(ii)]
    $\mathcal{S}(P',P') \le \mathcal{S}(P,P')$  for all $P,P' \in \mathcal{P}$, so that $\mathcal{S}$ induces the statistical divergence
    $\mathcal{D}_{S}(P,P') := 
    \mathcal{S}(P,P')
    -
    \mathcal{S}(P',P') \geq 0$ for which $\mathcal{D}_{S}(P,P') = 0 \Longleftrightarrow P = P'$.
\end{itemize}
\end{definition}
Scoring rules are foundational statistical objects concerned with measuring  accuracy of  predictive distributions (see, e.g., \citealp{gneiting2007strictly}). Given a target distribution $P_0$, a candidate predictive distribution $P$ is judged to be more accurate than an alternative candidate $P'$ by {an expected} score  $\mathcal{S}$ if and only if $\mathcal{S}(P, P_0) \leq \mathcal{S}(P', P_0)$. 
The logarithmic scoring rule, given by
$S(P, x) = - \log \dt P(x)$, 
is the most common example and is related to the Kullback-Leibler divergence via $ \KL(P' \| P)=\mathcal{D}_{S}(P, P')$.
Other popular choices include
the {spherical score,} {cumulative ranked probability score (CRPS)}, the energy score, and squared maximum mean discrepancy (MMD), which for a kernel function $k:\mathcal{X}\times\mathcal{X} \to\mathbb{R}_+$ and the Dirac measure at $x$ denoted by $\updelta_x$ 
is given by $\mathrm{MMD}^2(P,\updelta_{x}):=\E_{X\sim P,X'\sim P} [k(X,X')]
+k(x,x)
-2\E_{X\sim P}[k(X,x)]. $

\subsection{Key Ideas}
\label{sec:key-ideas}
The Bayes posterior $\Pi_n$ and its associated posterior predictive distribution $P_{\Pi_n}$ are
\begin{IEEEeqnarray}{rCl}
    \Pi_n(A):=\frac{\int_A  \dt P_\theta (x_{1:n})\dt \Pi(\theta)}{\int_\Theta  \dt P_\theta (x_{1:n})\dt \Pi(\theta)},
    \quad \quad 
    P_{\Pi_n}:=\int_\Theta P_\theta \dt \Pi_n(\theta).
    \nonumber %
\end{IEEEeqnarray}
Here, the predictive distribution $P_{\Pi_n}$ is derived as a consequence of the conditional probability calculation in $\Pi_n$---which itself is primarily geared towards quantifying  the uncertainty about which parameter $\theta$ best describes the observed data $x_{1:n}$.
In this paper, we propose a methodology that seeks to reverse this order: we propose to make prediction the primary object of interest, and do so by constructing a novel posterior that quantifies uncertainty about $\theta$ in {a manner that is optimal for prediction.}
The key motivation for this novel methodology is that while $P_{\Pi_n}$ is predictively optimal if $P_0 \in \mathcal{M}_{\Theta}$ \cite[see][]{aitchison1975goodness}, this optimality breaks down once $P_0 \notin \mathcal{M}_{\Theta}$ {(see e.g.  \citealp{grunwald2017inconsistency})}.
To understand the inadequacy of predictive distributions induced by $\Pi_n$, it is useful to recast the Bayes posterior as a variational problem \cite[see e.g.][]{zellner_optimal_1988}.
In particular, taking $\operatorname{KL}$ to be the  Kullback-Leibler divergence, and assuming for simplicity that the model $P_{\theta}$ pertains to independently sampled observations and admits a density $p_{\theta}$, the Bayes posterior can  equivalently be represented as 
\begin{IEEEeqnarray}{rCl}
\Pi_n =
\argmin_{Q\in\mathcal{P}(\Theta)}\left\{-\int\frac{1}{n}\sum_{i=1}^{n}\log p_\theta(x_i)\dt Q(\theta)+\frac{\KL(Q\|\Pi)}{n}\right\}.  
\nonumber
\end{IEEEeqnarray}
The above formulation of $\Pi_n$ allows for a simple but powerful insight: as $n\rightarrow\infty$, the  impact of our prior beliefs on the Bayes posterior diminishes.
To see this, simply note that the KL-based prior regularisation decays to zero while  the negative log likelihood term will usually converge to
$\int\mathbb{E}_{X\sim P_0}[-\log p_{\theta}(X)]\dt Q(\theta)$.
Notably, up to a constant independent of $Q$, this limiting function is equivalent to the average KL divergence over models given by
$\int \KL(P_0\|P_\theta)\dt Q(\theta)$. 
Therefore, one can generally expect the  variational program defining the Bayes posterior $\Pi_n$  to collapse onto to the minimizer of  $Q \mapsto \int \KL(P_0\|P_\theta)\dt Q(\theta)$, which is in turn minimised by a Dirac mass at $\theta^\star = \argmin_{\theta\in\Theta}\mathbb{E}_{X\sim P_0}[-\log p_{\theta}(X)]$ for most regular statistical models.
Crucially, finding the best-fitting parameter will generally not deliver accurate predictions outside of well-specified models \cite[see][]{aitchison1975goodness}. 
Consequently, unless  $P_0=P_{\theta^\star}$ for some $\theta^\star\in\Theta$, the predictive $P_{\Pi_n}$ associated to the Bayes posterior will generally \textit{not}  provide adequate predictive uncertainty quantification---even as we increase the data volume used to form our posterior beliefs.

In this paper,  we tackle this challenge. 
In particular, we develop a posterior inference method whose primary inferential goal is to maximise agreement between the posterior predictive  and the data, 
and whose parameter uncertainty is a direct consequence of this goal.
This leads us to the \textit{predictively oriented} (\ac{pro}) posterior, which we formulate by modifying the variational characterisation of Bayesian inference.
Rather than targeting the best-fitting model parameterisation, we instead target a prior-regularised version of the best-fitting predictive $P_{Q} := \int_\Theta P_\theta \dt Q(\theta)$ through the posterior distribution on $\Theta$ obtained as the minimiser
\begin{IEEEeqnarray}{rCl}
\argmin_{Q\in\mathcal{P}(\Theta)}\left\{-\frac{1}{n}\sum_{i=1}^{n} \log
P_Q(x_i)
+\frac{\KL(Q\|\Pi)}{n}\right\}.    
\nonumber
\end{IEEEeqnarray}
By similar arguments as before, we now observe that as $n\to\infty$, the above objective collapses onto $Q \mapsto \operatorname{KL}(P_0, P_Q)$. 
This limiting objective   is minimised by a Dirac mass at $\theta^{\star}$ \textit{only} if the model is well-specified, but generally yields a non-degenerate and predictively optimal model average $Q^{\star} \in \argmin_{Q \in \mathcal{P}(\Theta)}\operatorname{KL}(P_0, P_Q)$ outside the well-specified setting.

Here, the choice of the logarithm  to score the quality of a prediction is arbitrary, and
{can be immediately generalised to  treat prediction problems based on any proper scoring rule $P\mapsto S(P,\cdot)$ appropriate to the problem of interest.}
Doing so delivers the general formulation of  the \ac{pro} {posterior} used throughout the methodological and theoretical development of the paper: for some learning rate $\lambda_n>0$, 
\begin{equation}
    {Q}_n = \argmin_{Q\in\mathcal{P}(\Theta)}\left\{\frac{\lambda_n}{n}\sum_{i=1}^n S\left(P_Q, x_i\right)+\KL(Q\|\Pi)\right\}.
   \label{eq:pro-posterior-def}
\end{equation} 
Alternatively, the \ac{pro} posterior can also be represented through a fixed point  as
\begin{IEEEeqnarray}{rCl}
    \dt Q_n(\theta) & \propto & \dt \Pi(\theta) \cdot \exp\left\{ - \frac{\lambda_n}{n} \sum_{i=1}^n \left.\frac{\delta S(P_Q, x_i)}{\delta Q}\right|_{Q = Q_n}(\theta) \right\}.
    \nonumber
\end{IEEEeqnarray}
While the latter is of limited use for understanding the qualitative behavior of \ac{pro} posteriors, it illustrates two key features. 
First, the \ac{pro} posterior is generally not a coherent belief update in the sense of \citet{bissiri:etal:2016}: upon observing  $x_{n+1}$, there is no simple update rule to obtain $Q_{n+1}$ from both $Q_n$ and $x_{n+1}$ alone. 
Second, even for simple models, $Q_n$ does not admit an explicit density or measure representation.

The \ac{pro} posterior is distinct from other robust Bayesian methods such as \textit{Gibbs posteriors}, which still treat prediction as a the primary goal of inference;
\cite[e.g.][]{
zhang_e-entropy_2006,
jiang_gibbs_2008,
jewson2018principles, 
knoblauch2018doubly,
giummole_objective_2019,
syring2019calibrating,
schmongeneralized2020,
cherief2020mmd, matsubara2022robust,
martin_direct_2022,
matsubara2023generalisedDFD,
altamirano2023robustGP,
altamirano2023robustCP,
duran2024outlier,
mclatchie_predictive_2025,
laplante2025robust,
frazier2025impact}.
{Many of these} posteriors are based on an scoring rules  evaluated on sample data, and can be specified using the following variational optimisation problem:
\begin{IEEEeqnarray}{rCl}
    Q_n^{\dagger} & = &\argmin_{Q \in \mathcal{P}(\Theta)}\left\{\frac{\lambda_n}{n}\sum_{i=1}^n\int S(P_\theta, x_i)\dt  Q(\theta)+\KL(Q\|\Pi)\right\}.
    \label{eq:gibbs-posterior-as-opt-problem}
\end{IEEEeqnarray}
The above representation is well-established  \citep[see e.g., Theorem 1 of][]{knoblauch2019generalized}, 
and  makes clear that \textit{all} Gibbs posteriors---including the Bayes posterior, which is obtained by taking $S(P_\theta,x)=-\log p_\theta(x)$---seek to minimise the averaged scoring rule, rather than the score of the averaged model.
This in turn means that the inferences they produce are targeted towards finding the \textit{parameter}, rather than the \textit{predictive}, which best describes the data.
Crucially, {these two goals do not overlap outside of well-specified models,} where $P_0 \in \mathcal{M}_{\Theta}$ \cite[see][]{aitchison1975goodness}.

\subsection{Related Work}

The class of \ac{pro} posteriors we formulate in \eqref{eq:pro-posterior-def} represents a novel statistical paradigm for inference, but shares  similarities with several pre-existing approaches discussed below.
\begin{enumerate}
    \item The earliest attempts at computing objects similar to $Q_n$ date back to the  \textit{$\operatorname{PAC}_T^2$-variational posterior}  \citep{masegosa2020learning} and the \textit{$\operatorname{PAC}^m$-Bayes} approach \citep{morningstar2022pacm}.
    Both methods only consider the special case of the logarithmic score $S(P, x) = - \log \dt P(x)$, and the computational tools developed are variational approximations.
    Given the focus on variational methods, computational feasibility required various approximations to the log score.
    Relative to \ac{pro} posteriors, these approaches necessitate approximations both to the objective and the posterior of interest. These dual approximations imply that these methods are strictly outside the class defined  in \eqref{eq:pro-posterior-def}, and result in predictive methods whose generalisation errors are  larger than those associated to \ac{pro} posteriors.
    \item Relatedly, \cite{lai2024predictive} propose what they call \textit{predictive variational inference}, which differs from the optimization problem to \eqref{eq:pro-posterior-def} in two critical aspects.
    First, optimisation is not considered over $\mathcal{P}(\Theta)$, but over a parametric subset $\mathcal{Q}\subset \mathcal{P}(\Theta)$ for which variational methods are computationally feasible.
    Second, the approach also considers regularisation against the Bayes posterior (not just against the prior).
    These differences may seem small, but are crucial.
    The \ac{pro} posterior is neither representable as a fixed form variational family nor will it be asymptotically normal in general. Hence,  approximating it with a convenient variational family $\mathcal{Q}$ deliberately introduces an irreducible approximation error that is  asymptotically non-negligible. 
    Further---and ignoring the additional complication of variational approximation for the moment---regularisation against the  Bayes posterior rather than the prior implies a  target measure $\widetilde{Q}_n$ that is fundamentally different from the \ac{pro} posterior.
    In particular, since $\operatorname{KL}(Q, \Pi_n) = \int \sum_{i=1}^{n}-\log \mathsf{d}P_\theta(x_i)\mathsf{d}Q(\theta) + \operatorname{KL}(Q, \Pi)$, it defines
    $$
    \widetilde{Q}_n = \argmin_{Q\in\mathcal{P}(\Theta)}\left\{\sum_{i=1}^n\left[\frac{\lambda_n}{n} S\left(P_Q, x_i\right) -\int \log \mathsf{d}P_\theta(x_i)\mathsf{d}Q(\theta)\right]+\KL(Q\|\Pi)\right\}.
    $$
    This target is \textit{neither} geared towards prediction \textit{nor} parameter inference, and instead seeks to compromise between them. 
    While this interpolation may be interesting in some situations, it is not the setting studied in the current paper.
    Further, the predictive dominance results for the \ac{pro} posterior derived in the remainder of this paper apply \textit{neither} to posteriors that are variationally approximated \textit{nor} to objectives that regularise against the Bayes posterior.

    \item More directly related to our work is the work of \citet{shenprediction2025}, which proposes the special case of  \eqref{eq:pro-posterior-def} with the MMD in the specific context of misspecified deterministic simulation models. 
    In this setting, the approach was conceived as a heuristic ``for arresting the collapse of predictive uncertainty'' and to avoid that ``posterior predictions become deterministic.''
    Unlike what we derive in the remainder, \citet{shenprediction2025} also did not rigorously investigate the behaviour of their approach. 
    Beyond that, their work neither advocates for \ac{pro} inference as a general-purpose methodology, nor for the general formulation in \eqref{eq:pro-posterior-def}.
    \item 
    In the special case of the logarithmic scoring rule $S(P, x) = - \log \dt P(x)$, and when the prior regularisation term is dropped, \eqref{eq:pro-posterior-def} recovers \textit{nonparametric maximum likelihood estimation (NPMLE)}.
    Both practically and theoretically however, NPMLE and \ac{pro} posteriors are in quite dissimilar.
    On a practical level, dropping the KL-regularisation may cause non-identifiability even in very simple models \citep[see e.g.][]{laird1978nonparametric}, and may lead to non-unique or fully atomic measures as solutions; see e.g. \citet[][Theorem 21 in Chapter 5]{lindsay1995mixture} or \citet{jordan2015convex}; and so are ill-suited for  probabilistic uncertainty quantification in a generalised Bayesian sense.
    More foundationally, NPMLE and \ac{pro} posteriors have fundamentally different inferential goals. 
    \ac{pro} posteriors conceive of $\{P_{\theta}:\theta \in \Theta\}$ as the model of interest and use $Q_n$ to quantify predictively oriented posterior uncertainty about the parameterisation of this model relative to some prior belief $\Pi\in\mathcal{P}(\Theta)$. 
    In contrast, the NPMLE neither depends on prior beliefs nor constitutes an approach to compute posterior beliefs. 
    Instead, it treats the nonparametric mixture $\{P_Q=\int P_{\theta}\mathsf{d}Q: Q \in \mathcal{P}(\Theta)\}$ as the model of interest, and thus treats  $Q \in \mathcal{P}(\Theta)$ as the parameter of interest, for which it then obtains a point estimate. 
    Regularised versions of NPMLE exist, but differ fundamentally  from the \ac{pro} posterior: 
    regularisers deployed in NPMLE 
     do not 
    enforce absolute continuity with respect probabilistically formulated prior beliefs, and more generally speaking do not depend on Bayes-like prior beliefs at all. (see, e.g., \citealp{wilkins2026data, dunker2025regularized}).

\end{enumerate}

In the remainder, we proceed to study the general formulation of \ac{pro} posteriors in \eqref{eq:pro-posterior-def}. 
We do so because our findings suggest  that \ac{pro} posteriors provide a new form of  uncertainty quantification that is uniquely suited for the predictively oriented inferences that applied scientists are increasingly demanding from statistical procedures.

\section{Misspecification and Motivating Examples}

To develop the key intuitions of our methodology, we present two simple  examples that illustrate the behaviour of \ac{pro} posteriors under model misspecification.
Based on the insights developed from these examples, we will rigorously define different regimes of model misspecification that are important for the theoretical guarantees we will derive.

\subsection{Gaussian Location Model}\label{sec:init-comparison}
In contrast to Bayes and Gibbs posteriors, \ac{pro} posteriors \textit{adapt} to model misspecification.
We illustrate what this means by comparing the predictive distributions associated to $Q_n$ and $Q_n^{\dagger}$ for the model class  $\mathcal{M}_\Theta = \{P_\theta:\,\mathsf{N}(\theta,\sigma^2 = 1)\}$ and a range of different data-generating distributions $P_0$.
The Gaussian model is  particularly educational here as its predictive distribution is obtained from the posterior via smoothing with a Gaussian kernel of unit scale.
For predictive scores based on squared MMD,  Figure \ref{fig:initial-comparison} 
plots the results for 
the well-specified setting where 
$P_0 \in \mathcal{M}_{\Theta}$ and three misspecified  settings where $P_0\notin\mathcal{M}_\Theta$.
While  \ac{pro} posteriors behave similarly to Gibbs posteriors if the model is well-specified, they  continue to produce reasonable predictive distributions under misspecification.
In contrast, Gibbs posterior inferences are blind to misspecification.

As our example uses Gaussians, one may misinterpret the phenomenon displayed as a pathology.
This is not the case: in parametric models, Bayes posteriors and Gibbs posteriors are usually asymptotically Gaussian, and will concentrate around a single element of $\mathcal{M}_{\Theta}$.
As we formally show later on, this property is not shared by \ac{pro} posteriors. 
Instead, they stabilise towards the predictively optimal average of elements in $\mathcal{M}_{\Theta}$ for $P_0$. %
As a result---and unlike Bayes posteriors and Gibbs posteriors---\ac{pro} posteriors produce predictively meaningful parameter uncertainty even under severe forms of misspecification.

\begin{figure}[t!]
\centering
\includegraphics[width=1\textwidth]{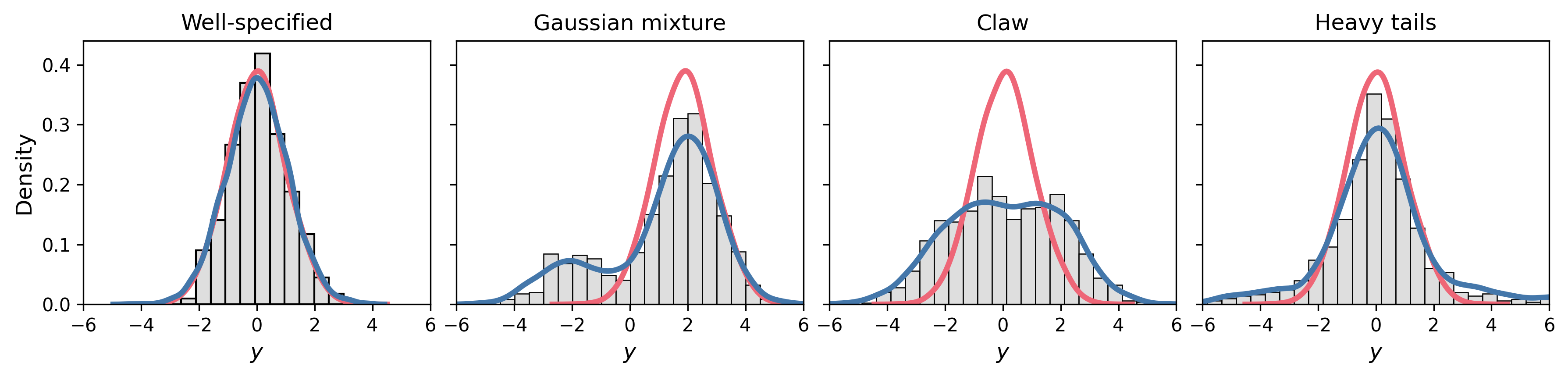}
\caption{Comparison of {\color{procolour}\textbf{\ac{pro} posteriors}} and {\color{gibbscolour}\textbf{Gibbs posteriors}} in a well-specified and three misspecified regimes.
\textbf{\color{gray}Grey bars} indicate a histogram of the observed data.
}
\label{fig:initial-comparison}
\end{figure}

\subsection{Linear Regression}
\label{sec:misspecification-regimes}

Throughout, robustness to model misspecification is an important motivation in our development of \ac{pro} posteriors.
Robustness is also one of the key drivers behind many Gibbs posteriors, usually to guard against Huber contamination \citep{altamirano2023robustGP,altamirano2023robustCP,matsubara2022robust,matsubara2023generalisedDFD,laplante2025robust,duran2024outlier, marusic2025theoretical} .
Whenever the nature of misspecification is more complicated than can be described through an outlier model of the Huber type, these approaches fail to provide robustness.

To contrast this behaviour with that of \ac{pro} posteriors, we consider a linear regression model which for $x_i \in \mathbb{R}$, $z_i \in \mathbb{R}^2$ and $i=1,\dots n$ is given by 
\begin{IEEEeqnarray}{rCl}
    x_i & = & z_i^{\top}\theta + \varepsilon_i.
    \nonumber
\end{IEEEeqnarray}
Here, we make two assumptions:  the parameter of interest $\theta \in \mathbb{R}^2$ is fixed but unknown, and  $\varepsilon_i \overset{\text{iid}}{\sim}\mathcal{N}(0,\sigma^2)$ for a fixed and known noise level $\sigma^2$.
Relative to this model for the data, we explore three different forms of misspecification:
\begin{enumerate}
    \item[(T)] 
    Heavy-tailed observation noise $\varepsilon_i \overset{\text{iid}}{\sim}   t_{3}(0, \sigma^2)$;
    \item[(CR)] 
    Random coefficients $\theta_2 = 2\cdot (-1)^{\xi_i}$, with $\xi_i \overset{\text{iid}}{\sim} \operatorname{Ber}(1/2)$  sampled from a Bernoulli;
    \item[(NT)] 
    The combined effect of $\varepsilon_i \overset{\text{iid}}{\sim}   t_{3}(0, \sigma^2)$ and $\theta_2 = 2\cdot (-1)^{\xi_i}$ with $\xi_i \overset{\text{iid}}{\sim} \operatorname{Ber}(1/2)$. 
\end{enumerate}

\begin{figure}[t!]
    \centering
    \includegraphics[width=\linewidth]{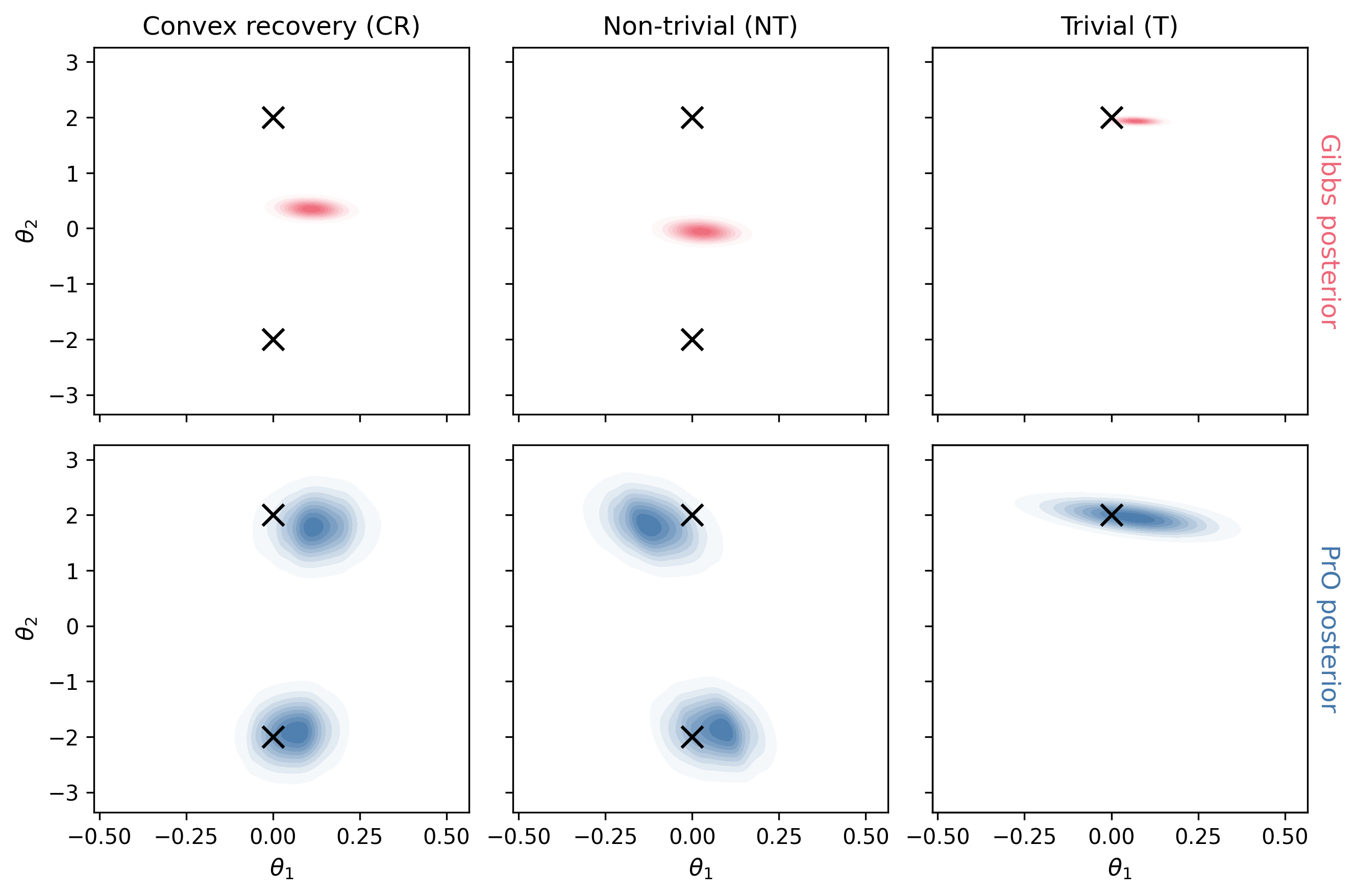}
    \caption{
    The \textcolor{gibbscolour}{\textbf{Gibbs}} and \textcolor{procolour}{\textbf{\ac{pro}}} posterior distributions under different forms of model misspecification.
    Coefficients $\theta_1, \theta_2$ are shown as black crosses.
    }
    \label{fig:linear-regression-posteriors}
\end{figure}

The first setting  constitutes a form of \textit{trivial (T)} misspecification: the inflated noise leads to overconfident inferences, but does not compromise the meaning of the model parameter $\theta$.
In line with this, we would expect both Gibbs and \ac{pro} posteriors to perform reasonably well at identifying the general location of this parameter.
In contrast, the second setting fundamentally compromises the interpretation of inferences on $\theta$.
However, it allows \textit{convex recovery (CR)}: a 2-component mixture over the misspecified linear regression model recovers the actual data-generating process.
While Gibbs posteriors will be adversely affected by this setting, \ac{pro} posteriors will be able to recover the true data-generating process exactly.
The last setting merges both of the previous complications, and constitutes a form of \textit{non-trivial (NT)} misspecification: not even convex combinations of the original model class will recover the true data-generating process exactly.
Still, we expect that the \ac{pro} posterior will perform well, since a convex combination of the original model will get us very close to the true data-generating mechanism.

\begin{figure}[b!]
    \centering
    \includegraphics[width=\linewidth]{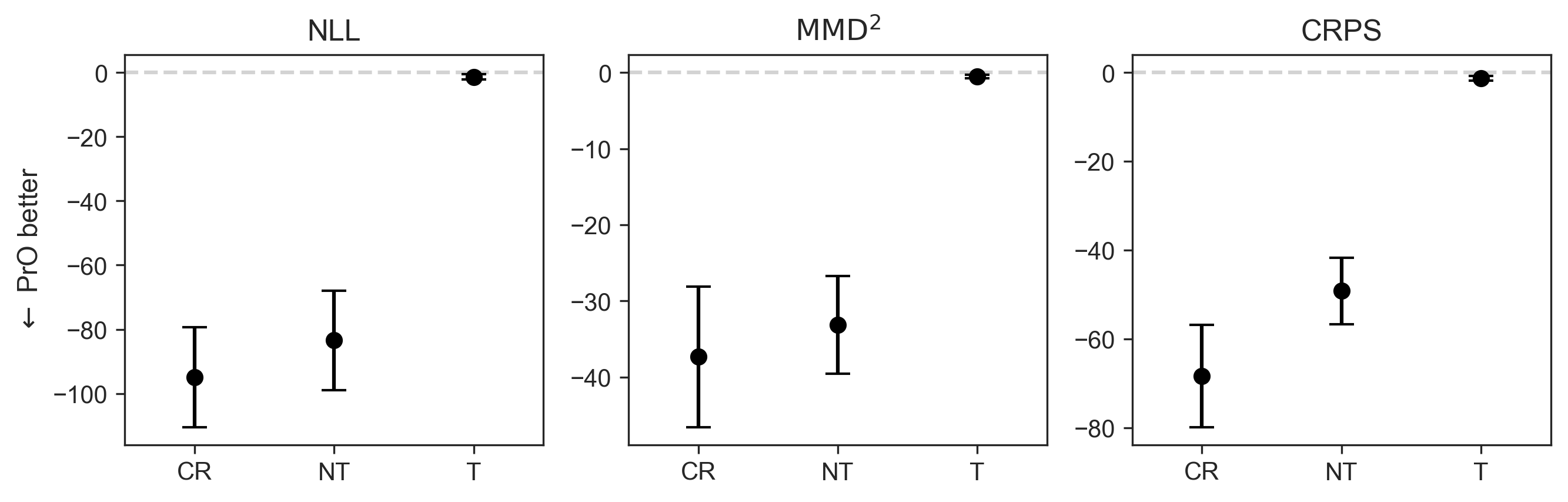}
    \caption{
    Mean and standard error of \textcolor{gibbscolour}{\textbf{Gibbs}} posterior's predictive loss relative to \textcolor{procolour}{\textbf{\ac{pro}}} posteriors evaluated using negative log likelihood (NLL), squared MMD ($\operatorname{MMD}^2$), and Continuous Ranked Probability Score (CRPS) on a test set for different forms of misspecification.
    }
    \label{fig:linear-regression-pred-diff}
\end{figure}

Our numerical results confirm these intuitions, and are summarised in Figure~\ref{fig:linear-regression-posteriors}, where we compare posterior inferences for Gibbs and \ac{pro} posteriors based on the MMD in all three settings, with implementation details and further information  deferred to Appendix \ref{sec:synthetic-regression-supp}.
Unlike the Gibbs posterior, the \ac{pro} posterior performs  uncertainty quantification for both convex recovery and the non-trivially misspecified setting by accounting for the distortions of the data-generating mechanism relative to the posited simple linear regression.
Even in the trivially misspecified setting, the \ac{pro} posterior's parameter uncertainty seems a more appropriate reflection of the true underlying uncertainty.
The reason is that the \ac{pro} posterior is driven by its implied predictive distribution, and therefore is more sensitive to the inflated noise in the trivially misspecified setting.
Figure
\ref{fig:linear-regression-pred-diff} translates the impact of parameter uncertainty into predictive performance, and shows an even clearer  advantage of the resulting \ac{pro} posterior predictive distributions.

\subsection{A  Taxonomy of Model Misspecification}

Based on the previous examples, it will prove useful to rigorously define this so far intuitively motivated taxonomy of model misspecification.
As we shall see later, doing so will prove   critical for comparing the behaviours of  Gibbs and \ac{pro} posteriors.
\begin{definition}
    \label{def:misspecification}
    We say that $\mathcal{M}_{\Theta}$ is \textit{well-specified} if $P_0 \in \mathcal{M}_{\Theta}$, and \textit{misspecified} if $P_0 \notin \mathcal{M}_{\Theta}$.
    Further, we call a misspecified $\mathcal{M}_{\Theta}$ \textit{trivially misspecified}  with respect to a scoring rule $\mathcal{S}$ if there exists $\theta' \in \Theta$ so that
    $\inf_{Q\in\mathcal{P}(\Theta)}\mathcal{S}(P_Q, P_0) \geq  \mathcal{S}(P_{\theta'}, P_0)$, and \textit{non-trivially misspecified} if $\mathcal{S}(P_Q, P_0) < \inf_{\theta \in \Theta}\mathcal{S}(P_{\theta}, P_0)$ {for some $Q \in \mathcal{P}(\Theta)$.}
    Further, we say that a misspecified $\mathcal{M}_{\Theta}$ allows \textit{convex recovery} if there exists $Q \in \mathcal{P}(\Theta)$ such that $P_Q = P_0$.
\end{definition}

To aid understanding, Figure \ref{Fig:misspec-fig} provides an illustration of these different forms of model misspecification.
What we call non-trivial  misspecification is a generalisation of what
\citet{grunwald2017inconsistency} termed ``{bad misspecification}''.
Convex recovery is a special case of this, and implies that the extended model class $\{P_Q: Q\in\mathcal{P}(\Theta)\}$ obtained via convex combinations of $\mathcal{M}_{\Theta}$ contains $P_0$.

\begin{figure}[H]%
\begin{center}
\centering
    \begin{subfigure}{0.32\textwidth}
        \centering
        \resizebox{\linewidth}{!}{\begin{tikzpicture}[x=1pt,y=1pt,scale=0.6]
  \coordinate (C) at (272.5,154);

  \fill[gray!15]
    (C) -- 
    plot[smooth cycle, tension=1.2] coordinates {
      (300,90)
      (370,120)
      (330,220)
      (220,230)
      (166,154)
    };

  \draw[thick]
    (C) -- 
    plot[smooth cycle, tension=1.2] coordinates {
      (300,90)
      (370,120)
      (330,220)
      (220,230)
      (166,154)
    };

  \draw[black, thick] (140, 80) -- (400, 80);
  \draw[black, thick] (400, 80) -- (400, 270);
  \draw[black, thick] (140, 270) -- (400, 270);
  \draw[black, thick] (140, 270) -- (140, 80);

  \draw[dashed, gibbscolour, thick]
    (195, 93) -- (209, 112);

  \node at (360,230) {{\Large $\mathcal{M}_\Theta$}};
  \filldraw[gibbscolour] (209, 112) circle (3pt) node[above=0.4cm, right=-0.5cm] {$P_{\theta^\star} = \arg\min_{\theta}\mathcal{S}(P_{\theta}, P_0)$};
  
  \filldraw (195, 93) circle (3pt) node[left] {$P_0$};

\end{tikzpicture}}
        \caption{Trivial misspecification.}
    \end{subfigure}
    \hfill
    \begin{subfigure}{0.32\textwidth}
        \centering
        \resizebox{\linewidth}{!}{\begin{tikzpicture}[x=1pt,y=1pt, scale=0.6]
  \coordinate (C) at (272.5,154);

  \fill[gray!20]
    (C) -- 
    plot[smooth cycle, tension=1.2] coordinates {
      (272.5,154)
      (300,100)
      (360,120)
      (330,220)
      (220,230)
      (166,154)
    };

  \draw[thick]
    (C) -- 
    plot[smooth cycle, tension=1.2] coordinates {
      (272.5,154)
      (300,100)
      (360,120)
      (330,220)
      (220,230)
      (166,154)
    };

  \draw[dashed, procolour, thick]
    (175,147) -- (300,100);
  \draw[dashed, gibbscolour, thick]
    (237.5 - 13, 123.5 - 30) -- (270,155);

  \draw[black, thick] (140, 80) -- (400, 80);
  \draw[black, thick] (400, 80) -- (400, 270);
  \draw[black, thick] (140, 270) -- (400, 270);
  \draw[black, thick] (140, 270) -- (140, 80);

  \node at (360,230) {{\Large $\mathcal{M}_\Theta$}};
  \filldraw[black] (300,100) circle (3pt) node[anchor=south west]{$P_{\theta_2}$};
  \filldraw[black] (175,147) circle (3pt) node[anchor=south west] {$P_{\theta_1}$};
  \filldraw[procolour] (245.5, 121.5) circle (3pt) node[above=0.3cm, left=-0.3cm]{$P_{Q^\star}$};
  \filldraw[gibbscolour] (270,155) circle (3pt) node[anchor=south west]{$P_{\theta^\star}$};
  
  \filldraw (237.5 - 13, 123.5 - 30) circle (3pt) node[right] {$P_0$};

\end{tikzpicture}}
        \caption{Non-trivial misspecification.}
    \end{subfigure}
    \hfill
    \begin{subfigure}{0.32\textwidth}
        \centering
        \resizebox{\linewidth}{!}{\begin{tikzpicture}[x=1pt,y=1pt, scale=0.6]
  \coordinate (C) at (272.5,154);

  \fill[gray!20]
    (C) -- 
    plot[smooth cycle, tension=1.2] coordinates {
      (272.5,154)
      (300,100)
      (360,120)
      (330,220)
      (220,230)
      (166,154)
    };

  \draw[thick]
    (C) -- 
    plot[smooth cycle, tension=1.2] coordinates {
      (272.5,154)
      (300,100)
      (360,120)
      (330,220)
      (220,230)
      (166,154)
    };

  \draw[dashed, procolour, thick]
    (175,147) -- (300,100);
  \draw[dashed, gibbscolour, thick]
    (245, 122) -- (270,155);

  \draw[black, thick] (140, 80) -- (400, 80);
  \draw[black, thick] (400, 80) -- (400, 270);
  \draw[black, thick] (140, 270) -- (400, 270);
  \draw[black, thick] (140, 270) -- (140, 80);

  \node at (360,230) {{\Large $\mathcal{M}_\Theta$}};
  \filldraw[black] (300,100) circle (3pt) node[anchor=south west]{$P_{\theta_2}$};
  \filldraw[black] (175,147) circle (3pt) node[anchor=south west] {$P_{\theta_1}$};
  \filldraw[black] (245.5, 121.5) circle (3pt) node[below left]{$\color{black}{P_0} = \color{procolour}{P_{Q^\star}}$};
  \filldraw[gibbscolour] (270,155) circle (3pt) node[anchor=south west]{$P_{\theta^\star}$};
  
\end{tikzpicture}}
        \caption{Convex recovery.}
    \end{subfigure}
\caption{
Illustration of key ideas introduced in Definition \ref{def:misspecification}.
Panel (a) depicts trivial  model misspecification: relative to the chosen score $\mathcal{S}$, a singular element of $\mathcal{M}_{\Theta}$ provides the best fit for $P_0$, and there is no benefit to averaging.
In contrast, Panel (b) illustrates the case of non-trivial model misspecification: constructing ${Q^\star} := \argmin_{Q}\mathcal{S}(P_Q, P_0)$ as a convex combination of $P_{\theta_1}$ and $P_{\theta_2}$ yields a better predictive for $P_0$ than any one point in $\mathcal{M}_{\Theta}$.
Panel (c) provides pictorial intuition for  convex recoverability: it is the special case where $P_{Q^{\star}}$ and $P_0$ coincide.}
\label{Fig:misspec-fig}
\end{center}
\end{figure}

\section{Formal Guarantees}\label{sec:theory-short}

The motivating example in  the previous section
shows empirically that \ac{pro} posteriors  \textit{adapt} to model misspecification: when the model is well-specified, they behave like Gibbs or Bayes posteriors, but deliver fundamentally different behaviours when the model is misspecified. In this section, we clarify this adaptivity and formally show that \ac{pro} posteriors are predictively superior to Bayes and Gibbs posteriors. 
Throughout, we will  often use the shorthand \textit{predictive scoring rule} to refer to a scoring rule $S(P_Q, x)$  applied directly to score a predictive distribution $P_Q = \int P_{\theta} \dt Q(\theta)$. 
Additionally, for two sequences $a_n,b_n$ we write $a_n\lesssim b_n$ to mean that $a_n\le C b_n$ for all $n$ large enough and some universal constant $C>0$. Further, we say that $a_n\asymp b_n$ if $a_n\lesssim b_n$ and $b_n\lesssim a_n$.

\subsection{Regularity conditions}

Throughout, we impose a minimal convexity assumption satisfied by virtually all scoring rules of interest---including kernel scoring rules and the logarithmic score.
\begin{assumption}\label{ass:convex} The scoring rule $S(\cdot,x)$ is convex in its first argument for all $x\in\mathcal{X}$, and its corresponding divergence $\mathcal{D}_S$ is jointly convex in both arguments. 
\end{assumption}

In addition to convexity of the scoring rule, we require regularity conditions on their expectations. 
While less restrictive conditions can be obtained, the following  regularity condition is easy to state, and suffices for our theoretical analysis.

\begin{assumption}\label{ass:entropy}
Let $r_n$ be a positive sequence such that $r_n\rightarrow 0$ as $n\rightarrow\infty$ and $\mathcal{P}_{\Pi}(\Theta):= \{Q\in\mathcal{P}(\Theta):Q\ll\Pi\}$, it holds that
$ \E\left[\sup_{Q\in\mathcal{P}_{\Pi}(\Theta)}\left|\frac{1}{n}\sum_{i=1}^{n}S(P_Q,X_i)-\mathcal{S}(P_Q,P_0)\right|\right]\lesssim r_n$.
\end{assumption}

As the speed at which the sequence $r_n$ decays depends on the scoring rule, we state it in full generality.
For the special case of bounded scoring rules, $r_n$ will always be of order $\log(n)/{n}^{1/2}$.
Further, even for unbounded scoring rules such as the logarithmic scoring rule, the rate  $\log(n)/{n}^{1/2}$ is achievable at the expense of additional restrictions and technicalities using techniques such as those in 
\cite{ghosal2001entropies}. 
For further discussion and details, see Section \ref{sec:theoretical-details}.

Though Assumption \ref{ass:entropy} is satisfied by many scoring rules, more interpretable alternatives exist.
In particular, Assumption \ref{ass:convex} and Jensen's inequality imply that 
\begin{IEEEeqnarray}{rCl}
    S\left( P_Q, x\right)
    & = &
    \int S\left(  P_{\theta}, x \right) \dt Q(\theta)
    - \Delta(Q, x),
    \label{eq:predictive-score-decomposition}
\end{IEEEeqnarray}
for  all $x \in \mathcal{X}$ and some non-negative remainder term $\Delta:\mathcal{P}(\Theta) \times \mathcal{X} \to \mathbb{R}_{+}$.
In many cases, such as the log predictive score, 
this remainder does not simplify further. However, kernel scoring rules (\citealp{gneiting2007probabilistic}) 
can be written as an expectation over a product measure in $\Theta$ for any $x\in\mathcal{X}$ (see Lemma \ref{lemma:kernel_scores} in Appendix \ref{app:lemmas}), so that the predictive score is
\begin{IEEEeqnarray}{rCl}
S\left(P_Q,x\right)
& = &
\int \underbrace{\left\{S(P_{\theta_1},x_i) 
- 
\delta(x_i;\theta_1,\dots,\theta_k)\right\}}_{ \textstyle
=:L(\theta_{1:k}, x)}
\dt Q^k(\theta_{1:k}); \quad Q^{k}(\theta_{1:k})=\bigotimes_{j=1}^{k}Q(\theta_j).
\label{eq:tractable_score_defn}
\end{IEEEeqnarray}
for some function $\delta: \mathcal{X} \times \Theta^k \to \mathbb{R}_+$ with  $k \in \mathbb{N}_+$.
This  representation allows us to replace Assumption \ref{ass:entropy} with more interpretable and testable conditions. 

\begin{assumptionp}{\ref*{ass:entropy}$'$}
\label{ass:Global}
{The predictive scoring rule satisfies \eqref{eq:tractable_score_defn} and for any $\varphi \in \Theta$, if  $\theta_{1:k} = (\varphi, \varphi, \dots \varphi)$, then $\delta(x;\varphi,\dots,\varphi) = 0$, for all $x \in \mathcal{X}$.}
Further, for any $\lambda\in\mathbb{R}$, $\mathcal{L}(\theta_{1:k}):=\E_{}[L(\theta_{1:k}, X)]$,  and some $C>0$, 
			$
			\bigintssss \mathbb{E}\left[e^{\lambda \{\mathcal{L}(\theta_{1:k})-\frac{1}{n}\sum_{i=1}^nL(\theta_{1:k}, X_i)\}}\right]\dt \Pi^k(\theta_{1:k})\le e^{\lambda^2C^2/n}.
		$
\end{assumptionp}

Here, the first part of Assumption \ref{ass:Global} is a mild regularity condition, and  automatically satisfied by kernel scoring rules for $k=2$. 
The second part imposes a Hoeffding-type moment condition that generally has to be verified, and which is very similar to those classically arising in the study of Gibbs posteriors. 
Though easier to verify, certain scores such as the logarithmic scoring rule violate Assumption \ref{ass:Global}, so that the more abstract conditions imposed by Assumption \ref{ass:entropy} must be maintained instead.

\subsection{Key Results}\label{sec:key-results}

Having established the necessary regularity conditions, we proceed by studying the impact of different misspecification regimes on \ac{pro} posteriors $Q_n$, and compare them to Gibbs posteriors $Q_n^{\dagger}$ constructed using the same scoring rule. 
To state our results, we define 
\begin{equation}\label{eq:rates}
\nu_n=\begin{cases}
    \log(n)/{n}^{1/2},&\text{ if Assumption \ref{ass:Global} is satisfied;}\\
    \max\{\log(n)/{n}^{1/2},r_n\},& \text{ if Assumption \ref{ass:entropy} is satisfied,}
\end{cases}
\end{equation}
and take $\lambda_n \asymp {n}^{1/2}$ up to logarithmic factors.
This allows us to state our first key result on the predictive performance of \ac{pro} posteriors.
\begin{theorem}\label{thm:DI-master}
Assumption \ref{ass:convex} is satisfied, and Assumption  \ref{ass:entropy} or \ref{ass:Global} holds. Under prior regularity conditions, the following holds for $n$ sufficiently large:
\begin{itemize}
    \myitem{$\operatorname{(WS)}$} 
    \label{item:master-theorem-DI-WS}
If the model is well-specified, $\E \left[\mathcal{D}_S\left( P_{Q_n},P_0\right)\right]
    - \E [\mathcal{D}_S(P_{ Q_n^\dagger},P_0)] \lesssim \nu_n %
    $;
 \myitem{$\operatorname{(NT)}$} 
    \label{item:master-theorem-DI-non-trivially-NT}
    If the model is non-trivially misspecified, $\E \left[\mathcal{D}_S\left( P_{Q_n},P_0\right)\right]
    <\E [\mathcal{D}_S(P_{ Q_n^\dagger},P_0)]$, 
    \\
    and %
    $\E \left[\mathcal{D}_S\left( P_{Q_n},P_0\right)\right]
    - \inf_{Q \in \mathcal{P}(\Theta)} \mathcal{D}_{\mathcal{S}}(P_Q, P_0) \lesssim \nu_n;$%
    \myitem{$\operatorname{(CR)}$} 
    \label{item:master-theorem-DI-non-trivially-CR}  
    If the model allows convex recovery, then $\E \left[\mathcal{D}_S\left( P_{Q_n},P_0\right)\right]\lesssim \nu_n.$ %
\end{itemize}
\end{theorem}

Theorem \ref{thm:DI-master} evaluates the \ac{pro} posterior predictive $P_{Q_n}$ and its relationship to the Gibbs posterior predictive $P_{Q_n^{\dagger}}$ in terms of the divergence $\mathcal{D}_{S}$ induced by the score chosen to construct both posteriors. %
\ref{item:master-theorem-DI-WS} concerns the well-specified setting. 
Here, Gibbs posteriors typically concentrate towards the correct model parameter at the usual $\log(n) / n^{1/2}$ rate. 
As the result shows, the difference in predictive performance between  \ac{pro} posteriors and Gibbs posteriors in this setting is small, and usually of order $\log(n) / n^{1/2}$.
For the case of non-trivial misspecification studied in \ref{item:master-theorem-DI-non-trivially-NT}, the situation changes: 
now, the \ac{pro} posteriors strictly outperforms Bayes and Gibbs posteriors.
In particular, this performance gap will be of order $\inf_{\theta \in \Theta}\mathcal{D}_S(P_{\theta}, P_0)-\inf_{Q \in \mathcal{P}(\Theta)}\mathcal{D}_S(P_{Q}, P_0)$, which grows larger as the benefits of averaging over $\mathcal{M}_{\Theta}$ become more pronounced.
To confirm this intuition, the result in \ref{item:master-theorem-DI-non-trivially-CR} considers the special case of convex recovery, where $\inf_{Q \in \mathcal{P}(\Theta)}\mathcal{D}_S(P_{Q}, P_0) = 0$.
As expected, the  \ac{pro} posterior recovers the true data-generating process at rate $\nu_n$.
In contrast, the Gibbs posterior predictive $P_{Q_n^{\dagger}}$  \textit{never} recovers the true data generating process in this setting, and  instead collapses onto a point mass at $\theta^{\star} = \argmin_{\theta \in \Theta}\mathcal{D}_S(P_{\theta}, P_0)$.
Lastly, the result also shows that the only case in which the choice of scoring rule has a notable impact is the non-trivially misspecified case: when the model satisfies convex recovery or is well-specified, \textit{any}  scoring rule will recover the same predictive distribution for $n$ sufficiently large.
This suggests that \ac{pro} posteriors are generally less sensitive to the choice of scoring rule than Gibbs posteriors. 

The above result validates the use of \ac{pro} posteriors whenever predictive performance matters, but sheds little light on the qualitative behaviour of \ac{pro} posteriors as expressions of parameter uncertainty.
Given that they no longer collapse to a point mass, is it safe to assume that they behave well as the sample size increases? 
Indeed, in the context of parameter inference, convergence to a limiting distribution as more data are observed is fundamental to ensure that one is learning a well-defined target.
Our second result assuages these concerns, and shows that \ac{pro} posteriors \textit{do} indeed converge to a well-defined target given by the predictively optimal model average $Q^* := \argmin_{Q\in\mathcal{P}(\Theta)}\mathcal{D}_S(P_Q,P_0)$.
Below, we state this convergence  relative to a distance $\dt$ defined on  probability distributions and dependent on additional regularity conditions (see Section \ref{sec:theoretical-details} for details).
\begin{corollary}\label{corollary:fake-concentration-dtf}
Assumption \ref{ass:convex}  and Assumption \ref{ass:entropy} or \ref{ass:Global}  holds. 
Under convex recovery and  additional regularity conditions, there exists $\alpha > 0$ so that
$
\mathbb{E}[\dt(Q_n,Q^\star)] \leq C(\log n/n^{1/2})^{1/\alpha}.
$ 
\end{corollary}

While the above result relies on convex recovery, our empirical results strongly suggest that this is an artifact of our proof technique, and that this behaviour should be expected to hold under much weaker conditions.

\subsection{Predictive calibration}
\label{sec:calibration}

Beyond achieving predictive optimality, a natural question is whether \ac{pro} posterior predictive distributions are calibrated. The literature entertains various  modes of calibration but in this section we only discuss two variants: a weak probabilistic notion and a stronger almost sure notion. 
Under mild conditions, this weaker probabilistic notion holds automatically whenever the model class allows for convex recovery.
\begin{corollary}
    \label{corollary:pointwise-coverage}
    Suppose that the conditions of Theorem \ref{thm:DI-master} hold, that the model allows for convex recovery, and that $\mathcal{D}_{S}$ metrises weak convergence.
    For $C \subset \mathcal{X}$ any fixed and  measurable $P_0$-continuity set,  $|P_{Q_n}(C) - P_0(C)|$ converges to zero in probability as $n\to\infty$.
\end{corollary}

The above demonstrates that so long as an average over models can correctly capture the data-generating process, \ac{pro} posteriors produce {predictions that are probabilistically and marginally calibrated in the sense of \citet{gneiting2007probabilistic}.} %
The requirement that the chosen scoring rule metrise weak convergence is exceedingly mild.
For instance, it is satisfied by the MMD based on characteristic kernel functions such as Gaussian or Mat\'ern kernels, the logarithmic score, the Cumulative Ranked Probability Score, as well as the energy score.
On top of pointwise coverage guarantees, the logarithmic scoring rule leads to {a stronger calibration guarantee beyond those given by probabilistic and marginal calibration.} %
\begin{theorem}\label{thm:calibration}

{For $S(P,x)=-\log \dt P(x)$ denoting the logarithmic scoring rule, assume} that  the conditions of Theorem \ref{thm:DI-master} hold, and that the model allows for convex recovery.
Then, letting $(C_n)_{n\geq1}$ be any sequence of potentially data-dependent measurable sets $C_n \subset \mathcal{X}$ so
that $P_{Q_n}(C_n)=1-\alpha$
almost surely, it holds that
$\E[P_0(C_n)] \longrightarrow  1-\alpha$ as $n\to\infty$.
\end{theorem}

The results in Corollary \ref{corollary:pointwise-coverage} and Theorem \ref{thm:calibration} are a consequence of the \ac{pro} posterior's emphasis on prediction and model misspecification. Bayes posteriors satisfy frequentist calibration guarantees only if the model is well-specified, but generally fail to deliver predictively calibrated inferences under model misspecification. In contrast, since \ac{pro} posteriors construct uncertainty through predictive performance, they inherit asymptotically valid predictive coverage so long as convex recovery of $P_0$ is feasible.

Should the model not be convexly recoverable,  the  \ac{pro} posteriors predictive will generally not be calibrated.
Indeed, as Theorem \ref{thm:DI-master} shows, there will be an irreducible and persistent difference between $P_{Q_n}$ and $P_0$ outside this setting, even as $n\to\infty$.
Practically speaking however, our experiments demonstrate that it is reasonable to expect the \ac{pro} posterior predictive distribution  to provide approximately correct coverage whenever the irreducible difference between $P_{Q_n}$ and $P_0$ is small (see e.g., Section \ref{sec:redshift}).

 \section{Comparisons}

It is worth contrasting \ac{pro} posterior uncertainty  with a set of ideas and methods known under the umbrella of \textit{predictive Bayes} \cite[see e.g.][]{berti21class, fortini20quasibayes, fortini23predictionbased, fong2023martingale, fortini25exchangeability}.
Despite similar sounding names, the roles played by prediction in the construction of \ac{pro} posteriors and predictive Bayes approaches are fundamentally different:
the \ac{pro} posterior is geared towards \textit{learning the predictive} from {seen data}, predictive Bayes approaches instead \textit{posits predictive distributions}  to model the {unseen data}. 
Put differently, while \ac{pro} posteriors make the predictive distribution the object of {posterior inference}, predictive Bayes methods specify predictive distributions instead of  priors and likelihoods.

Conversely,  priors and likelihoods can be recast as a convenient way of  specifying a sequence of predictive distributions on unseen data:
formalising this in a recent read paper at the Royal Statistical Society, \citet{fong2023martingale} used Doob's Theorem to argue that the  source of Bayesian uncertainty is unseen data.
{Taken to its logical extreme, if we observe the entire population $P_0$, no model or parameter uncertainty remains, and Bayesian methods become perfectly confident in the conclusions they draw from data.}
Indeed, this is the phenomenon underlying posterior concentration, which in turn also applies to various generalisations of Bayesian uncertainty, including Gibbs and Martingale posteriors.

While posterior concentration is clearly desirable if there exists some value of $\theta\in\Theta$ that can make accurate predictions about the true data-generating process $P_0$, most scientists would argue that posterior concentration is no longer a faithful expression of how model uncertainty interacts with the data volume when none of the considered model parameterisations can fully  explain the natural phenomenon under study.
\ac{pro} uncertainty remedies  this: when the model is misspecified, \ac{pro} posteriors generally {do not concentrate} onto a single parameterisation as more data are observed.
Instead, they account for the  irreducible uncertainty that arises as a consequence of model misspecification, and stabilise towards a non-trivial and predictively optimal average over models as more observations are collected.
This is aligned with how most scientists conceive of model uncertainty: as the model fails to make accurate predictions, one is uncertain about its ability to represent the data-generating process---even if one were to collect further observations.
Conversely, if the model should be well-specified  so that $P_0 \in \mathcal{M}_{\Theta}$, \ac{pro} uncertainty behaves similarly to Bayes and Gibbs posteriors: as accurate predictions are possible using a single model parameterisation, \ac{pro} posteriors concentrate around the corresponding parameter value as more data is observed.

\subsection{Setting for Numerical Examples}
\label{sec:golf}

To round off the above discussion with numerical illustrations, we  contrast \ac{pro} posteriors with several other related approaches.
We do so using the golf putting data of \citet{gelman_probability_2002}, which compares the probability of a successful put $(x \in \{0, 1\})$ by professional golfers to the distance ($z \in \mathbb{R}$) from which the shot was made for $n = 5,988$ shots.
Throughout, we fit this data using the logistic regression model 
\begin{IEEEeqnarray}{rCl}
P^{\mathsf{LRM}}_{\theta}(x \mid z) & = &
\sigma(\theta^\top  z)^{x}\left\{1-\sigma(\theta^\top z)\right\}^{1-x} \text{ for } \;
\sigma(\theta^\top z) = 
\left\{1 +  \exp(\theta_1 + \theta_2z ) \right\}^{-1},
\label{eq:logistic-regression}
\end{IEEEeqnarray}
and chart uncertainty estimates across several methods that derive from this set-up.

\begin{figure}[ht!]
\centering
\includegraphics[width=\linewidth]{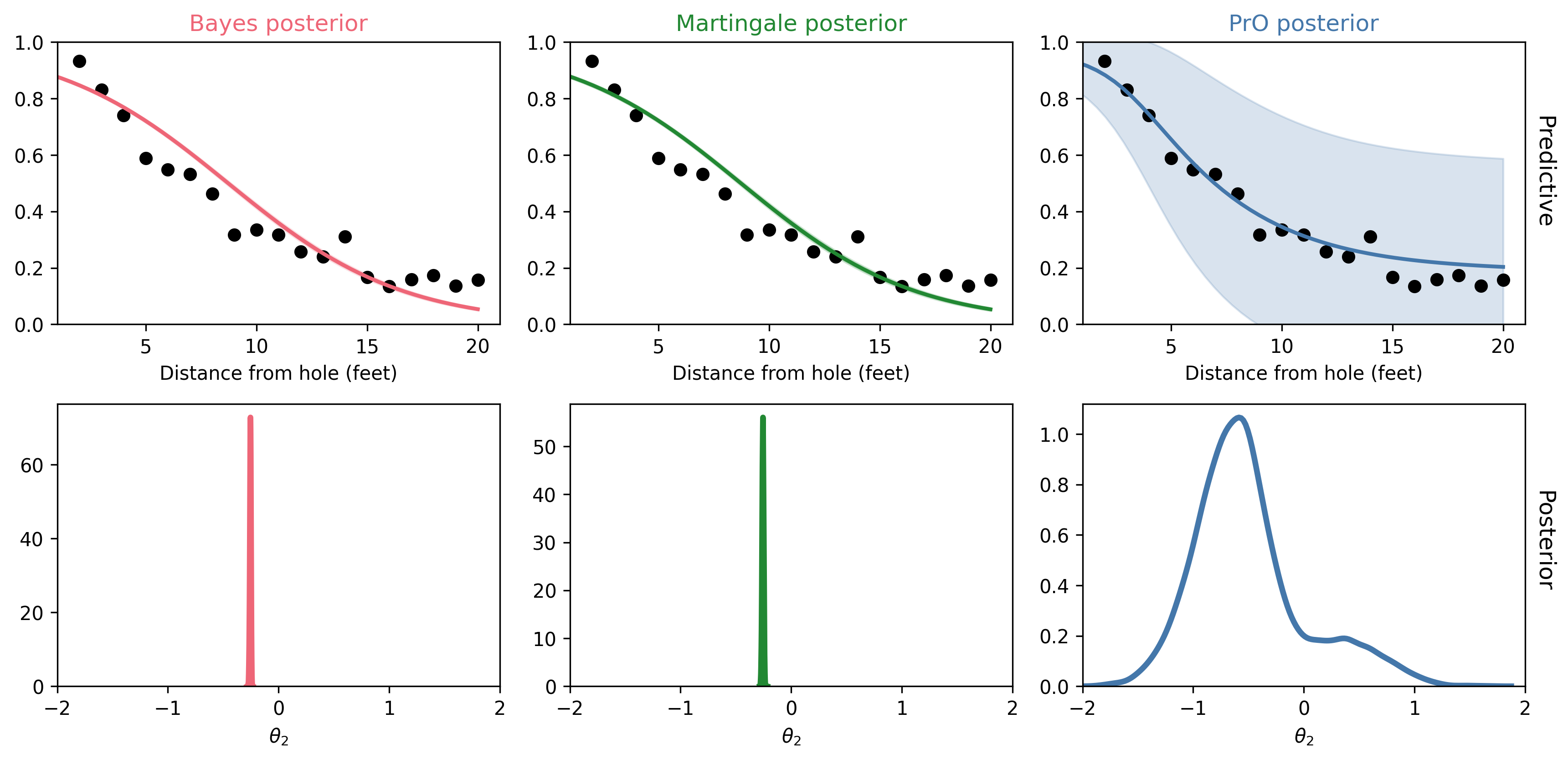}
\caption{
\textit{Golf putting data}\/. Points represent the proportion of successful putts made by professional golfers as a function of the distance from the hole.
The first row shows the predictive mean and one standard deviation intervals of the \textcolor{gibbscolour}{\textbf{Bayes}}, \textcolor{mgpcolour}{\textbf{Martingale}}, and \textcolor{procolour}{\textbf{\ac{pro}}} posteriors. 
The second row shows the marginals of the corresponding parameter posteriors.
}
\label{fig:golf}
\end{figure}

\subsection{Bayes Posterior and Martingale Posterior Uncertainty}
First, we compare \ac{pro} posterior uncertainty with classical Bayes posteriors and a Martingale posterior \citep{fong2023martingale} based on the Bayesian bootstrap and the negative log likelihood loss \citep[see][]{lyddon2019general, rubin1981bayesian}, and display the results in Figure \ref{fig:golf}.
As we can see, these Bayes  and Martingale posteriors  concentrate around a single parameter value. 
Their predictive distributions mirror this, and are qualitatively similar to what one would obtain by using the maximum likelihood estimator to construct a plug-in predictive distribution \citep[see][]{mclatchie_predictive_2025}.
This  sharply contrasts with  the \ac{pro} posterior: its predictive distribution indicates that while the success of short puts can be estimated with some fidelity, the success of longer puts is associated with substantive levels of uncertainty.
One possible explanation of this phenomenon is that the logistic regression model is too coarse, and that there are in fact two sub-groups of professional golfers. 
This hypothesis seems to be backed up by the \ac{pro} posterior on $\theta_2$, a parameter modelling the effect  of distance on the probability of a successful put: unlike the Bayes or  Martingale posterior, the marginal \ac{pro} posterior of $\theta_2$ has two modes which could account for these two sub-groups.
To investigate this more carefully, we also compute the posterior marginal effect of distance on the probability of a successful put, which for a posterior $Q$ is given by  
\begin{IEEEeqnarray}{rCl}
    \int \frac{\partial P^\mathsf{LRM}_\theta(x \mid z)}{\partial z}\dt Q(\theta) & = &\int \theta_2 \cdot 
    P^\mathsf{LRM}_\theta(x=1 \mid z)
    P^\mathsf{LRM}_\theta(x=0 \mid z)
    \dt Q(\theta),
    \nonumber
\end{IEEEeqnarray}
and we plot this effect as a function of the distance $z$ in Figure \ref{fig:golf-marginal-effects}.
Inferences across the Bayes, Martingale, and \ac{pro} posteriors  agrees on the qualitative nature of the marginal effect of distance:
while the marginal effect of distance is negative and decreasing until around 5--7 feet, it increases thereafter and monotonically in the range of 7--20 feet, and rises close to a zero marginal effect.

\begin{figure}[t!]
\centering
\includegraphics[width=\linewidth]{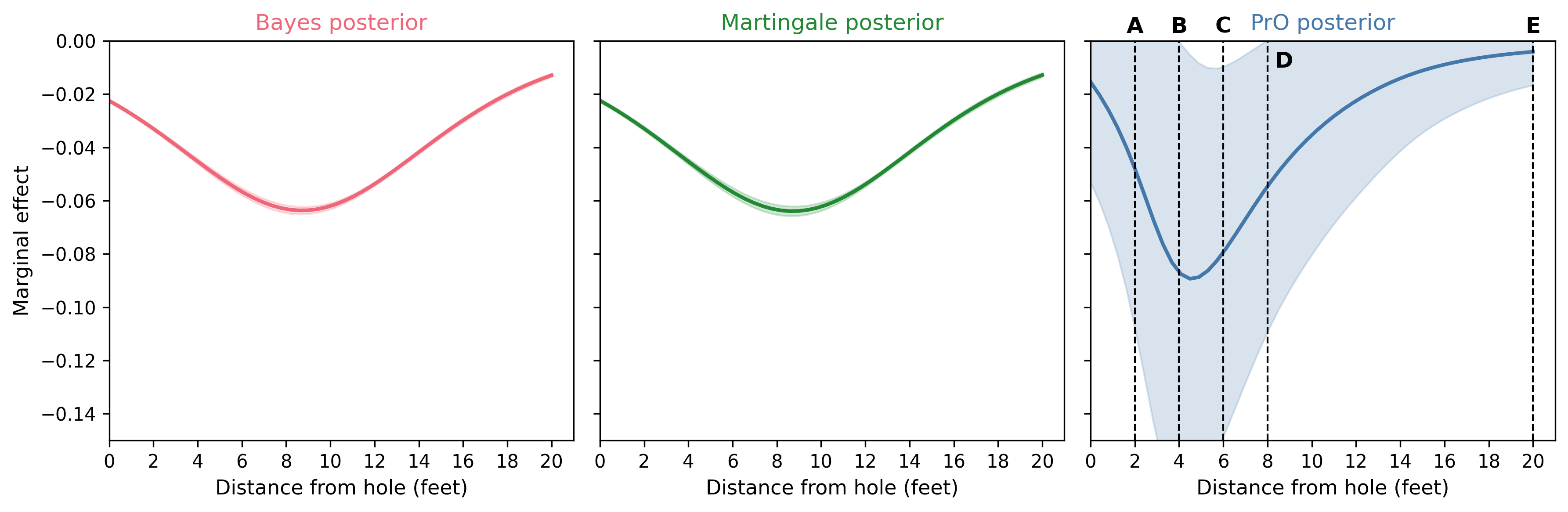}
\includegraphics[width=\linewidth]{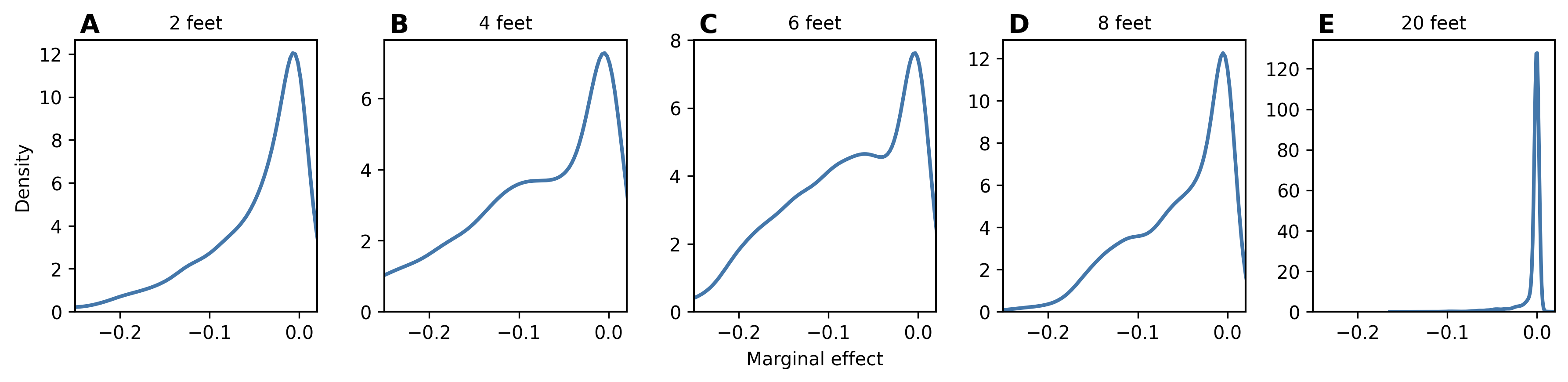}
\vspace*{-3em}
\caption{%
\textit{Top row}\/:
the mean (solid line) and one standard error intervals (shaded area) of the posterior marginal effects of the \textcolor{gibbscolour}{\textbf{Bayes}}, \textcolor{mgpcolour}{\textbf{Martingale}}, and \textcolor{procolour}{\textbf{\ac{pro}}} posteriors as a continuous function of distance in the golf putting data. 
\textit{Bottom row}\/:
we investigate the posterior marginal effect under \textcolor{procolour}{\textbf{\ac{pro}}} posterior at $2,4,6,8$, and $20$ feet.}
\label{fig:golf-marginal-effects}
\end{figure}

While all posteriors agree on the shape of this curve, the accompanying uncertainty differs substantively.
Unlike for the Bayes and Martingale posteriors, the marginal effects  obtained from the \ac{pro} posterior at each distance are rich multimodal distributions, and allow a more fine-grained model criticism.
More precisely, they show that while posterior uncertainty of the marginal effect is unimodal for very short and very long distances, it bifurcates into a bimodal distribution for  middle distances, which once again reaffirms our previous observations:  there seem to be two relevant sub-groups in the data, and middle distances seem the most informative at discriminating between them.

We note that this ease at which \ac{pro} posteriors can be used for model criticism
has led previous authors to speculate that one may be able to use related methods as a diagnostic tool for model misspecification \citep[see][]{lai2024predictive,chazal_computable_2025}.
While this has previously been argued on the basis of heuristics, our findings in Section \ref{sec:theory-short} provide the  statistical tools required for making these ideas more rigorous.

\subsection{Bayesian Mixture Models}
The above example raises an immediate question: given that the \ac{pro} posteriors seems to indicate the presence of a group structure in the model, should we think of \ac{pro} posteriors as mixture models?
In fact, since the \ac{pro} posterior objective \eqref{eq:pro-posterior-def} directly targets the marginalised predictive $\int P_{\theta} \mathsf{d}Q(\theta)$---an object which we might call a mixture---how are the aims of \ac{pro} posterior inference different from those of mixture modelling?

To answer this, consider the Bayesian mixture model given by
\begin{IEEEeqnarray}{rCl}
    P_{\theta^{(1:M)}, \omega^{(1:M)}}(x \mid z)
    & = &
    \sum_{m=1}^M
    \omega^{(m)} 
    \cdot 
    P^{\mathsf{LRM}}_{\theta^{(m)}}(x \mid z),
    \nonumber
\end{IEEEeqnarray}
where the vector $\omega^{(1:M)}$ is non-negative and sums to one. 
One of the key practical challenges in Bayesian mixture modelling consists in choosing priors on $M$ and the mixture weights $\omega^{(1:M)}$.
This is a key difference: \ac{pro} posteriors make no attempt at \textit{modelling}  mixtures over $P_{\theta}^{\textsf{LRM}}$, and thus require no such priors.
Instead, their priors are formulated with respect to the typically more interpretable parameters of the base model $P_{\theta}^{\textsf{LRM}}$.
Thus, while the \ac{pro} posterior $Q_n$ is a predictively oriented mixture over $P_{\theta}^{\textsf{LRM}}$, it is not a mixture \textit{model}, since the mixing measure itself requires neither a statistical model nor a prior.

In contrast, mixture models  explicitly \textit{model} the mixture over $P_{\theta}^{\textsf{LRM}}$.
This changes the parameter of interest: 
relative to the simple logistic model in \eqref{eq:logistic-regression}, a mixture model with $M$ components inflates the number of parameters by a factor of $M$.
Though the resulting mixture $P_{\theta^{(1:M)},\, \omega^{(1:M)}}$ is more flexible than $P_{\theta}^{\textsf{LRM}}$, it is also less parsimonious, and more difficult to interpret. 
In fact, mixture models are usually used as density estimators, with parameters in each component being synthetic artifacts without inherent meaning, so that 
 $P_{\theta^{(1:M)},\, \omega^{(1:M)}} = P_{\theta^{\rho(1:M)},\, \omega^{\rho(1:M)}} $ for any permutation $\rho$.

The above reveals several important practical differences between \ac{pro} posteriors and Bayesian mixture models: 
the latter abide by  standard Bayesian inference, which comes at the price of having to model the mixture and changing the parameter of interest.
In contrast, \ac{pro} posteriors  continue to draw inferences with respect to a simpler and more interpretable parameter.
This allows for conclusions,  summaries, and model criticisms that lack clear counterparts in mixture models.
For instance, since
\begin{IEEEeqnarray}{rCl}
    \log\left\{ 
        \frac{P_{\theta}^{\mathsf{LRM}}(x=1\mid z)}{
        P_{\theta}^{\mathsf{LRM}}(x=0\mid z)
        }
    \right\}
    & = & \theta_1 + \theta_2 z,
    \nonumber
\end{IEEEeqnarray}
the parameter $\theta_2$ is directly interpretable as the linear effect of distance on the log-odds of a successful put.
The \ac{pro} posterior uncertainty which indicates two sub-groups of golfers will therefore map directly onto the log-odds of the inferred model.
We illustrate this phenomenon in Figure \ref{fig:golf-log-odds}, which shows that for larger put distances, the skill differences between these two groups become more pronounced (as measured in log odds).
Mixture models would not admit this visualisation: they make it impossible to express $(\theta^{(1:M)}, \omega^{(1:M)})$ in terms of their linear contributions to the log-odds of the underlying model $P_{\theta^{(1:M)},\, \omega^{(1:M)}}$.

\begin{figure}[t!]
\centering
\includegraphics[width=\linewidth]{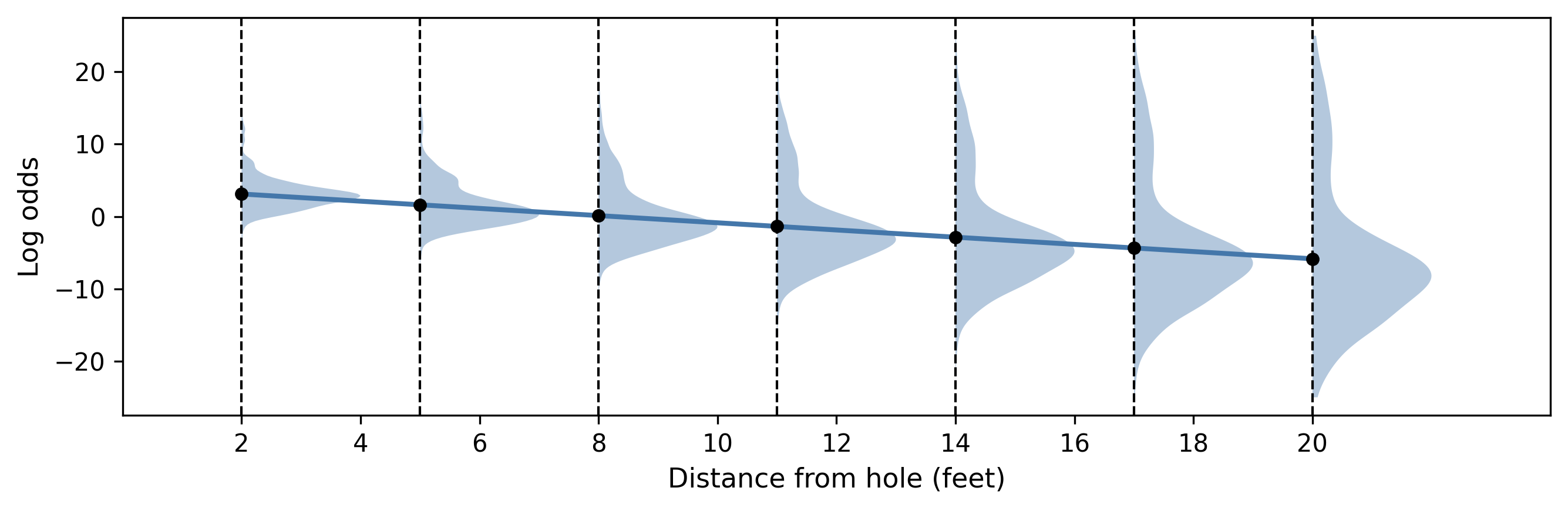}
\vspace*{-3em}
\caption{%
The distribution of \textcolor{procolour}{\textbf{\ac{pro}}} posterior log-odds for the golf putting data.
The black dots represent the posterior mean at those distances, and the solid line represents the posterior mean as a continuous function of distance.}
\label{fig:golf-log-odds}
\end{figure}

\section{Illustrations}
\label{sec:experiments}
We supplement the discussion of the previous section by illustrating the promise of \ac{pro} posteriors on three statistical problems.  
To set the scene, Section \ref{sec:computation} introduces the Mean Field Langevin algorithm used for computation.
In Section \ref{sec:binary-classification}, we highlight the implications of convex recovery in a binary classification problem.
Moving past independently sampled data, we illustrate that \ac{pro} posteriors outperform Gibbs posteriors even when the theory's assumptions are violated in Section~\ref{sec:gmrf}.
There, we investigate the impact of using different scoring rules in a graphical model for spatially dependent house prices in Boston.
Finally, in Section~\ref{sec:redshift}, we provide a practical demonstration of the theory developed in Section~\ref{sec:calibration} through a redshift prediction task from astronomy.
\subsection{Computation via Wasserstein Gradient Flows}
\label{sec:computation}

The tools for the computation of \ac{pro} posteriors are limited.
In fact, most previous work targeting related objects did not even consider the possibility of exact computation, and instead focused on approximating $Q_n$ with a parameterised variational family \citep[see][]{jankowiak2020deep, jankowiak2020parametric,masegosa2020learning,morningstar2022pacm, lai2024predictive}.
The only pre-existing work directly targeting exact computation is  \citet{shenprediction2025}, which built on the  ideas in \citet{wild2023rigorous} and the mean-field Langevin algorithm of \citet{del2013mean} to derive an asymptotically exact sampler for the special case where the scoring rule is the squared MMD.
Here, we explain how this can be generalised to any \ac{pro} posterior.

We may rewrite the objective minimised in \eqref{eq:pro-posterior-def} as
\begin{IEEEeqnarray}{rCl}
    \mathscr{L}(Q)
    & = &
    \frac{\lambda_n}{n} \sum_{i=1}^n  S\left(P_Q, x_i\right) 
    -
    \int \log \dt \Pi(\theta) \,\dt Q(\theta)
    +
    \int \log \dt Q(\theta)\, \dt Q(\theta).
    \label{eq:functional-for-WGF}
\end{IEEEeqnarray}
For functionals like this, the Wasserstein gradient flow (WGF) can be used to construct a gradient descent scheme via interacting particle systems derived from a McKean-Vlasov process.
If the Wasserstein gradient of $\frac{1}{n} \sum_{i=1}^n  S\left(P_Q, x_i\right)$  is given by $\mathcal{W}(Q)$, the Wasserstein gradient of $\mathscr{L}(Q)$ with respect to $Q$ evaluated at $\vartheta \in \Theta$ is
\begin{IEEEeqnarray}{rCl}
    \nabla_{\operatorname{W}} \mathscr{L}(Q)[\vartheta]
    & = &
    \lambda_n
       \mathcal{W}(Q)[\vartheta]
    -
    \nabla_\vartheta \log \dt \Pi(\vartheta)
    +
    \nabla_\vartheta \log \dt Q(\vartheta).
    \nonumber
\end{IEEEeqnarray}
The WGF $\{Q_t\}_{t>0}$ defined by this gradient is  the continuous-time process which evolves in the direction of steepest descent relative to the convex functional $Q \mapsto \mathscr{L}(Q)$ given by $\nabla_{\operatorname{W}}\mathscr{L}(Q)$.
As $\mathscr{L}(Q)$ is convex, and as the minimiser of $\mathscr{L}(Q)$ is the \ac{pro} posterior, the WGF acts as a continuously-valued version of gradient descent on probability measures, so that $Q_t$ converges to the \ac{pro} posterior as $t\to\infty$ under mild conditions.

To implement this gradient descent algorithm numerically, we resort to the WGF's mean-field representation, which gives rise to the stochastic differential equation (SDE) 
\begin{IEEEeqnarray}{rCl}
    \dt \vartheta_t & = &
    - \left\{
    \lambda_n \mathcal{W}(Q_t)[\vartheta_t]
        - \nabla_\vartheta \log \dt \Pi(\vartheta_t)
    \right\}
    \dt t
    +
    \sqrt{2}
    \dt B_t, \;\; \text{ for } \vartheta_t \sim Q_t,
    \nonumber
\end{IEEEeqnarray}
and where the Brownian motion $B_t$  appears as the WGF of entropy corresponds to the heat equation \citep{jordan1998variational}.
Unfortunately, the above SDE still depends on the unknown measure $Q_t$, which makes exact forward simulation impossible.
To remedy this, we use the approximations  $Q_t^{(j)}(\theta) \approx \frac{1}{p}\sum_{\ell=1, \ell\neq j}^p\updelta_{\vartheta_t^{(\ell)}}$, where $\updelta_x$ denotes the Dirac measure at $x$, and $\vartheta_t^{(j)}$ for $j=1,2, \dots p$ are a set of $p$ interacting particles. 
This finally results in a system of $p$ interacting particles whose evolution is governed by the $p$ equations
\begin{IEEEeqnarray}{rCl}
    \dt \vartheta_t^{(j)}
    & = &
    - \left\{
    \mathcal{W}(Q_t^{(j)})[\vartheta_t^{(j)}]
    - \nabla_\vartheta \log \dt \Pi(\vartheta_t^{(j)})\right\}\dt t
    +
    \sqrt{2}\dt B_t^{(j)}, \quad j=1,2\dots p,
    \nonumber
\end{IEEEeqnarray}
where $\{B_t^{(j)}:j=1,2,\dots p\}$ denotes $p$ independent Brownian motions.
To obtain a high-fidelity approximation of the \ac{pro} posterior from this particle system, one discretises it into time steps $t_1, t_2, \dots $, and evolves it for a sufficiently long time.
Afterwards, one averages over all time steps past an initial burn-in  $\tau$ and approximates the \ac{pro} posterior as
\begin{IEEEeqnarray}{rCl}
    Q_n \approx \frac{1}{|\{t_i: t_i > \tau\}|}
    \sum_{t_i > \tau}
    \wh{Q}[t_i] \quad \text{ for }\quad
    \wh{Q}[t_i]
    =
    \frac{1}{p}
    \sum_{j=1}^p
    \updelta_{\vartheta_{t_i}^{(j)}}.
    \nonumber
\end{IEEEeqnarray}
Here, the notation highlights that this approximation should be thought of as an average of $|\{t_i: t_i > \tau\}|$  measures $\wh{Q}[t]$, each of which would recover $Q_n$ exactly if there were no discretisation error and both $p\to\infty$ and $t\to\infty$
\citep[see e.g. Lemma 6 in Appendix E of][]{wild2023rigorous}.
We showcase this algorithm for the squared MMD and the logarithmic scoring rule in Appendix \ref{appendix:WGF-for-log-score-and-mmd}, and elaborate on practical considerations in Appendix \ref{appendix:WGF-practicalities-and-implementation}. 
Code to replicate the results is freely available at \url{https://github.com/yannmclatchie/pyprop}.

\subsection{Binary classification}\label{sec:binary-classification}
To strengthen intuitions about the predictive benefits conferred by \ac{pro} posteriors, we construct a  synthetic data set in which a generalised linear model with logistic link function allows for convex recovery as set out in Definition \ref{def:misspecification}. 
Labels $y_i$ are assigned deterministically as a function of the covariate vector $(x_i^{(0)}, x_i^{(1)})$ which is sampled uniformly on $[-2,2]^2$. 
When the covariate vector falls into the top-left quadrant of $[-2, 2]^2$ (so that $x_i^{(0)} < 0, x_i^{(1)} > 0$), $y_i=0$.
In contrast, the bottom-left quadrant (for which $x_i^{(0)} > 0, x_i^{(i)} < 0$) yields $y_i=1$, while labels in the remaining two quadrants are sampled as $y_i \sim \operatorname
{Ber}(0.5)$. 

We generate $n=1,000$ observations from this model and compute both the corresponding Bayes posterior and the \ac{pro} posterior based on the logarithmic scoring rule.
We compare their posterior predictive distributions in Figure~\ref{fig:binary-classification}, and conclude that Bayesian inference produces a subpar decision boundary which leads to an overconfident misinterpretation about the data-generating process in the top-right and bottom-left quadrants.
In contrast, the \ac{pro} posterior identifies the correct decision boundary, and recovers a predictive distribution $P_{Q_n}$ for which $\KL(P_0\|P_{Q_n})$ is close to zero.
This mirrors the predictions of Theorems \ref{thm:DI-master}, and is a consequence of the fact that the \ac{pro} posterior predictive allows for convex recovery of the exact decision boundary. 
Figure \ref{fig:binary-classification-asymptotic} further illustrates this by probing the effect of $n$, and confirms that the theory derived is practically relevant even for relatively small data set sizes.

\begin{figure}[t!]
    \centering
    \includegraphics[width=0.666\linewidth]{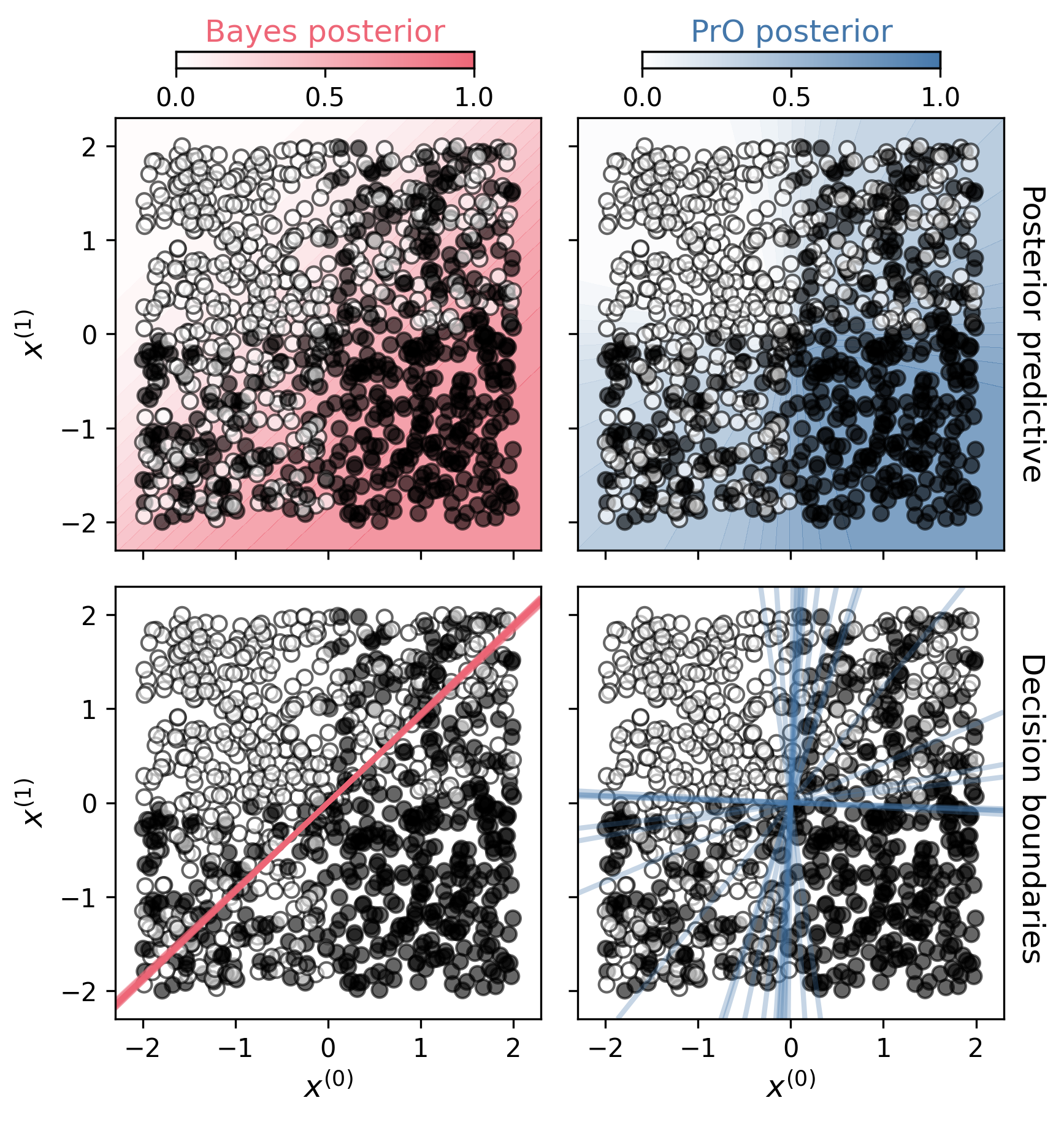}
    \caption{ 
    \textit{Top row:} \textcolor{gibbscolour}{\textbf{Bayes}} and \textcolor{procolour}{\textbf{\ac{pro}}} posterior predictive distributions overlaid with the raw data.
    \textit{Bottom row:} Decision boundaries induced by individual draws from the posterior predictives. 
    While the Bayes posterior concentrates onto a single decision boundary, the \ac{pro} posterior is  multimodal.
    }
    \label{fig:binary-classification}
\end{figure}
\begin{figure}[h!]
    \centering
    \includegraphics[width=0.66\linewidth]{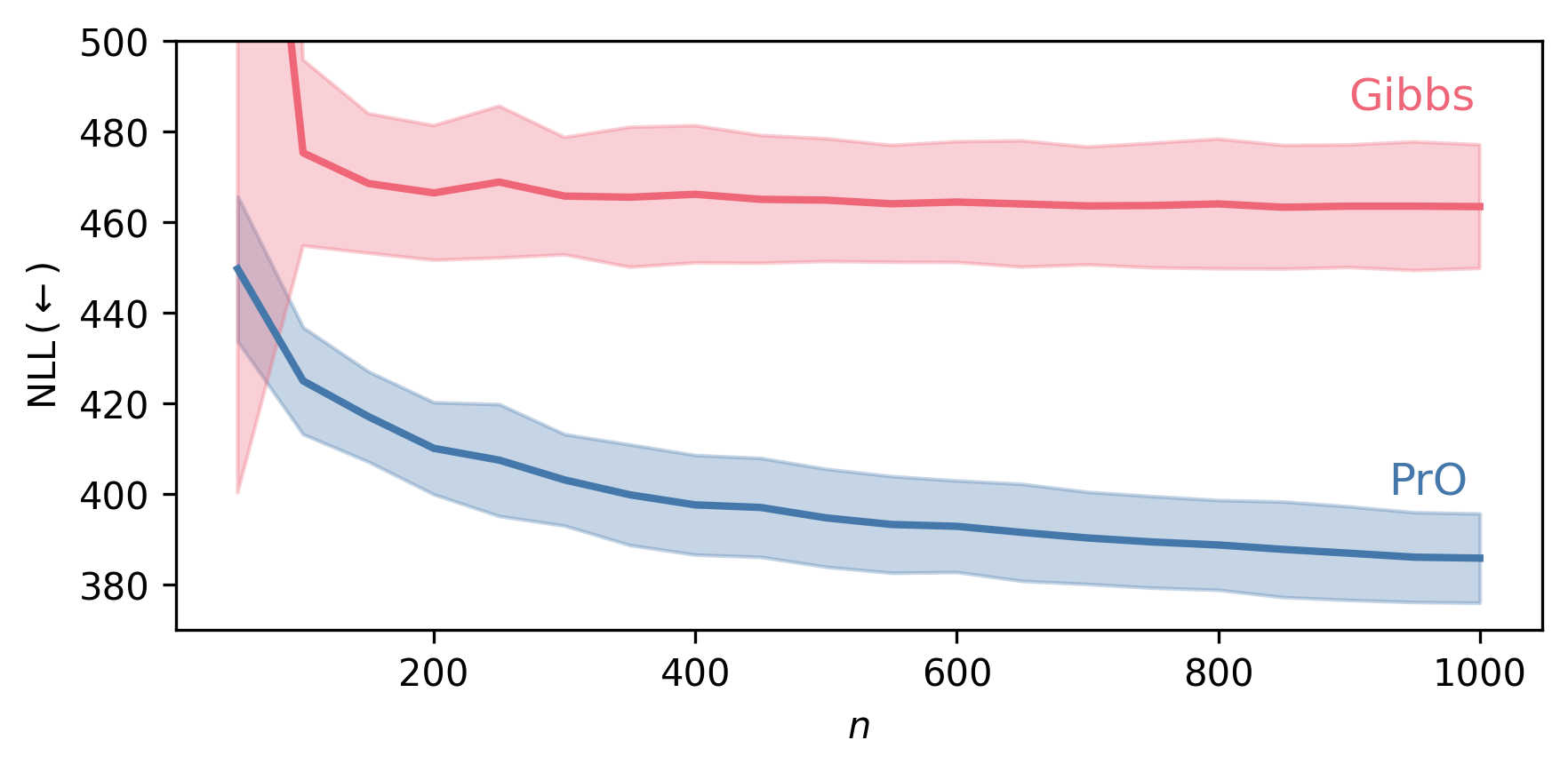}
    \caption{
    Constructing \textcolor{gibbscolour}{\textbf{Bayes}} and \textcolor{procolour}{\textbf{\ac{pro}}} posteriors using $n$ observations, we plot their
    negative log likelihood (NLL) on test data, with standard error intervals obtained from $20$ replications.
    }
    \label{fig:binary-classification-asymptotic}
\end{figure}

\subsection{Spatial autoregressive modelling}
\label{sec:gmrf}
Throughout, our theoretical developments have assumed independently sampled data.
In practice however, prediction is often most interesting for sequentially or spatially dependent data.
While our current analysis does not explicitly account for this, the following example shows that \ac{pro} posteriors continue to function as expected 
in spatial settings.
Here, we model the median house price across $n = 506$ different census tracts in Boston in the 1970s and 1980s.
These data were originally collected by \citet{harrison_hedonic_1978}, and are summarised graphically in the left-most panel of Figure~\ref{fig:gmrf}. 
Here, census tracts $i$ and $j$ are displayed as nodes in an undirected graph with  adjacency matrix $\Omega$ whose entry at $(i,j)$ is given by $\omega_{i,j} = 0$ unless the tracts $i$ and $j$ share at least one vertex, in which case $\omega_{i,j} = 1$.

Median house values are then modelled as the conditional auto-regression
\begin{equation*}
    y_i \mid y_1, \dots y_{i-1}, y_{i+1}, \dots y_n \sim \normal\left(\theta\sum_{j\ne i} \tilde \omega_{i,j}y_j,\sigma^2\right),
\end{equation*}
where $\tilde \omega_{i,j}$ denote the rescaled $\omega_{i,j}$ so that $\sum_{j\ne i} \tilde \omega_{i,j} = 1$ for all $i$ and where $\sigma^2 > 0$ is a fixed constant which we estimate from data (for a more detailed explanation, see   Appendix~\ref{appendix:reg-mmd}).
From this model, we construct two Gibbs and \ac{pro} posteriors for $\theta$ using the squared MMD and log score respectively; 
see Appendix~\ref{app:supp-gmrf} for further details.

\begin{figure}[h!]
    \centering
    \includegraphics[width=\linewidth]{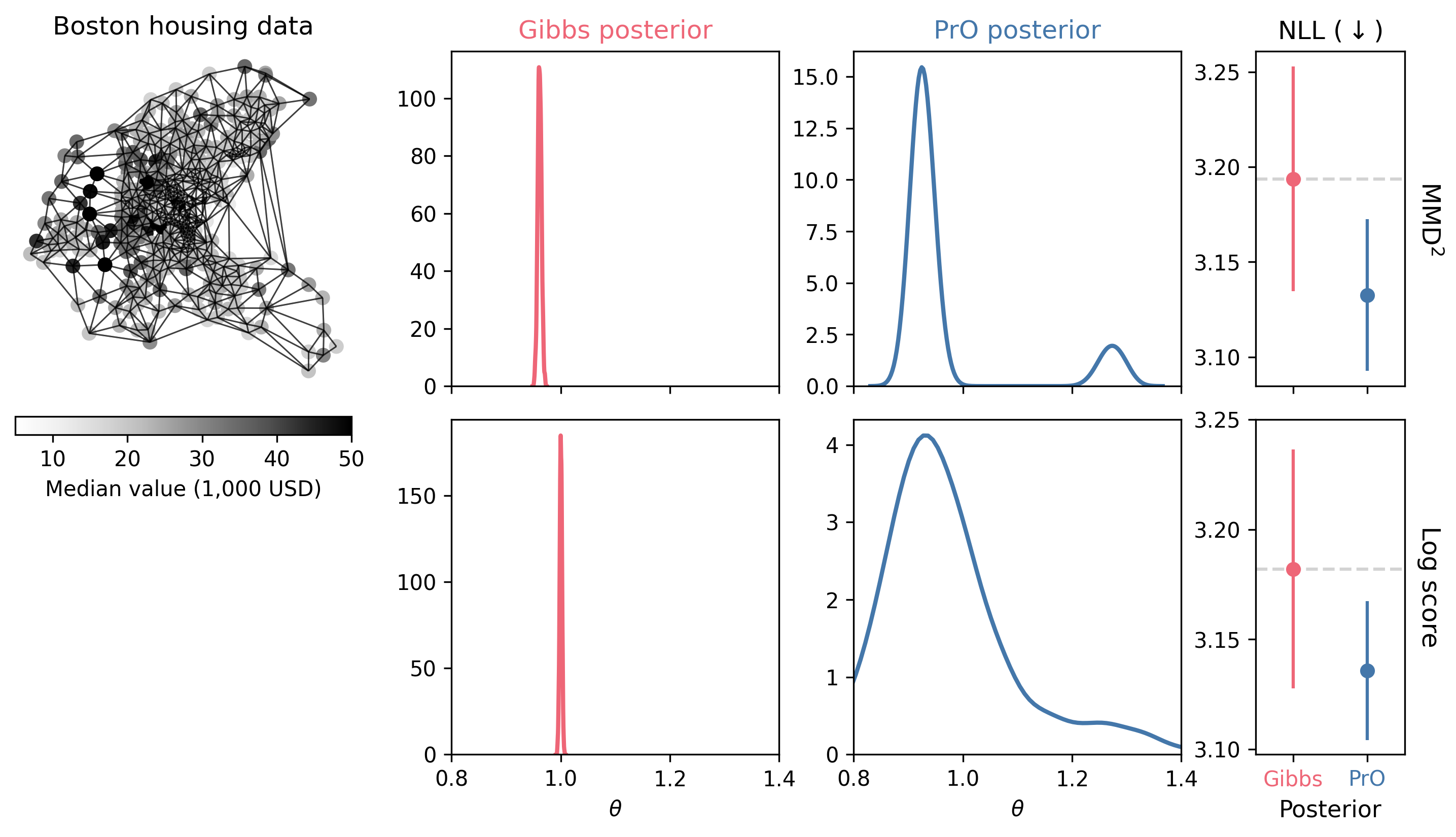}
    \caption{
    \textit{Top left:} illustration of the Boston housing data set.
    Nodes correspond to census tracts, and shade represents price level.
    Census tracts are connected by edges if they are spatially adjacent.
    \textit{Middle columns:}
    {\bfseries\color{gibbscolour} Gibbs} and {\bfseries\color{procolour} \ac{pro}} parameter posteriors under the MMD score (top row) and the log score (bottom row).
    \textit{Right column:} mean and standard error of the leave-one-node-out negative log likelihood (lower is better).}
    \label{fig:gmrf}
\end{figure}

Importantly, the underlying data distribution is bimodal, and has its primary mode at around $\$20,000$ and a second smaller mode at $\$50,000$ (see Appendix~\ref{app:supp-gmrf}), with the second mode is due to censoring values exceeding $\$50,000$ \citep{gilley_harrison_1996}.
In other words, the data set has two sub-groups: 
tracts with median prices close to $\$20,000$, and tracts with median prices are more than twice that value.
Our results in Figure~\ref{fig:gmrf} show the effect on inference:
while the Gibbs posterior fails to capture this hidden structure in the data, the \ac{pro} posterior transfers it to the parameter space. 
Its bi-modality reflects the fact that while there is a small group of census tracts with very high house prices surrounded by cheaper areas, a much larger group will be surrounded by similarly priced tracts.
While the former group is  accounted for through a smaller mode around $\theta \in [1.2, 1.4]$, the latter is captured by the larger mode around  $\theta \in [0.8, 1.0]$.
As the rightmost panel of Figure~\ref{fig:gmrf} confirms, this is not only a richer description of parameter uncertainty, but also leads to better predictive performance.
Confirming our earlier observation that the choice of scoring rule seems to matter less for \ac{pro} posteriors, these results and interpretations are qualitatively the same whether one uses the squared MMD or the logarithmic scoring rule.
\subsection{Photometric redshift prediction}
\label{sec:redshift}
We conclude our numerical demonstrations with a large-sample regression task.
Using approximately $78{,}000$ galaxies and quasars from the seventeenth Sloan Digital Sky Survey (SDSS) published in \citet{abdurrouf_seventeenth_2022}, we predict each source's redshift.
We do so by considering five variables: 
the differences in brightness between wavelengths in ultraviolet  and green, 
green and red,
red and near-infrared, 
and two near-infrared bands,
as well as the magnitude of the object's red band.
Additionally, we use an intercept term and obtain all polynomial and interaction terms up to order two, which yields a regression model with $21$ parameters---similarly to what was previously done in  \citet{connolly1995slicing} and \citet{brunner1997toward}.
On this model, we then compute the \ac{pro} posterior with the logarithmic scoring rule and contrast its inferences with those obtained from a standard Bayes posterior in Figure~\ref{fig:sdss}.
Further details can be found in Appendix~\ref{app:supp-sdss}.
\begin{figure}[t!]
    \centering
    \includegraphics[width=\linewidth]{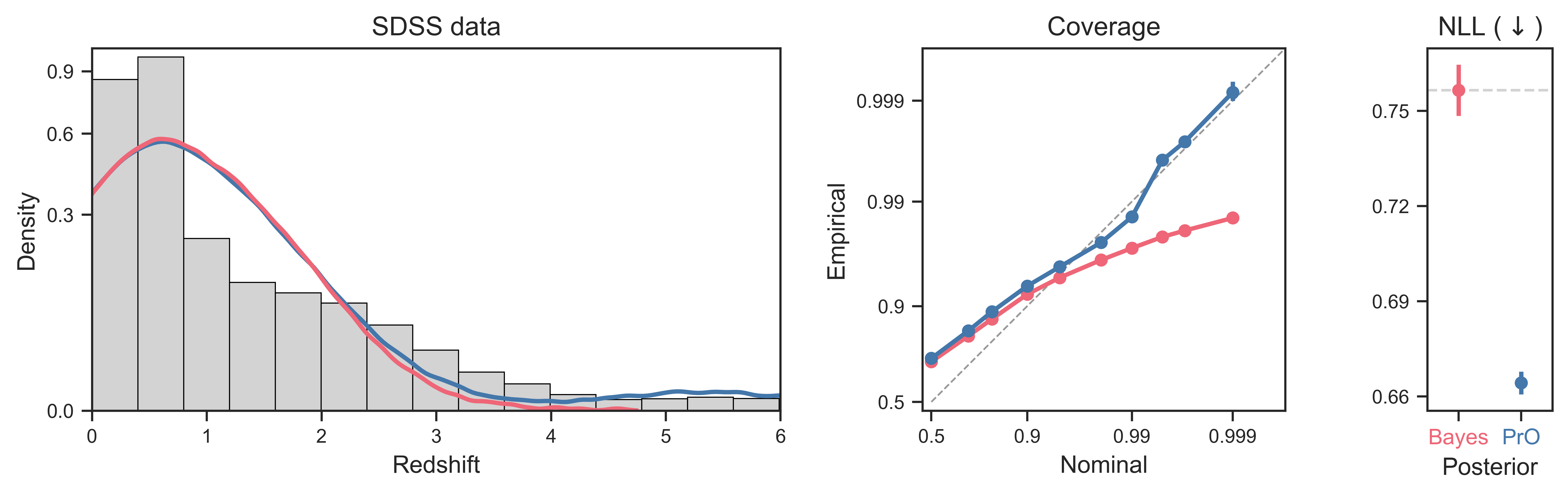}
    \caption{
    \textit{Left:} 
     {\bfseries\color{gibbscolour} Bayes} and {\bfseries\color{procolour} \ac{pro}} posterior predictive densities fit on a training set (in square root scale) and averaged over the test set of a representative cross-validation fold, overlaid with all of the observed data in the grey histogram.
    \textit{Middle:}
    nominal and empirical coverage of the predictive intervals in logit scale; 
    mean and standard error shown across $10$ cross-validation folds.
    \textit{Right:} 
    mean and standard error of the $10$-fold cross-validated negative log likelihood (lower is better).}
    \label{fig:sdss}
\end{figure}

As the \textit{Left} panel of Figure~\ref{fig:sdss} shows, the redshift distribution is strongly right-skewed: 
the majority of observations have low redshift, with a long and heavy tail of high-redshift sources.
This is challenging for the Bayes posterior, as the tails are wider than the  model can account for.
In contrast, the \ac{pro} posterior can average over multiple models to produce predictive distributions which are well-adapted to this feature.
As a result, its coverage is far superior to that of the Bayes posterior (see \textit{Middle} panel of Figure~\ref{fig:sdss}), which is in line with the theoretical results obtained in Section \ref{sec:calibration}.
Unsurprisingly given the skewed nature of the redshift distribution, this feature is particularly notable in the tails: the Bayes posterior's predictive intervals become systematically over-confident for nominal coverage levels $>0.9$.
A direct result is that the \ac{pro} posterior attains a better overall predictive performance than the Bayes posterior as measured with cross-validated negative log predictive density (see \textit{Right} panel of Figure~\ref{fig:sdss}), where the improvement is most likely due to the tail events which caused the Bayes posterior's miscalibration.

\section{Discussion}
\label{sec:discussion}
In this paper, we have advocated for a new statistical paradigm that characterizes epistemic uncertainty as a function of a model's predictive capability.%
This led to the \ac{pro} posterior, a method for uncertainty quantification that predictively dominates competing  Bayesian and generalised Bayesian methods.
We proved that \ac{pro} posteriors adapt to model misspecification and avoid posterior concentration around a single  model parameterisation whenever this improves predictive performance.

While we have presented   evidence which advocates the adoption of \ac{pro} posteriors for inference and uncertainty quantification, much remains to be done.
Perhaps the most pressing concern is  computation: sampling schemes obtained from interacting particle systems require further work before they can become an attractive off-the-shelf algorithmic framework.
Na\"ive implementations are infeasible for high-dimensional problems, generally computationally demanding, and require careful tuning of several hyperparameters.
However, since research on improving such samplers has recently intensified \citep[see][]{suzuki2023convergence, fu2023mean, liu2023polyak, chazal_computable_2025}, we are hopeful that more purpose-built algorithms can overcome these limitations.

Beyond computation, there are several open theoretical questions.
The most obvious of these concerns the choice of scoring rule for \ac{pro} posteriors.
While this choice is relatively unimportant in the case of convex recovery, we expect there to be important qualitative differences for certain forms of non-trivial misspecification.
For example, with Huber-contaminated  data-generating processes of the form $P_0 = (1-\varepsilon)P_{Q^*} + \varepsilon \cdot c$ for some $Q^* \in \mathcal{P}(\Theta)$ and some contamination $c \in \mathcal{P}(\mathcal{X})$, existing research on parametric inference \citep[see e.g.][]{cherief2020mmd,cherief2022finite,dellaporta2022robust,alquier2023estimation,alquier_universal_2024} would indicate that \ac{pro} posteriors based on outlier-robust scoring rules like the squared maximum mean discrepancy should recover  posteriors very similar to $Q^*$, while logarithmic scoring rules would behave completely differently. 
A second  question concerns the degree to which data can be allowed to be dependent: 
while our theory was derived under the assumption of independent data, our experiments suggest that the predictive benefits of \ac{pro} posteriors extend beyond this setting.

In the bigger picture, \ac{pro} posteriors have interesting connections to a number of other approaches grounded in frequentist statistics and particularly NPMLE, and studying their intersection may result in mutually beneficial insights.
Similarly, there are parallels with the efforts to develop optimal aggregation procedures for predictors.
For example,  \citet{nemirovski2000topics} introduces \textit{convex aggregation} as the statistical task of finding the optimal convex combination of predictors.
Intriguingly, for model classes like bounded regression with a finite number of candidates, the minimax-optimal convex aggregation procedure with respect to the squared risk can be shown to coincide precisely with  \eqref{eq:pro-posterior-def} when the KL-regularisation term is removed \citep[see, e.g.,][]{lecue2013empirical,yang_minimax_2014}.
This also suggests that \ac{pro} posteriors may be close to minimax-optimal convex aggregation procedures for some problems.

\subsection*{Acknowledgements}

Science is a community effort; and this manuscript in particular has benefited tremendously from interactions with many brilliant colleagues.
Most notably, we feel deep gratitude towards Prof. Chris Oates for many conversations and insightful remarks.
Further, we would like to thank Dr. Louis Sharrock and Dr. Sam Power for their help with the computational aspects of this paper, and Dr. Sharrock further for pointing out an inconsistency in an earlier presentation of Corollary~\ref{corollary:concentration-dtf}.
We are also grateful for correspondence with Prof. Pierre Alquier that pointed to the connection with PAC-Bayesian majority voting.
Additionally, we owe a tremendous debt of gratitude to several colleagues who took time to read an earlier version of this manuscript and gave incredibly insightful comments---often pages and pages of them!---that significantly improved the paper; most notably Prof. David Dunson, Prof. Jeff Miller, Prof. Aki Vehtari, Prof. Pierre Jacob, Dr. Andres Masegosa, Dr. Edwin Fong, Dr. Jack Jewson, and Dr. Takuo Matsubara.
Lastly, we thank an anonymous panel of reviewers at JRSS-B's discussion track for several helpful pointers.

{
\small
\spacingset{1.0} %
\bibliographystyle{chicago}
\bibliography{Bayes_comp,bib2}
}

\newpage

\section{Detailed Theoretical Developments}
\label{sec:theoretical-details}

Herein, we fully develop and explain the results summarised in Theorem \ref{thm:DI-master}, and discuss some additional results left out from the main paper. Throughout all results, we take $\lambda_n = n^{1/2}/\sqrt{kC^2\log(n)^2}$ for some universal constant $C$ which depends on the particular assumptions imposed,  and  where we set $k=1$ when Assumption \ref{ass:entropy} is satisfied.

\subsection{A Preliminary Result}
Recall that $\mathcal{D}_S$ represents the divergence associated with the scoring rule $S(\cdot,x)$, where $\mathcal{D}_S(P,P_0)=\mathcal{S}(P,P_0)-\mathcal{S}(P_0,P_0)$, and $\mathcal{S}(P,P_0)=\E_{X\sim P_0}[S(P,X)]$. 

Regardless of the fact that that different scores can require slightly different regularity conditions, \ac{pro} posteriors deliver similar theoretical results in both cases.  %

\begin{lemma}
\label{lemma:pac-bayes-bound-master}
Assumption \ref{ass:convex} is satisfied. If Assumption \ref{ass:entropy} or Assumption \ref{ass:Global} holds, then 
\begin{IEEEeqnarray}{rCl}
        \E\left[\mathcal{D}_S(P_{Q_n},P_0)\right]
		\le \inf_{Q\in\mathcal{P}(\Theta)}
        \left\{
        \mathcal{D}_S\left(P_Q, P_0\right) +\frac{\KL({Q\|\Pi})}{\lambda_n}\right\}+\nu_n.
        \nonumber
\end{IEEEeqnarray} for $\nu_n$ as defined in \eqref{eq:rates}. 
\end{lemma}

Lemma \ref{lemma:pac-bayes-bound-master} delivers a meaningful preliminary bound on the predictive performance of \ac{pro} posteriors. Further, this result applies to \ac{pro} posterior based on \textit{any predictive score} that satisfies our Assumptions.

\subsection{Generic Misspecification}

Under mild regularity conditions,  Gibbs posteriors based on scoring rules concentrate towards $\theta^{\star} := \argmin_{\theta\in\Theta}\mathcal{S}(P_{\theta}, P_0)$, which defines the model that is closest to $P_0$ in terms of $\mathcal{S}$.
In fact, results like this only require existence of $\theta^{\star}$ and mild conditions on the probability mass that $\Pi$ has around $\theta^{\star}$. 
The next assumption provides such conditions, and is sufficient for Gibbs posterior concentration to occur.
In the remainder, we will use these assumptions to compare the behaviour of Gibbs posteriors with that of \ac{pro} posteriors as $n$ increases.

\begin{assumption}
\label{ass:prior-mass-condition}
There exists $\theta^{\star}=\argmin_{\theta\in\Theta}\mathcal{S}(P_\theta,P_0)$, and  constants $c_S>0$ and $d_{S}>0$ and some $r_0 > 0$ such that, for any $r \leq r_0$, and $\mathcal{B}_r:=\left\{\theta \in \Theta : \mathcal{D}_S(P_{\theta},P_{0}) - \mathcal{D}_S(P_{\theta^\star},P_{0}) \leq r\right\}$, $\Pi\left(\mathcal{B}_r\right) \geq ({r}/{c_S})^{d_{S}}$ .
\end{assumption}

The above is slightly stronger than the classical prior mass condition of \citet{ghosal:ghosh:vdv:00}, insofar as the condition is essentially geared toward parametric models: 
while it may not be satisfied outside of Euclidean spaces, it is  mild and unrestrictive when $\Theta$ is a Euclidean space. 
In fact, the above assumption is considered standard in the PAC-Bayes literature, where interest traditionally is more focused on parametric models (see, e.g., Section 4.4 of \citealp{alquier2024user}); 
however, we note that this condition can be replaced with the standard prior mass condition, in general rates, without meaningfully altering our main results. Using Assumption \ref{ass:prior-mass-condition}, we obtain the first result comparing predictive performance between \ac{pro} posteriors and Gibbs posteriors in the case where we allow for generic model misspecification.

\begin{theorem}
\label{theorem:Gibbs-vs-pro-posterior-misspecification}
Under Assumptions \ref{ass:convex}-\ref{ass:prior-mass-condition}, 
$
\E [\mathcal{D}_S\left( P_{Q_n},P_0\right)-\mathcal{D}_S(P_{\theta^\star},P_0)]\lesssim 
\nu_n.
$ Furthermore, 
$
\E [\mathcal{D}_S\left( P_{Q^\dagger_n},P_0\right)-\mathcal{D}_S(P_{\theta^\star},P_0)]\lesssim 
\log(n)/n^{1/2}.
$
\end{theorem}

In words, 
Theorem \ref{theorem:Gibbs-vs-pro-posterior-misspecification} says that the \ac{pro} posteriors predictive performance is comparable to that of the Gibbs posterior based on the same scoring rule. While this may be reassuring, it does not show that the \ac{pro} posterior confers a clear advantage.
The reason for this is that Assumption \ref{ass:prior-mass-condition} does not rule out the case where the predictively optimal average $Q^{\star} \in \arginf_{Q \in \mathcal{P}(\Theta)}\mathcal{S}(P_Q, P_0)$ is a point mass at $\theta^{\star} = \argmin_{\theta \in \Theta}\mathcal{D}_{\mathcal{S}}(P_{\theta}, P_0)$, the score-minimising value.
In this case---which is unlikely to occur in sufficiently complex data sets---Gibbs and \ac{pro} posteriors alike would allocate most of their probability mass around $\theta^{\star}$ as $n\to\infty$, and would not produce meaningfully different predictive distributions.

\subsection{Non-Trivial Misspecification}
By excluding trivial forms of misspecification, we can obtain a clearer idea of the advantages \ac{pro} posteriors maintain over Bayes and Gibbs posteriors. Following the idea of \textit{bad misspecification} presented by \citet{grunwald2017inconsistency}, the more general notion of non-trivial misspecification in Definition \ref{def:misspecification} will ensure that averaging over $\mathcal{M}_{\Theta}$ is predictively advantageous. 

\begin{assumption}
\label{ass:non-trivial-misspecification-exact}
There exists $
Q^{\star}\in\arg\inf_{Q\in\mathcal{P}(\Theta)} \mathcal{D}_S\left(P_Q,P_0\right)
$, such that $\KL(Q^\star\|\Pi)<\infty$ and 
$
\mathcal{D}_S\left( P_{\theta^{\star}} ,P_0\right)>\mathcal{D}_S\left(
    P_{Q^{\star}},P_0\right).
$ Further, the mapping $\theta\mapsto P_\theta$ is continuous and identifiable. 
\end{assumption}

When the magnitude of misspecification is small, the Gibbs posterior is already highly performant, so that the \ac{pro} posterior can only achieve a small amount (if at all) of predictive superiority.
Non-trivial misspecification is therefore sufficient for the \ac{pro} posterior to achieve predictive domination in any score.

\begin{theorem}
\label{theorem:generalisation-guarantee-exact}
Under Assumptions
\ref{ass:convex}-\ref{ass:non-trivial-misspecification-exact} and for $n$ large enough but finite, it holds that $\mathbb{E}\left[\mathcal{D}_S\left(
P_{Q_n},P_0 \right)\right] < \mathbb{E}[\mathcal{D}_S(
P_{Q_n^\dagger}
,P_0)]$, and  $\mathbb{E}\left[\mathcal{D}_S\left(
P_{Q_n}
,P_0 \right)-\mathcal{D}_S(
P_{Q^\star},P_0 )\right]\lesssim \nu_n$. %
\end{theorem}
The above result constitutes part \ref{item:master-theorem-DI-non-trivially-NT} in Theorem \ref{thm:DI-master}.
The result shows that under suitable regularity conditions, \ac{pro} posteriors strictly dominate both Bayes posteriors  and Gibbs posteriors in terms of their predictive performance.
Crucially, this is not an asymptotic result: 
it holds in expectation for finite samples.
While the result shows the superior predictive performance of \ac{pro} posteriors, it leaves several key questions unanswered: 
what is the qualitative behaviour of \ac{pro} posteriors? 
Can we characterise the shape they take, at least asymptotically? 
And how practically significant is the performance gap?
Unfortunately, it is not possible to establish these properties in full generality: 
without further assumptions on the nature of the misspecification, we could be arbitrarily close to being in the trivially misspecified regime. In the next section, we further refine the nature of misspecification we consider to deduce an answer to these questions.

\subsection{Misspecification with Convex Recovery}
From Definition \ref{def:misspecification}, we recall that convex recovery implies the existence of some $Q \in \mathcal{P}(\Theta)$ for which $P_Q$ can recover the data-generating process $P_0$. Interestingly,  the \ac{pro} posterior recovers the truth so long as $P_0$ is convexly recoverable as either a purely discrete or continuous mixture. 
\begin{assumption}\label{ass:prior-mass-convex-recovery-general} One of the following is satisfied.
\begin{itemize}
    \item[(i)] There exists $Q^\star\in\mathcal{P}(\Theta)$, with $\KL(Q^\star\|\Pi)<\infty$, such that $P_0=\int P_\theta \dt Q^\star(\theta)$.
    \item[(ii)]For some $J \in \mathbb{N}$, 
$P_0=\sum_{j=1}^J \omega_j P_{\theta_j^\star}$, where $(\omega_j)_{j=1}^J\in\Delta^{J-1}$ are non-negative weights summing to unity, and 
$\theta_j^{\star} \in \Theta$. For $\mathcal{B}_{r,j} = \{\theta:\mathcal{D}_S(P_{\theta},P_{\theta_j^\star})\leq r_{j}\}$, there are constants $c_S,d_S>0$ so that, for any $r\le r_0$, $\Pi(\mathcal{B}_{r,j})\ge (r/c_S)^{d_S}$  for $j=1,2,\dots, J$.
\end{itemize}
\end{assumption}
Assumption \ref{ass:prior-mass-convex-recovery-general} restricts the nature of the true model to be either a purely discrete or continuous mixture with components $P_\theta$. Assumption \ref{ass:prior-mass-convex-recovery-general}(i) requires that this mixture is absolutely continuous with respect to the prior, via the condition $\KL(Q^\star\|\Pi)<\infty$, while Assumption \ref{ass:prior-mass-convex-recovery-general}(ii) is its  discrete {analogue}, and is an extension of Assumption \ref{ass:prior-mass-condition} to $J$ components. Similar to Assumption \ref{ass:prior-mass-condition}, Assumption \ref{ass:prior-mass-convex-recovery-general}(ii) requires that the prior should place enough mass near each component of the discrete mixture representing the true data generating process $P_0$.
\begin{theorem}
\label{theorem:generalisation-guarantee-misspecification-mixture-exact-convex}
Under Assumptions \ref{ass:convex}-\ref{ass:prior-mass-convex-recovery-general},
and for $n$ sufficiently large but finite, 
$
\mathbb{E}\left[\mathcal{D}_S\left(P_{Q_n},P_0 \right)\right]<\mathbb{E}[\mathcal{D}_S(P_{Q_n^\dagger},P_0)]$, and $\mathbb{E}\left[\mathcal{D}_S\left(
P_{Q_n}
,P_0 \right)\right] \lesssim \nu_n$.%
\end{theorem}
The above result constitutes part \ref{item:master-theorem-DI-non-trivially-CR} in Theorem \ref{thm:DI-master}, and shows that, up to logarithmic factors, we recover the true data generating process $P_0$ at the standard parametric rate of $n^{-1/2}$ when Assumption \ref{ass:Global} is satisfied; if Assumption \ref{ass:entropy} is satisfied, we recover $P_0$ at the slower of the rates $r_n$ and the standard parametric rate.

\subsubsection{Correct Model Specification}\label{sec:correct_appendix_DI}

Theorems \ref{theorem:generalisation-guarantee-exact} and \ref{theorem:generalisation-guarantee-misspecification-mixture-exact-convex} prove that the \ac{pro} posterior has superior predictive properties to Gibbs and Bayes posteriors under model misspecification, but it is unclear what penalty we pay when using the \ac{pro} posterior in the context of a well-specified model.
The next result shows that there is little to worry about.

\begin{corollary}
\label{corollary:concentration-if-model-correct}
If  Assumptions \ref{ass:convex}-\ref{ass:entropy} and \ref{ass:prior-mass-convex-recovery-general}(ii) hold for $J=1$ and $\theta^{\star}_1 = \theta^{\star}$,  then
$\E\left[\mathcal{D}_S\left(P_{Q_n},P_0\right)\right]-\E[\mathcal{D}_S(P_{Q^\dagger_n},P_0)]\lesssim \nu_n
$.%
\end{corollary}
The above result constitutes part \ref{item:master-theorem-DI-WS} in Theorem \ref{thm:DI-master} and shows that there is little evidence for a loss in predictive statistical efficiency when using \ac{pro} posteriors in the well-specified regime: the worst case difference relative to Gibbs posteriors based on the same scoring rule is bounded by a vanishing sequence of order $\nu_n$, which is $\log(n)/{n}^{1/2}$ in the case of Assumption \ref{ass:Global} and $\max\{\log(n)/{n}^{1/2},r_n\}$ in the case of Assumption \ref{ass:entropy}.

Since Gibbs posteriors concentrate onto the true parameter $\theta^*$ in well-specified models, Corollary \ref{corollary:concentration-if-model-correct} raises the question of whether or not \ac{pro} posteriors behave the same.
Under additional assumptions, we can demonstrate that this is indeed the case. 
For $P,Q\in\mathcal{P}(\Theta)$, let $\dt_1(P,Q)$ denote a metric on the space of probability measures. 
Further, let $\dt_2(P,Q)$ be a symmetric discrepancy on $\mathcal{P}(\mathcal{X})$. 
Here, $\dt_2(P,Q)$ need not be a metric, but is assumed to satisfy a weak triangle inequality.
\begin{assumption}\label{ass:concentration-dtf}
\begin{enumerate*}[
    label=(\roman*), 
]
    \item The map $\theta\mapsto P_\theta$ is continuous and identifiable, and there exist $Q^\star\in\mathcal{P}(\Theta)$ such that $Q^\star=\argmin_{Q\in\mathcal{P}(\Theta)}\mathcal{D}_S(P_{Q},P_0)$;\label{ass:concentration-dtf-i}
    \item there exists constant $C$, and $\varepsilon>0$ such that, for any $Q,Q'\in\mathcal{P}(\Theta)$, if $\dt_2(P_{Q},P_{Q'})\le \varepsilon$, then  
    $
    \dt_1(Q,Q')\le C\dt_2(P_{Q},P_{Q'})
    $; and \label{ass:concentration-dtf-ii}
    \item for some $\alpha\ge1$ {and a constant $K>0$},  $\dt_2(P_Q,P_{Q'})^\alpha\le {K}\mathcal{D}_S(P_Q,P_{Q'}) $. \label{ass:concentration-dtf-iii}
\end{enumerate*}
\end{assumption}

Assumption \ref{ass:concentration-dtf} allows us to translate concentration in the \ac{pro} posterior predictive $P_{Q_n}$ to concentration in the \ac{pro} posterior $Q_n$ defined on the parameter space.
While part \ref{ass:concentration-dtf-i} 
is standard, part \ref{ass:concentration-dtf-ii} gives a form of identifiability condition of the posterior.
In particular, it ensures that if two predictives $P_Q, P_{Q'}$ are close, then their underlying posteriors $Q, Q'$ must also be close.
Without such an assumption, concentration in the predictive space, $P_Q$, would not necessarily translate into concentration over $Q$. 
Lastly, part \ref{ass:concentration-dtf-iii} is a type of reverse H\"{o}lder continuity, that depends on the choice of score, distance, and the model class $\mathcal{M}_{\Theta}$. 
Its purpose is to relate the score-induced discrepancy measure to a more regular measure of distance on the space of measures, which is needed to transfer convergence.
While it is difficult to establish these conditions in general, they can be shown to hold for the squared MMD and the logarithmic scoring rule for several model classes of interest.

\begin{corollary}\label{corollary:concentration-dtf}
If Assumptions \ref{ass:convex}-\ref{ass:entropy}, and \ref{ass:prior-mass-convex-recovery-general}-\ref{ass:concentration-dtf} hold, then for any sequence $M_n\rightarrow\infty$ such that $M_n\nu_n^{1/\alpha}=o(1)$,   
    \begin{equation*}
        P_0\left\{\dt_1(Q_n,Q_\star )
        \le 
        M_n\nu_n^{1/\alpha}\right\}\ge1-o(1)
    \end{equation*}
\end{corollary}

Under various choices of $\alpha$, $\mathrm{d}_1$, and $\mathrm{d}_2$, Corollary~\ref{corollary:concentration-dtf} applies to both the squared MMD (with $\alpha=2$ and $\mathrm{d}_2(P,Q)=\mathrm{MMD}(P,Q)$) and to the intractable exact logarithmic scoring rule (with $\alpha=2$ and $\mathrm{d}_2(P,Q)=\mathrm{TV}(P,Q)$), which satisfy Assumption~\ref{ass:concentration-dtf} part~\ref{ass:concentration-dtf-iii}. 
For the special case of regular parametric models, the local equivalence between $\mathrm{d}_2(P_\theta,P_{\theta^\star})$ and the Euclidean distance $\|\theta-\theta^\star\|_2$ around the true value $\theta^\star$ allows us to refine our conclusions.
In particular, in this setting, for any posterior $Q$ concentrating around $\theta^\star$, the squared distance $\mathrm{d}_2(P_Q,P_{\theta^\star})^2$ behaves like $\|\mathbb{E}_{\theta\sim Q}[\theta]-\theta^\star\|_{2}^2+\textnormal{Tr}(\mathbb{V}_{\theta\sim Q}[\theta])$.
Thus, if we chose a squared metric $\mathrm{d}_1(Q,Q')^2$ on posterior distributions that 
locally behaved like $\|\mathbb{E}_{\theta\sim Q}[\theta]-\mathbb{E}_{\theta\sim Q'}[\theta]\|_{2}^2 + \|\mathbb{V}_{\theta\sim Q}[\theta]-\mathbb{V}_{\theta\sim Q'}[\theta]\|^2$ for a suitable metric on matrices, we would recover posterior concentration rates towards $\theta^\star$ directly in the Euclidean distance.
Put differently, up to a local approximation, the notion of concentration used in Corollary~\ref{corollary:concentration-dtf} via such a metric $\mathrm{d}_1$ is equivalent to the classical definition of posterior concentration with respect to the Euclidean distance for regular parametric models. 
In this setting, and if  $\alpha=2$, our results would now translate to a concentration rate of $n^{-1/4}$ in the Euclidean distance, which the reader will notice is slower than the standard $n^{-1/2}$ rate that is usually achievable in parametric models. 
Critically, this behaviour is \textit{not} a feature of the \ac{pro} posterior in particular, but of the temperature scaling employed within PAC-Bayes-type bounds.
Indeed, the $n^{-1/4}$ rate we achieve in the Euclidean space is a direct consequence of minimising our PAC-Bayes-type generalisation bounds for the \ac{pro} posterior by choosing the rate $\lambda_n \propto n^{1/2}$. 
Note that this result is entirely consistent with the application of PAC-Bayes bounds to Gibbs posteriors: here too, one often obtains $\lambda_n\propto n^{1/2}$ by minimising generalisation bounds, which would then \textit{also} result in a $n^{1/4}$ concentration rate on the Euclidean space by the same arguments as before \citep[see e.g.][]{alquier2016properties}.

\doparttoc %
\faketableofcontents %
\part{} %

\newpage
\appendix

\addcontentsline{toc}{section}{Appendix} %
\part{Appendix} %
\parttoc %

\section{Notations}
In this section we re-describe the main notations introduced in the main text to ensure the appendix is as self-contained as possible. We consider observed data $\x=(x_1,\dots,x_n)$ for $x_i\in\mathcal{X}$ and $i=1,\dots,n$, where each $x_i$ is generated from the true unknown distribution $P_0$. We denote expectations of functions $g:\mathcal{X}\to\mathbb{R}$ with respect to $X\sim P_0$ by $\mathbb{E}[g(X)]$ or by $\mathbb{E}_{X\sim P_0}[g(X)]$ if we wish to emphasise the distribution over which integration occurs. The model class is given by $\mathcal{M}_{\Theta} := \{P_{\theta}: \theta\in\Theta\}$ and a prior measure $\Pi \in \mathcal{P}(\Theta)$ over elements of $\mathcal{M}_{\Theta}$. Likewise, let $\mathcal{P}(\Theta)$ denote the set of distributions over $\Theta$, and recall $P_Q=\int P_\theta \dt Q(\theta)$ for $Q\in\mathcal{P}(\Theta)$. The function $S(P,x)$ denotes a scoring rule for the distribution $P$, while $\mathcal{S}(P,P)=\E_{X\sim P}[S(P,X)]$ denotes the expected score, and $\mathcal{D}_S(P,P_0)=\mathcal{S}(P,P_0)-\mathcal{S}(P_0,P_0)$ is the divergence associated to $P$, evaluated at the true distribution $P_0$. Further, define $S_n(P_\theta)=\frac{1}{n}\sum_{i=1}^{n}S(P_\theta,x_i)$, and $L_n(\theta_{1:k})=\frac{1}{n}\sum_{i=1}^{n}L(\theta_{1:k},x_i)$, where $L(\theta_{1:k},x)$ is as defined in Assumption \ref{ass:Global}.

\section{Lemmas}\label{app:lemmas}
This section contains several useful lemmas that are used throughout the remainder of the text. 
{First, we demonstrate our claim that \ac{pro} posterior based on any kernel scoring rule are tractable with $k=2$. This implies that for scores based on squared maximum mean discrepancy (MMD$^2$) and Cumulative Ranked Probability Score (CRPS),  the function $\Delta(Q, x)$ in \eqref{eq:predictive-score-decomposition} is constant in $x$.

To prove this, take $x,x'\in\mathcal{X}$, let $\kappa(x,x')$ denote a positive-definite kernel function with feature map $\phi:\mathcal X\to\mathcal H$, such that $\kappa(x,x')=\langle\phi(x),\phi(x')\rangle_{\mathcal H}$.  For fixed $\theta\in\Theta$, define the mean embedding $\mu_\theta\in\mathcal H$ by 
$
\mu_\theta := \int \phi(z)\,\dt P_\theta(z)$ and the inner product
$
\langle\mu_\theta,\mu_{\theta'}\rangle_{\mathcal H}
=\int\!\!\int \kappa(z,z')\,\dt P_\theta(z)\dt P_{\theta'}(z').$ Following \cite{gneiting2007probabilistic}, all kernel scoring rules can be expressed as 
$S(P,x)=\E_{X,X'\sim P}[\kappa(X,X')]-2\E_{X\sim P}[\kappa(X,x)]+\kappa(x,x).$ 
\begin{lemma}\label{lemma:kernel_scores}
For any kernel scoring rule, and for all $x\in\mathcal{X}$: (i) the score $P_\theta\mapsto S(P_\theta,x)$ is convex; (ii) $
S(P_Q,x)=\int L_{}(\theta_{1:2},x)\dt Q^2(\theta_{1:2})$ with
$
L_{}(\theta_{1:2},x)= S(P_{\theta_1},x)- \frac{1}{2}\|\mu_{\theta_1}-\mu_{\theta_2}\|^2_{\mathcal{H}}.
 $
\end{lemma}
}
\begin{proof}
To see that (i) is satisfied, use the feature map representation to write the score as
$$
S(P_\theta,x)=\|\mu_\theta\|_{\mathcal{H}}^2-2\langle\mu_\theta,\phi(x)\rangle_{\mathcal{H}}+\|\phi(x)\|_{\mathcal{H}}^2.
$$The map $P_\theta\mapsto \mu_\theta$ is linear, so that $P_\theta\mapsto\|\mu_\theta\|_{\mathcal{H}}^2$ is quadratic. Thus, $P_\theta\mapsto S(P_\theta,x)$ is the composition of a linear map and a convex function, and so is convex. To deduce the second result, write 
$$
S\left(\int P_\theta\dt Q(\theta),x\right)=\int S(P_\theta,x)\dt Q(\theta)+\underbrace{S\left(\int P_\theta\dt Q(\theta),x\right)-\int S(P_\theta,x)\dt Q(\theta)}_{=-\Delta(Q,x)\text{ in \eqref{eq:predictive-score-decomposition}}}.
$$
Expand the first integral, and use the mean embedding definition for $\mu_\theta$, to write
\begin{flalign*}
S\left(P_Q,x\right)=\int\int \int\int \kappa(z,z')\,\dt P_\theta(z)\dt P_{\theta'}(z')\dt Q(\theta)\dt Q(\theta')&=\langle \E_{\theta\sim Q}(\mu_\theta),\E_{\theta'\sim Q}(\mu_{\theta'})\rangle\\&= \|\E_{\theta\sim Q}(\mu_\theta)\|^2_{\mathcal{H}}.
\end{flalign*}
Similarly, write 
\begin{flalign*}
\int S\left(P_\theta,x\right)\dt Q(\theta)&=\int \int\int \kappa(z,z')\,\dt P_\theta(z)\dt P_{\theta}(z')\dt Q(\theta)=\E_{\theta\sim Q}[\|\mu_\theta\|^2_{\mathcal{H}}].
\end{flalign*}
Hence, for $\Delta(Q,x)$ as in \eqref{eq:predictive-score-decomposition},
\begin{flalign*}
-\Delta(Q,x)=\|\E_{\theta\sim Q}[\mu_\theta]\|^2_{\mathcal{H}}-\E_{\theta\sim Q}[\|\mu_\theta\|^2_{\mathcal{H}}]&=-\E_{\theta\sim Q}\left[\|\mu_\theta-\E_{\theta'\sim Q}[\mu_{\theta'}]\|^2_{\mathcal{H}}\right]\\&=-\frac{1}{2}\int \|\mu_\theta-\mu_{\theta'}\|^2_{\mathcal{H}}\dt Q(\theta)\dt Q(\theta').
\end{flalign*}
\end{proof}

\begin{lemma}\label{lem:gibbs}
Under Assumptions \ref{ass:convex}-\ref{ass:prior-mass-condition},  for any $\lambda_n\geq d_S/r_0$, we have
$$
\E \bigintssss \mathcal{D}_S(P_\theta,P_0) \dt Q_n^\dagger(\theta) \le \mathcal{D}_S(P_{\theta^\star},P_0) + \frac{d_S}{\lambda_n} \, \log\left(\dfrac{e\cdot c_S\lambda_n}{d_S}\right) + \frac{\lambda_nC^2}{n} \, .
$$
Furthermore, for the specific choice $\lambda_n=n^{1/2}\, \sqrt{{d_S}/{C^2} \, \log\left({e\cdot c_S}/{d_S}\right)}$, for any $n\,\geq \frac{d_S C^2}{r_0^2\,\log\left(\frac{e\cdot c_S}{d_S}\right)}$, we have 
$$   
\mathbb{E}\left[\mathcal{D}_S\left(
P_{Q_n^\dagger}
,P_0 \right)\right] \le \mathcal{D}_S(P_{\theta^\star},P_0) + C\left(\log n/n\right)^{1/2} .
$$
\end{lemma}
\begin{proof}
The proof follows similar arguments to those used to prove Lemma \ref{lemma:pac-bayes-bound-master} but where we take 
$
L(\theta,\theta')=R(\theta)=\mathcal{S}(P_\theta,P_0)-\mathcal{S}(P_{0},P_0)
$ and 
$
L_n(\theta,\theta')=R_n(\theta)=S_n(P_\theta)-S_n(P_{0}).
$ If  Assumption \ref{ass:Global} is satisfied for $R(\theta)-R_n(\theta)$, then 
$$
1\ge \E_{} 
\int
e^{\lambda_n\left\{R(\theta)-R_n(\theta)\right\}-\frac{\lambda_n^2C^2}{n}} \dt \Pi(\theta).
$$
Apply the Donsker-Varadhan lemma with $h(\theta)=R(\theta)-R_n(\theta)$ to obtain 
\begin{flalign*}
	1\ge \E_{} e^{\sup_{Q\in\mathcal{P}(\Theta)}\lambda_n \int\left[R(\theta)-R_n(\theta)\right]\dt Q(\theta)-\KL(Q\|\Pi)-\frac{\lambda_n^2 C^2}{n}} .
\end{flalign*}
Apply Jensen's inequity and take logs to obtain 
\begin{flalign*}
	0\ge \E\sup_{Q\in\mathcal{P}(\Theta)}\left\{\lambda_n \int\left[R(\theta)-R_n(\theta)\right]\dt Q(\theta)-\KL(Q\|\Pi)-\frac{\lambda_n^2 C^2}{n} \right\}.	
\end{flalign*}Since the above is true for any $Q\in\mathcal{P}(\Theta)$, it is also true for $Q=Q^\dagger_n$. Applying this choice and re-arranging terms yields 
\begin{flalign*}
	\E_{}\int R(\theta) \dt Q^\dagger_n(\theta)\le \E_{}\left\{\int R_n(\theta)\dt Q^\dagger_n(\theta)+\frac{\KL({Q^\dagger_n\|\Pi})}{\lambda_n}+\frac{\lambda_n C^2}{n}\right\}.
\end{flalign*}Use the definitions of $R(\theta)$ and $R_n(\theta)$, and Assumption \ref{ass:convex}, to obtain the inequality 
\begin{flalign}
&\E_{}\left[\mathcal{S}(P_{Q_n^\dagger},P_0)-\mathcal{S}(P_0,P_0)\right]\nonumber \\
	\le&\E_{}\left[\int \{\mathcal{S}(P_\theta,P_0)-\mathcal{S}(P_0,P_0)\}\dt Q^\dagger_n(\theta)\right]
    \label{eq:lemma-5-auxiliary}
    \\
    \le&\E_{}\left\{\int \left\{S_n(P_\theta)-S_n(P_{0})\right\}\dt Q^\dagger_n(\theta)+\frac{\KL({Q^\dagger_n\|\Pi})}{\lambda_n}\right\}+\frac{\lambda_n C^2}{n}
    \nonumber \\
		\le& \E_{}\inf_{Q\in\mathcal{P}(\Theta)}\left\{\int \left\{S_n(P_\theta)-S_n(P_{0})\right\}\dt Q(\theta)+\frac{\KL({Q\|\Pi})}{\lambda_n}\right\}+\frac{\lambda_n C^2}{n}
    \nonumber
	\\=&\E_{}\inf_{Q\in\mathcal{P}(\Theta)}\left\{\int \left\{S_n(P_\theta)-S_n(P_{0})\right\}\dt Q(\theta)+\frac{\KL({Q\|\Pi})}{\lambda_n}\right\}+\frac{\lambda_n C^2}{n}\label{eq:lemma-5-auxiliary-2}.
\end{flalign}

Define the distributions $\Pi_{\mathcal{B}_r}$ indexed by a radius $r>0$ with corresponding density:
$$
\dt\Pi_{\mathcal{B}_r}(\theta):=\frac{1(\theta\in\mathcal{B}_r)\dt \Pi(\theta)}{\Pi(\mathcal{B}_r)} \quad \textnormal{with} \quad \mathcal{B}_r:=\left\{\theta\in\Theta:\mathcal{D}_S(P_{\theta},P_{0})\le \mathcal{D}_S(P_{\theta^\star},P_{0}) + r \right\} .
$$
We can now further upper bound \eqref{eq:lemma-5-auxiliary} by
restricting the infimum on the right-hand side of  \eqref{eq:lemma-5-auxiliary-2} to distributions belonging to the set $\{\Pi_{\mathcal{B}_r}:r>0\}\subset \mathcal{P}(\Theta)$, which yields
\begin{align*}
    \E \bigintssss \mathcal{D}_S(P_\theta,P_0) \dt Q_n^\dagger(\theta)&\le \E_{}\inf_{r>0}\left\{\int \left\{S_n(P_\theta)-S_n(P_{0})\right\}\dt \Pi_{\mathcal{B}_r}(\theta)+\frac{\KL({\Pi_{\mathcal{B}_r}\|\Pi})}{\lambda_n}\right\}+\frac{\lambda_n C^2}{n}\\&\le \inf_{r>0}\left\{\int \mathcal{D}_S(P_\theta,P_0)\dt \Pi_{\mathcal{B}_r}(\theta)+\frac{\KL({\Pi_{\mathcal{B}_r}\|\Pi})}{\lambda_n}\right\}+\frac{\lambda_n C^2}{n}%
    \\
    & = \inf_{r>0}\left\{ \bigintssss \mathcal{D}_S(P_\theta,P_0) \dt \Pi_{\mathcal{B}_r}(\theta) + \frac{\log\left\{1/\Pi(\mathcal{B}_r)\right\}}{\lambda_n} \right\} + \frac{\lambda_nC^2}{n} \\
    & \leq \inf_{0<r\leq r_0}\left\{ \bigintssss \left\{ \mathcal{D}_S(P_{\theta^\star},P_0) + r \right\} \dt \Pi_{\mathcal{B}_r}(\theta) + \frac{d_S\,\log\left(c_S/r\right)}{\lambda_n} \right\} + \frac{\lambda_nC^2}{n} \\
    & = \mathcal{D}_S(P_{\theta^\star},P_0) + \inf_{0<r\leq r_0}\left\{ r + \frac{d_S\,\log\left(c_S/r\right)}{\lambda_n} \right\} + \frac{\lambda_nC^2}{n} \, .
\end{align*}
As soon as $\lambda_n\geq \frac{d_S}{r_0}$, taking $r=\frac{d_S}{\lambda_n}\leq r_0$ then leads to
\begin{align*}
    \E \bigintssss \mathcal{D}_S(P_\theta,P_0) \dt Q_n^\dagger(\theta) & \le \mathcal{D}_S(P_{\theta^\star},P_0) + \frac{d_S}{\lambda_n} + \frac{d_S}{\lambda_n} \, \log\left(\dfrac{c_S\lambda_n}{d_S}\right) + \frac{\lambda_nC^2}{n} \\
    & = \mathcal{D}_S(P_{\theta^\star},P_0) + \frac{d_S}{\lambda_n} \, \log\left(\dfrac{e\cdot c_S\lambda_n}{d_S}\right) + \frac{\lambda_nC^2}{n} \, ,
\end{align*}
which establishes the first result.

To obtain the second, substitute $\lambda_n=n^{1/2}\,\sqrt{\frac{d_S}{C^2} \, \log\left(\frac{e\cdot c_S}{d_S}\right)}$ in the bound, for $n\,\geq \frac{d_S C^2}{r_0^2\,\log\left(\frac{e\cdot c_S}{d_S}\right)}$, to obtain
\begin{align*}
    &\E \bigintssss \mathcal{D}_S(P_\theta,P_0) \dt Q_n^\dagger(\theta) \\& \le\mathcal{D}_S(P_{\theta^\star},P_0)  +\frac{C\sqrt{d_S}}{n^{1/2}} \frac{\log\left(n^{1/2}\,\dfrac{e\cdot c_S}{C\sqrt{d_S}}\sqrt{\log\left(\frac{e\cdot c_S}{d_S}\right)}\right)}{\sqrt{\log\left(\frac{e\cdot c_S}{d_S}\right)}}  + \frac{C\sqrt{d_S}}{n^{1/2}}\sqrt{\log\left(\frac{e\cdot c_S}{d_S}\right)} \, ,
\end{align*}
and by convexity, Assumption \ref{ass:convex} implies that  
$
\E [ \mathcal{D}_S(P_{Q_n^\dagger},P_0)]\le \E \bigintssss \mathcal{D}_S(P_\theta,P_0) \dt Q_n^\dagger(\theta),
$
which ends the proof.
\end{proof}
\subsection{Existence and uniqueness of the minimiser}\label{sec:existence-uniqueness}
It is interesting to investigate the sufficient conditions under which the objective of~\eqref{eq:functional-for-WGF} admits a unique minimiser.
We posit two assumptions in order to do so:
one on the score and the other on the prior.
\begin{lemma}
    Suppose that Assumption~\ref{ass:convex} holds and further that the score is bounded from below so that $\inf_{Q\in\mathcal{P}(\Theta)} \mathscr{L}(Q) = M > -\infty$.
    Assume further that the prior $\Pi$ admits density with respect to the Lebesgue measure $\dt\Pi(\theta)\propto \exp\{-U(\theta)\}$ for $U:\Theta\to\mathbb{R}$ such that:
    \begin{enumerate}[(i),nosep]
        \item there exist constants $c > 0$ and $c^\prime \in \mathbb{R}$ such that for all $\theta\in\Theta$,
        \begin{equation*}
            \nabla_\theta U(\theta)\cdot\theta \geq c\lVert\theta\rVert^2 + c^\prime,
        \end{equation*} 
        and,\label{ass:prior-coercivity}
        \item $\nabla_\theta U$ is Lipschitz continuous.\label{ass:prior-lipz}
    \end{enumerate}
    Then there exists a unique measure $Q_n\in\mathcal{P}(\Theta)$ such that $Q_n = \arginf_{Q\in\mathcal{P}(\Theta)}\mathscr{L}$, and further this $Q_n$ is characterised by the implicit equation
    \begin{equation*}
        \dt Q_n(\theta)\propto \dt \Pi(\theta) \cdot \exp\left\{\left. - \frac{\lambda_n}{n} \sum_{i=1}^n \frac{\delta S(P_Q, x_i)}{\delta Q}\right|_{Q = Q_n}(\theta) \right\}.
    \end{equation*}
\end{lemma}
\begin{proof}
    We have by Assumption~\ref{ass:convex} that the score $S(\cdot,\,x)$ is convex is its first argument.
    In combination with the fact that the mapping $Q\mapsto P_Q$ is linear, we have the the mapping $Q\mapsto S(P_Q,\,x)$ is convex.
    The existence and uniqueness of a minimiser is then provided under the same conditions in Proposition 2.5 of \citet{hu_mean-field_2020}, and we point the reader to their Section 4 for a proof.
    
    To characterise the form of this minimiser, we take the functional derivative of~\eqref{eq:functional-for-WGF} with respect to $Q$, set it equal to zero, and solve for $Q$.
    We begin by noting that the functional derivative of the KL regulariser evaluated at $\theta$ is given by
    \begin{equation}\label{eq:func-deriv-regulariser}
        1 + \log\left\{\frac{\dt Q}{\dt\Pi}(\theta)\right\},
    \end{equation}
    while that of the score term is
    \begin{equation}\label{eq:func-deriv-score-wrt-Q} 
        \frac{\lambda_n}{n} \sum_{i=1}^n \frac{\delta S(P_Q, x_i)}{\delta Q}(\theta).
    \end{equation}
    By the linearity of functional derivatives, setting the sum of~\eqref{eq:func-deriv-regulariser} and~\eqref{eq:func-deriv-score-wrt-Q} equal to zero and removing constant terms, we find that $Q_n$ must satisfy
    \begin{equation*}
        0 = \frac{\lambda_n}{n} \sum_{i=1}^n \left.\frac{\delta S(P_Q, x_i)}{\delta Q}\right|_{Q = Q_n}(\theta) + \log\left\{\frac{\dt Q_n}{\dt\Pi}(\theta)\right\}
    \end{equation*}
    whence rearranging yields the stated result.
\end{proof}
The prior coercivity conditions of~\ref{ass:prior-coercivity} and~\ref{ass:prior-lipz} are satisfied, for instance, by choosing the prior to be a multivariate standard Gaussian prior, common to many applications \citep{hu_mean-field_2020}.
And the condition that the infimum be attained at non-infinite $M$ is immediate, for example, for the squared MMD under bounded kernels \citep[see also][Theorem 1]{shenprediction2025}.
In the case of the log score, we refer the reader to \citet[Proposition 1]{liu_detecting_2025} who establish the necessary conditions.

\section{Proofs of Main results}

\begin{proof}[Proof of Lemma \ref{lemma:pac-bayes-bound-master}]
The result depends on slightly different arguments depending on whether Assumption \ref{ass:entropy} or Assumption \ref{ass:Global} is satisfied. Therefore, we break the result into cases.

\noindent\textbf{Scores Satisfying Assumption \ref{ass:Global}.}
Consider any fixed integer $k>1$. Recall $Q^k = \bigotimes_{j=1}^kQ$. Apply the Donsker-Varadhan Lemma with $h(\theta_{1:k})=k\lambda_n\{L(\theta_{1:k})-L_n(\theta_{1:k})\}$, 
to see that, $P_0$-almost surely,  $\forall Q\in\mathcal{P}(\Theta)$, we have 
$$
\bigintssss k\lambda_n \left\{ L(\theta_{1:k})-L_n(\theta_{1:k}) \right\} \dt Q^k(\theta_{1:k}) \le \text{KL}(Q^k\|\Pi^k) + {\log} \bigintssss e^{k\lambda_n\{L(\theta_{1:k})-L_n(\theta_{1:k})\}} \dt \Pi^k(\theta_{1:k}) . 
$$
Rearranging terms, we have $P_0$-almost surely, $\forall Q\in\mathcal{P}(\Theta)$, 
$$
\bigintssss L(\theta_{1:k}) \dt Q^k(\theta_{1:k}) \le  \bigintssss L_n(\theta_{1:k}) \dt Q^k(\theta_{1:k}) + \frac{\text{KL}(Q^k\|\Pi^k)}{k\lambda_n} + \frac{ {\log}\bigintssss e^{k\lambda_n\{L(\theta_{1:k})-L_n(\theta_{1:k})\}} \dt \Pi^k(\theta_{1:k}) }{k\lambda_n} ,
$$
and $P_0$-almost surely, for $Q=Q_n$ in particular:
$$
\bigintssss L(\theta_{1:k}) \dt Q_n^k(\theta_{1:k}) \le  \bigintssss L_n(\theta_{1:k}) \dt Q_n^k(\theta_{1:k}) + \frac{\text{KL}(Q_n^k\|\Pi^k)}{k\lambda_n} + \frac{{\log}\bigintssss e^{k\lambda_n\{L(\theta_{1:k})-L_n(\theta_{1:k})\}} \dt \Pi^k(\theta_{1:k}) }{k\lambda_n} .
$$
Taking expectation with respect to $P_0$ on both sides and applying Jensen's inequality to the log then leads to:
\begin{align*}
    \E \bigintssss L(\theta_{1:k}) \dt Q_n^k(\theta_{1:k}) & \le \E \bigintssss L_n(\theta_{1:k})  \dt Q_n^k(\theta_{1:k}) + \frac{\E\text{KL}(Q_n^k\|\Pi^k)}{k\lambda_n} + \frac{\E\log  \bigintssss e^{k\lambda_n\{L(\theta_{1:k})-L_n(\theta_{1:k})\}} \dt \Pi^k(\theta_{1:k})}{k\lambda_n} \\
    & \le \E \bigintssss L_n(\theta_{1:k}) \dt Q_n^k(\theta_{1:k}) + \frac{\E\text{KL}(Q_n^k\|\Pi^k)}{k\lambda_n} + \frac{\log \E \bigintssss e^{k\lambda_n\{L(\theta_{1:k})-L_n(\theta_{1:k})\}} \dt \Pi^k(\theta_{1:k})}{k\lambda_n} .
\end{align*}
The last term on the right-hand-side is controlled using Tonelli's theorem and Assumption \ref{ass:Global}(ii), under which:
$$
\log \E \bigintssss e^{k\lambda_n\{L(\theta_{1:k})-L_n(\theta_{1:k})\}} \dt \Pi^k(\theta_{1:k}) = \log \bigintssss \E e^{k\lambda_n\{L(\theta_{1:k})-L_n(\theta_{1:k})\}} \dt \Pi^k(\theta_{1:k}) \leq \frac{k^2\lambda_n^2C_{}^2}{n} ,
$$
which then gives
\begin{align*}
    \E \bigintssss L(\theta_{1:k}) \dt Q_n^k(\theta_{1:k}) & \le \E \bigintssss L_n(\theta_{1:k})  \dt Q_n^k(\theta_{1:k}) + \frac{\E\text{KL}(Q_n^k\|\Pi^k)}{k\lambda_n} + \frac{\lambda_nkC_{}^2}{n} \\
    & = \E\left[ \bigintssss L_n(\theta_{1:k})  \dt Q_n^k(\theta_{1:k}) + \frac{\text{KL}(Q_n^k\|\Pi^k)}{k\lambda_n}\right] + \frac{\lambda_nkC_{}^2}{n} .
\end{align*}
Recall that $Q_n$ is the minimiser of the quantity inside the expectation on the right-hand-side to finally write:
$$
\E \bigintssss L(\theta_{1:k}) \dt Q_n^k(\theta_{1:k}) \le \E\left[ \inf_{Q\in\mathcal{P}(\Theta)}\left\{ \bigintssss L_n(\theta_{1:k})  \dt Q^k(\theta_{1:k}) + \frac{\text{KL}(Q^k\|\Pi^k)}{k\lambda_n} \right\} \right] + \frac{\lambda_nkC_{}^2}{n} .
$$

Now, we lower bound $\bigintssss L(\theta_{1:k}) \dt Q_n^k(\theta_{1:k}) $ on the left-hand-side: using the convexity in Assumption \ref{ass:convex} and the structure of the loss in Assumption \ref{ass:Global}(i), we have 
\begin{flalign*}
\mathcal{S}\left(\int P_\theta \dt Q_n(\theta),P_0 \right)&=\int \mathcal{S}(P_\theta,P_0)\dt Q_n(\theta)-\mathbb{E}_{}\{\Delta(Q_n,X)\}\\&\le \int \mathcal{S}(P_\theta,P_0)\dt Q_n(\theta)-\int\delta(\theta_{1:k})\dt Q^k(\theta_{1:k})\equiv \int L(\theta_{1:k})\dt Q_n(\theta_{1:k}),
\end{flalign*} 
where $\delta(\theta_{1:k})=\E_{}[\delta(\theta_{1:k},X)]$.
Recalling  that $P_{Q_n}=\int P_\theta \dt Q_n(\theta)$, this gives us
\begin{flalign}
\E\left[\mathcal{S}\left(P_{Q_n},P_0 \right)\right] \le \E\inf_{Q\in\mathcal{P}(\Theta)}\left\{ \bigintssss L_n(\theta_{1:k})  \dt Q^k(\theta_{1:k}) + \frac{\text{KL}(Q^k\|\Pi^k)}{k\lambda_n} \right\}  + \frac{\lambda_nkC_{}^2}{n} \label{eq:new2}.
\end{flalign}
Using Assumption \ref{ass:Global}(i), we can rewrite equation \eqref{eq:new2} as 
\begin{flalign*}
\E_{} \left\{\mathcal{S}\left(P_{Q_n},P_0\right)\right\} & \leq \E\left[ \inf_{Q\in\mathcal{P}(\Theta)}\left\{ \bigintssss L_n(\theta_{1:k})  \dt Q^k(\theta_{1:k}) + \frac{\text{KL}(Q^k\|\Pi^k)}{k\lambda_n} \right\} \right] + \frac{\lambda_nkC^2}{n} \\
& = \E\left[ \inf_{Q\in\mathcal{P}(\Theta)}\left\{S_n(P_Q) + \frac{\text{KL}(Q^k\|\Pi^k)}{k\lambda_n} \right\} \right] + \frac{\lambda_nkC^2}{n} 
\\
& \leq \inf_{Q\in\mathcal{P}(\Theta)}\left\{ \mathcal{S}(P_Q,P_0) + \frac{\text{KL}(Q^k\|\Pi^k)}{k\lambda_n} \right\} + \frac{\lambda_nkC_{}^2}{n} .
\end{flalign*}
Recalling that $\mathcal{S}(P_0,P_0)=\E_{}S_n(P_0)$, we subtract $\mathcal{S}(P_0,P_0)$ from both sides and write 
\begin{flalign*}
\E\{\mathcal{D}_S(P_{Q_n},P_0)\}&=\E_{}\left\{\mathcal{S}\left(P_{Q_n},P_0\right)-\mathcal{S}(P_0,P_0)\right\} \\
& \le \inf_{Q\in\mathcal{P}(\Theta)}\left\{ \mathcal{S}(P_0,P_0) - \mathcal{S}(P_Q,P_0) + \frac{\text{KL}(Q^k\|\Pi^k)}{k\lambda_n} \right\} + \frac{\lambda_nkC_{}^2}{n}  \\
& = \inf_{Q\in\mathcal{P}(\Theta)}\left\{ \mathcal{D}(P_Q,P_0) + \frac{\text{KL}(Q^k\|\Pi^k)}{k\lambda_n} \right\} + \frac{\lambda_nkC_{}^2}{n}\\&=  \inf_{Q\in\mathcal{P}(\Theta)}\left\{ \mathcal{D}(P_Q,P_0) + \frac{\text{KL}(Q\|\Pi)}{\lambda_n} \right\} + \frac{\lambda_nkC_{}^2}{n}.
\end{flalign*}

\noindent\textbf{Scores Satisfying Assumption \ref{ass:entropy}.}
 For scores that do not satisfy Assumption \ref{ass:Global},  the above proof technique breaks down. We can, however, derive a similar bound using the variational definition of $Q_n$. 

From the definition of $\mathcal{D}(P_Q,P_0)$, we have 
\begin{flalign}
\E[\mathcal{D}_S(P_{Q_n},P_0)]&=\E[\mathcal{S}(P_{Q_n},P_0)]-\mathcal{S}(P_0,P_0)\\&=\E[{S}_n(P_{Q_n})]-\mathcal{S}(P_0,P_0)-\left\{\E[{S}_n(P_{Q_n})]-\E[\mathcal{S}(P_{Q_n},P_0)]\right\} \nonumber\\&\le \E[{S}_n(P_{Q_n})]-\mathcal{S}(P_0,P_0)+\E\sup_{Q\in\mathcal{P}_{\Pi}(\Theta)}|S_n(P_Q)- \mathcal{S}(P_{Q},P_0)  |\nonumber\\&\le \E[{S}_n(P_{Q_n})]-\mathcal{S}(P_0,P_0)+r_n,\label{eq:vary1}
\end{flalign}
where the last line follows from Assumption \ref{ass:entropy}. Using the variational definition of $Q_n$, we have 
$$
S_n(P_{Q_n})\le S_n(P_{Q_n})+\frac{\KL(Q_n\|\Pi)}{\lambda_n}=\inf_{Q\in\mathcal{P}(\Theta)}\left\{S_n(P_Q)+\frac{\KL(Q\|\Pi)}{\lambda_n}\right\}.
$$
Taking expectations then yields
\begin{flalign}
\E\left\{S_n(P_{Q_n^{}})\right\}&\le\E\left[\inf_{Q\in\mathcal{P}(\Theta)}\left\{S_n(P_{Q})+\frac{\KL(Q\|\Pi)}{\lambda_n}\right\}\right]\le   \inf_{Q\in\mathcal{P}(\Theta)}
\left\{\mathcal{S}(P_{Q},P_0)+\frac{\KL(Q\|\Pi)}{\lambda_n}\right\}.\label{eq:vary2}
\end{flalign}
Plugging equation \eqref{eq:vary2} into \eqref{eq:vary1} yields
\begin{flalign*}
\E[\mathcal{D}_S(P_{Q_n},P_0)]&=\E[\mathcal{S}(P_{Q_n},P_0)]-\mathcal{S}(P_0,P_0)\\&\le   \inf_{Q\in\mathcal{P}(\Theta)}
\left\{\mathcal{S}(P_{Q},P_0)+\frac{\KL(Q\|\Pi)}{\lambda_n}\right\}-\mathcal{S}(P_0,P_0)+r_n\\&=\inf_{Q\in\mathcal{P}(\Theta)}
\left\{\mathcal{D}_S(P_{Q},P_0)+\frac{\KL(Q\|\Pi)}{\lambda_n}\right\}+r_n.
\end{flalign*}The stated result now follows from the definition of $\nu_n$.
\end{proof}

\begin{proof}[Proof of Theorem \ref{theorem:Gibbs-vs-pro-posterior-misspecification}]
From the convexity in Assumption \ref{ass:convex}, 
$$
\mathcal{D}_S(P_Q,P_0)\le \int \mathcal{D}_S(P_\theta,P_0)\dt Q(\theta).
$$Applying bound in the RHS of Lemma \ref{lemma:pac-bayes-bound-master} delivers
\begin{flalign}
\E[\mathcal{D}_S(P_{Q_n},P_0)]&\le 
\inf_{Q\in\mathcal{P}(\Theta)}
\left\{\mathcal{D}_S(P_{Q},P_0)+\frac{\KL(Q\|\Pi)}{\lambda_n}\right\}+\nu_n\nonumber\\&\le \inf_{Q\in\mathcal{P}(\Theta)}
\left\{\int \mathcal{D}_S(P_\theta,P_0)\dt Q(\theta)+\frac{\KL(Q\|\Pi)}{\lambda_n}\right\}+\nu_n\label{eq:new_new1}.
\end{flalign}
Recalling the definition of $\mathcal{B}_r$ in Assumption \ref{ass:prior-mass-condition}, define the candidate density 
$$
\dt \Pi_{\mathcal{B}_r}(\theta):=\frac{1(\theta\in\mathcal{B}_r)\dt \Pi(\theta)}{\Pi(\mathcal{B}_r)}.
$$
Now, the RHS of the inequality in \eqref{eq:new_new1} is true for any $Q\in\mathcal{P}(\Theta)$, and remains true for  $Q=\Pi_{\mathcal{B}_n}$. 
For $0\le r\le r_0$,  over $\mathcal{B}_r$, $\mathcal{D}_S(P_\theta,P_0)\le \mathcal{D}_S(P_{\theta^\star},P_0)+r$, so that we obtain
\begin{flalign*}
\E\left[\mathcal{D}_S(P_{Q_n},P_0)\right]  &\le \inf_{r>0} \left\{\int\mathcal{D}_S(P_\theta,P_0) \dt\Pi_{\mathcal{B}_r}(\theta)+\frac{\KL(\Pi_{\mathcal{B}_r}^k\|\Pi^k)}{k\lambda_n}\right\}+\nu_n,\\&= \inf_{r>0} \left\{\int\mathcal{D}_S(P_\theta,P_0) \dt\Pi_{\mathcal{B}_r}(\theta)+\frac{\KL(\Pi_{\mathcal{B}_r}\|\Pi)}{\lambda_n}\right\}+\nu_n\\&= \inf_{r>0} \left\{\int\mathcal{D}_S(P_\theta,P_0) \dt\Pi_{\mathcal{B}_r}(\theta)+\frac{\log\left\{1/\Pi(\mathcal{B}_r)\right\}}{\lambda_n}\right\}+\nu_n,
\\
&\le \inf_{r_0\ge r>0} \left\{\int\left\{\mathcal{D}_S(P_{\theta^\star},P_0)+r\right\} \dt\Pi_{\mathcal{B}_r}(\theta)+\frac{d_S\log (c_S/r)}{\lambda_n}\right\}+\nu_n\\
&=\mathcal{D}_S(P_{\theta^\star},P_0)+\inf_{r_0\ge r>0} \left\{r+\frac{d_S\log (c_S/r)}{\lambda_n}\right\}+\nu_n \, .
\end{flalign*}
The remainder of the proof follows the same as in the proof of Lemma \ref{lem:gibbs}.

A similar argument holds for $Q_n^\dagger$.

\end{proof}

\begin{proof}[Proof of Theorem \ref{theorem:generalisation-guarantee-exact}]
Assumption \ref{ass:non-trivial-misspecification-exact} ensures the existence of some $Q^\star\in\arg\inf_{Q\in\mathcal{P}(\Theta)} \mathcal{D}_S\left(P_Q,P_0\right)$ such that
$$
{\mathsf{JenGap}}:= \mathcal{D}_S(P_{\theta^\star},P_0) - \mathcal{D}_S\left(P_{Q^\star},P_0\right)> 0.
$$
From Lemma \ref{lemma:pac-bayes-bound-master},
 we can obtain:
\begin{flalign}
    \mathbb{E} \left[\mathcal{D}_S\left(\int P_\theta \, \dt Q_n(\theta),P_0 \right)\right]
    & \le \inf_{Q\in\mathcal{P}(\Theta)} \left\{ \mathcal{D}_S\left(\int P_\theta\dt Q(\theta),P_0\right) + \frac{\text{KL}(Q\|\Pi)}{\lambda_n} \right\} + \nu_{n} \nonumber\\
    & = \inf_{Q\ll\Pi} \left\{ \mathcal{D}_S\left(\int P_\theta\dt Q(\theta),P_0\right) + \frac{\text{KL}(Q\|\Pi)}{\lambda_n} \right\} + \nu_{n} \nonumber\\
    & \leq \mathcal{D}_S\left(\int P_\theta \, \dt Q^\star(\theta),P_0 \right) + \frac{\textnormal{KL}({Q}^\star\|\Pi)}{\lambda_n} + \nu_{n}\label{eq:new_eq3}.
 \end{flalign}   
 Adding and subtracting $\mathcal{D}_S(P_{\theta^\star},P_0)$ and using the definition of $\mathsf{JenGap}$, we can re-arrange \eqref{eq:new_eq3} as
\begin{flalign*}
    \mathbb{E} \left[\mathcal{D}_S\left(\int P_\theta \, \dt Q_n(\theta),P_0 \right)\right]
    & \le
    \mathcal{D}_S(P_{\theta^{\star}},P_0) + \frac{\textnormal{KL}({Q}^\star\|\Pi)}{\lambda_n} + \nu_n - \mathsf{JenGap} .
\end{flalign*}
Since, by Assumption \ref{ass:non-trivial-misspecification-exact}, $Q^\star$ is such that $\KL(Q^\star\|\Pi)<\infty$, for $n$ large enough we have
$$
\frac{\textnormal{KL}({Q}^\star\|\Pi)}{\lambda_n} + \nu_n \leq \frac{1}{4}{\mathsf{JenGap}} \, ,
$$
so that
$$\mathbb{E} \left[\mathcal{D}_S\left(\int P_\theta \, \dt Q_n(\theta),P_0 \right)\right] \leq \mathcal{D}_S(P_{\theta^{\star}},P_0) - \frac{3}{4}{\mathsf{JenGap}} \, .
$$

{Furthermore, from the continuity of $\theta\mapsto P_\theta$ in Assumption \ref{ass:non-trivial-misspecification-exact}, and the concentration in Lemma \ref{lem:gibbs}, we have for $n$ large enough but finite,}
$$
\mathbb{E} \left[\mathcal{D}_S\left(\int P_\theta \, \dt Q_n^\dagger(\theta),P_0 \right)\right] \geq \mathcal{D}_S(P_{\theta^{\star}},P_0) - \frac{\mathsf{JenGap}}{4} \, .
$$

Therefore, for $n$ large enough but finite,
\begin{align*}
    \mathbb{E} \left[\mathcal{D}_S\left(\int P_\theta \, \dt Q_n(\theta),P_0 \right)\right] & \leq  \mathcal{D}_S(P_{\theta^{\star}},P_0) - \frac{3}{4}{\mathsf{JenGap}} \\
    & \leq \mathbb{E} \left[\mathcal{D}_S\left(\int P_\theta \, \dt Q_n^\dagger(\theta),P_0 \right)\right] - \frac{1}{2}{\mathsf{JenGap}} \\
    & < \mathbb{E} \left[\mathcal{D}_S\left(\int P_\theta \, \dt Q_n^\dagger(\theta),P_0 \right)\right] . 
\end{align*}

To obtain the second result, return to equation \eqref{eq:new_eq3}, and use the fact that $Q^\star\in\arginf_{Q\in\mathcal{P}(\Theta)}\mathcal{D}_S(P_Q,P_0)$ satisfies $\KL(Q^\star\|\Pi)<\infty$, by Assumption \ref{ass:non-trivial-misspecification-exact}, to obtain
\begin{flalign*}
        \E_{}\left[\mathcal{D}_S\left(\int P_\theta\dt Q_n(\theta),P_0\right)\right] &\le\inf_{Q\in\mathcal{P}(\Theta)} \left\{ \mathcal{D}_S\left(\int P_\theta\dt Q(\theta),P_0\right) + \frac{\text{KL}(Q\|\Pi)}{\lambda_n} \right\} +\nu_{n}\\&\le \mathcal{D}_S(P_{Q^\star},P_0)+ \frac{\text{KL}(Q^\star\|\Pi)}{\lambda_n}+ \nu_n\\&\le \mathcal{D}_S(P_{Q^\star},P_0)+\frac{C}{\lambda_n}+\nu_n.
\end{flalign*}In the case of Assumption \ref{ass:Global}, $\nu_n=\log(n)/{n}^{1/2}$ so that choosing $\lambda_n\propto n^{1/2}/\log(n)$ yields the stated result. In the case of Assumption \ref{ass:entropy}, $\nu_n=\min\{\log(n)/{n}^{1/2},r_n\}$, and again choosing $\lambda_n\propto {n}^{1/2}\log(n)$ also implies the stated result.
\end{proof}

\begin{proof}[Proof of Theorem \ref{theorem:generalisation-guarantee-misspecification-mixture-exact-convex}]
We break the result up into cases depending on whether (i) or (ii) in Assumption \ref{ass:prior-mass-convex-recovery-general} is satisfied. 

\noindent\textbf{Assumption \ref{ass:prior-mass-convex-recovery-general}(i).}
From equation \eqref{eq:new_eq3} in the proof of Theorem \ref{theorem:generalisation-guarantee-exact}, we have that    
\begin{flalign*}
\E_{}\left[\mathcal{D}_S\left(P_{Q_n},P_0\right)\right]&\le\inf_{Q\in\mathcal{P}(\Theta)} \left\{ \mathcal{D}_S\left(\int P_\theta\dt Q(\theta),P_0\right) + \frac{\text{KL}(Q\|\Pi)}{\lambda_n} \right\} + \nu_n\\&\le \inf_{Q\ll\Pi} \left\{ \mathcal{D}_S\left(\int P_\theta\dt Q(\theta),P_0\right) + \frac{\text{KL}(Q\|\Pi)}{\lambda_n} \right\} + \nu_n\\&=\mathcal{D}_S(P_{Q^\star},P_0)+ \frac{C}{\lambda_n} + \nu_n,
\end{flalign*}where the equality follows from the fact that $\KL(Q^\star\|\Pi)<\infty$, and the last inequality follows by taking $\lambda_n\propto {n}^{1/2}/\log(n)$ so that $C/\lambda_n+\nu_n\le C'\cdot\nu_n$, for $n$ large enough. We remind the reader that $\mathcal{D}_S(P_{Q^\star},P_0)=0$ under Assumption \ref{ass:prior-mass-convex-recovery-general}(i).

Since $\mathcal{D}_S(P_{Q^\star},P_0)=0$ under convex recovery, we have 
$$
{\mathsf{JenGap}}=\mathcal{D}_S(P_{\theta^\star},P_0)>0,
$$ 
and there exist an $n$ large enough so that 
\begin{flalign*}
\E_{}\left[\mathcal{D}_S\left(P_{Q_n},P_0\right)\right] \le C'\nu_n \le \frac{\mathcal{D}_S(P_{\theta^\star},P_0)}{2} .
\end{flalign*}
{Further, for $n$ large enough, as in the proof of Theorem \ref{theorem:generalisation-guarantee-exact}, }
$
\E_{}\left[\mathcal{D}_S\left(P_{Q_n^\dagger},P_0\right)\right] > \frac{\mathcal{D}_S(P_{\theta^\star},P_0)}{2}.
$
Putting these inequalities together, we see that, for $n$ large enough, 
\begin{flalign*}
\E_{}\left[\mathcal{D}_S\left(P_{Q_n},P_0\right)\right] \le \frac{\mathcal{D}_S(P_{\theta^\star},P_0)}{2} < \E_{}\left[\mathcal{D}_S\left(P_{Q_n^\dagger},P_0\right)\right].
\end{flalign*}

\noindent\textbf{Assumption \ref{ass:prior-mass-convex-recovery-general}(ii).} Recall $\mathcal{B}_{r,j} = \{\theta:\mathcal{D}_S(P_{\theta},P_{\theta_j^\star})\leq r_{}\}$, denote $\dt Q_r^\star(\theta)=\sum_{j=1}^{K} \omega_j\dt \Pi^\star_{r,j}(\theta)$ where $\omega\in\Delta^{K-1}$,  and
$$
\dt \Pi^\star_{r,j}(\theta)=\begin{cases}
\frac{\dt\Pi(\theta)}{\Pi(\mathcal{B}_{r,j})}&\text{ if }\theta\in\mathcal{B}_{r,j}\\0&\text{ else }	
\end{cases}.
$$
Using $Q^\star_r$ within equation \eqref{eq:new_eq3} in the proof of Theorem \ref{theorem:generalisation-guarantee-exact}, yields
\begin{flalign*}
&\E_{}\left[\mathcal{D}_S\left(P_{Q_n},P_0\right)\right]\\&\le\inf_{Q\in\mathcal{P}(\Theta)} \left\{ \mathcal{D}_S\left(\int P_\theta\dt Q(\theta),P_0\right) + \frac{\text{KL}(Q\|\Pi)}{\lambda_n} \right\} + \frac{\lambda_nkC_k^2}{n}\\
& \leq \inf_{r>0}\left\{\mathcal{D}_S\left(\int P_\theta \, \dt Q^\star_r(\theta),P_0 \right) + \frac{\textnormal{KL}(Q^\star_r\|\Pi)}{\lambda_n}\right\} + \nu_n \\
& = \inf_{r>0}\left\{\mathcal{D}_S\left(\sum_{j=1}^{K} \omega_j \int P_\theta \, \dt \Pi^\star_{r,j}(\theta),\sum_{j=1}^K \omega_j P_{\theta_j^\star} \right) + \frac{\textnormal{KL}\left(\sum_{j=1}^{K} \omega_j \Pi^\star_{r,j}\|\Pi\right)}{\lambda_n} \right\} + \nu_n\\
& \leq \inf_{r>0}\left\{\sum_{j=1}^K \omega_j \, \int \, \mathcal{D}_S\left( P_\theta , P_{\theta_j^\star}\right) \dt \Pi^\star_{r,j}(\theta) + \frac{\sum_{j=1}^{K} \omega_j \textnormal{KL}\left(\Pi^\star_{r,j}\|\Pi\right)}{\lambda_n} \right\}+ \nu_n \\
& \leq \inf_{r>0}\left\{\sum_{j=1}^K \omega_j \, \int \, \mathcal{D}_S\left( P_\theta , P_{\theta_j^\star}\right) \dt \Pi^\star_{r,j}(\theta) + \frac{\sum_{j=1}^{K} \omega_j \log\{1/\Pi(\mathcal{B}_{r,j})\}}{\lambda_n} \right\}+ \nu_n
\\
& \leq \inf_{r>0}\left\{\sum_{j=1}^K \omega_j \, \int \, \mathcal{D}_S\left( P_\theta , P_{\theta_j^\star}\right) \dt \Pi^\star_{r,j}(\theta) + \frac{d_s\log(c_s/r)}{\lambda_n}\right\} + \nu_n \\
& \leq \inf_{r>0}\left\{r+ \frac{d_s\log(c_s/r)}{\lambda_n}\right\}  + \nu_n.
\end{flalign*}Optimising $r$ using similar arguments to the end of Lemma \ref{lem:gibbs} yields the stated result.
\end{proof}

\begin{proof}[Proof of Corollary \ref{corollary:concentration-if-model-correct}]
The proof follows directly from Theorem \ref{theorem:generalisation-guarantee-misspecification-mixture-exact-convex}; in particular,
$$
\E\left[\mathcal{D}_S\left(P_{Q_n},P_0\right)\right]-\E[\mathcal{D}_S(P_{Q^\dagger_n},P_0)] \le \E\left[\mathcal{D}_S\left(P_{Q_n},P_0\right)\right] \lesssim \nu_n \, .
$$
\end{proof}

\begin{proof}[Proof of Corollary \ref{corollary:concentration-dtf}]
    Begin by noting that by Assumption \ref{ass:concentration-dtf}(i) there exists $Q^\star\in\mathcal{P}(\Theta)$ such that $P_{Q^\star} = P_0$, so that $\mathsf{d}_2(P_{Q^\star},\,P_0)=0$.
    Then, by Assumption \ref{ass:concentration-dtf}(iii) and Jensen's inequality, 
    \begin{IEEEeqnarray}{rCl}
        \E\{\mathsf{d}_2(P_{Q_n},\,P_{Q^\star})\} 
        &\leq&
        K^{1/\alpha}\E\{d_{\mathcal{S}}(P_{Q_n},\,P_{Q^\star})^{1/\alpha}\} \nonumber \\
        &\leq&
        K^{1/\alpha}\left[\E\{d_{\mathcal{S}}(P_{Q_n},\,P_{Q^\star})\}\right]^{1/\alpha}, \label{eq:posterior-jensen}
    \end{IEEEeqnarray}
    where the second inequality follows from the concavity of the map $x\mapsto x^{1 / \alpha}$ for $\alpha \ge 1$.
    For a sequence $M_n\rightarrow\infty$ such that $M_n\nu_n^{1/\alpha}=o(1)$, by Markov's inequality
    \begin{flalign*}
        P_0\left\{\mathsf{d}_2(P_{Q_n},\,P_{Q^\star})>M_n\nu_n^{1/\alpha}\right\}
        &\le \frac{\E\{\mathsf{d}_2(P_{Q_n},\,P_{Q^\star})\}}{M_n\nu_n^{1/\alpha}} \\
        &\le \frac{K^{1/\alpha}\left[\E\{d_{\mathcal{S}}(P_{Q_n},\,P_{Q^\star})\}\right]^{1 / \alpha}}{M_n\nu_n^{1/\alpha}} 
        \lesssim \frac{1}{M_n},
    \end{flalign*}
    where the second inequality follows from~\eqref{eq:posterior-jensen}, and the final inequality from Theorem \ref{theorem:generalisation-guarantee-misspecification-mixture-exact-convex}. 
    Since $M_n\rightarrow\infty$, taking the complementary event yields 
    \begin{equation*}
        P_0\left\{\mathsf{d}_2(P_{Q_n},\,P_{Q^\star})\leq M_n\nu_n^{1/\alpha}\right\}\ge 1-o(1).
    \end{equation*}
    Since $M_n\nu_n^{1/\alpha}=o(1)$, for all $n$ sufficiently large we have $M_n\nu_n^{1/\alpha}\leq\varepsilon$, where $\varepsilon$ is the constant in Assumption \ref{ass:concentration-dtf}(ii).
    Then, by Assumption \ref{ass:concentration-dtf}(ii),
    \begin{equation*}
        P_0\{\mathsf{d}_1({Q_n},\,{Q^\star})\leq CM_n\nu_n^{1/\alpha}\}
        \geq 
        P_0\{\mathsf{d}_2(P_{Q_n},\,P_{Q^\star})\leq M_n\nu_n^{1/\alpha}\}
        \geq 1-o(1).
    \end{equation*}
    Since $M_n$ is arbitrary up to a fixed multiplicative constant, the stated result follows.
\end{proof}

\begin{proof}[Proof of Corollary \ref{corollary:pointwise-coverage} ]
    By (CR) of Theorem \ref{thm:DI-master} and Markov's Inequality, for any $\varepsilon> 0$, it holds that 
    $P_0^{(n)}(\mathcal{D}_S(P_{Q_n}, P_0) \geq \varepsilon) \leq \frac{\nu_n}{\varepsilon}$.
    As $\mathcal{D}_S$ metrises weak convergence, this means that 
    $P_{Q_n}$ weakly converges to $P_0$ in probability.
    The result now follows by the Portmanteau Lemma.
\end{proof}

\begin{proof}[Proof of Theorem \ref{thm:calibration} ]
Conditionally on the observed data, the definition of total
variation distance gives
\[
\left|P_0(C_n)-P_{Q_n}(C_n)\right|
\leq
\mathrm{TV}(P_{Q_n},P_0).
\]
Taking expectations and using $P_{Q_n}(C_n)=1-\alpha$ almost surely yields
\[
\left|\E[P_0(C_n)]-(1-\alpha)\right|
\leq
\E[\mathrm{TV}(P_{Q_n},P_0)].
\]
The right-hand side converges to zero {by using Pinsker's inequality, then applying Jensen's inequality, and lastly applying the conclusion of Theorem \ref{thm:DI-master}.}
\end{proof}

\section{Further experimental details}
\subsection{Normal location illustrations}\label{app:init-comparison-posteirors}
We simulate $n = 1,000$ data from each of the data-generating processes:
\begin{IEEEeqnarray*}{rClr}
    y_i &\sim& \mathsf{N}(0, 1^2) & \quad\text{(``Well-specified'')} \\
    y_i &\sim& \begin{cases}
        \mathsf{N}(-2, 1^2),&\quad\text{with probability }0.2 \\
        \mathsf{N}(2, 1^2),&\quad\text{with probability }0.8
    \end{cases} & \quad\text{(``Mixture'')} \\
    y_i &\sim& \begin{cases}
        \mathsf{N}(-2, 1^2),&\quad\text{with probability }1/3 \\
        \mathsf{N}(0, 1^2),&\quad\text{with probability }1/3 \\
        \mathsf{N}(2, 1^2),&\quad\text{with probability }1/3
    \end{cases} & \quad\text{(``Claw'')} \\
    y_i &\sim& \mathsf{Student}-t_{1.5} & \quad\text{(``Heavy-tails'')}.
\end{IEEEeqnarray*}
For each of these, we estimate the MMD with $m=128$ Monte Carlo samples, run the Wasserstein gradient flow for $p = 32$ particles for all data, and for $4,000$ iterations similarly to MALA.
We use $\lambda_n = n$.
We show the posterior distributions of the Gibbs and \ac{pro} posteriors and the Wasserstein gradient flow trajectories in Figure~\ref{fig:initial-comparison-trajectories}, and in Figure~\ref{fig:mmd-illustrations} we show their induced predictive distributions.
\begin{figure}[t!]
    \centering
    \includegraphics[width=0.7\linewidth]{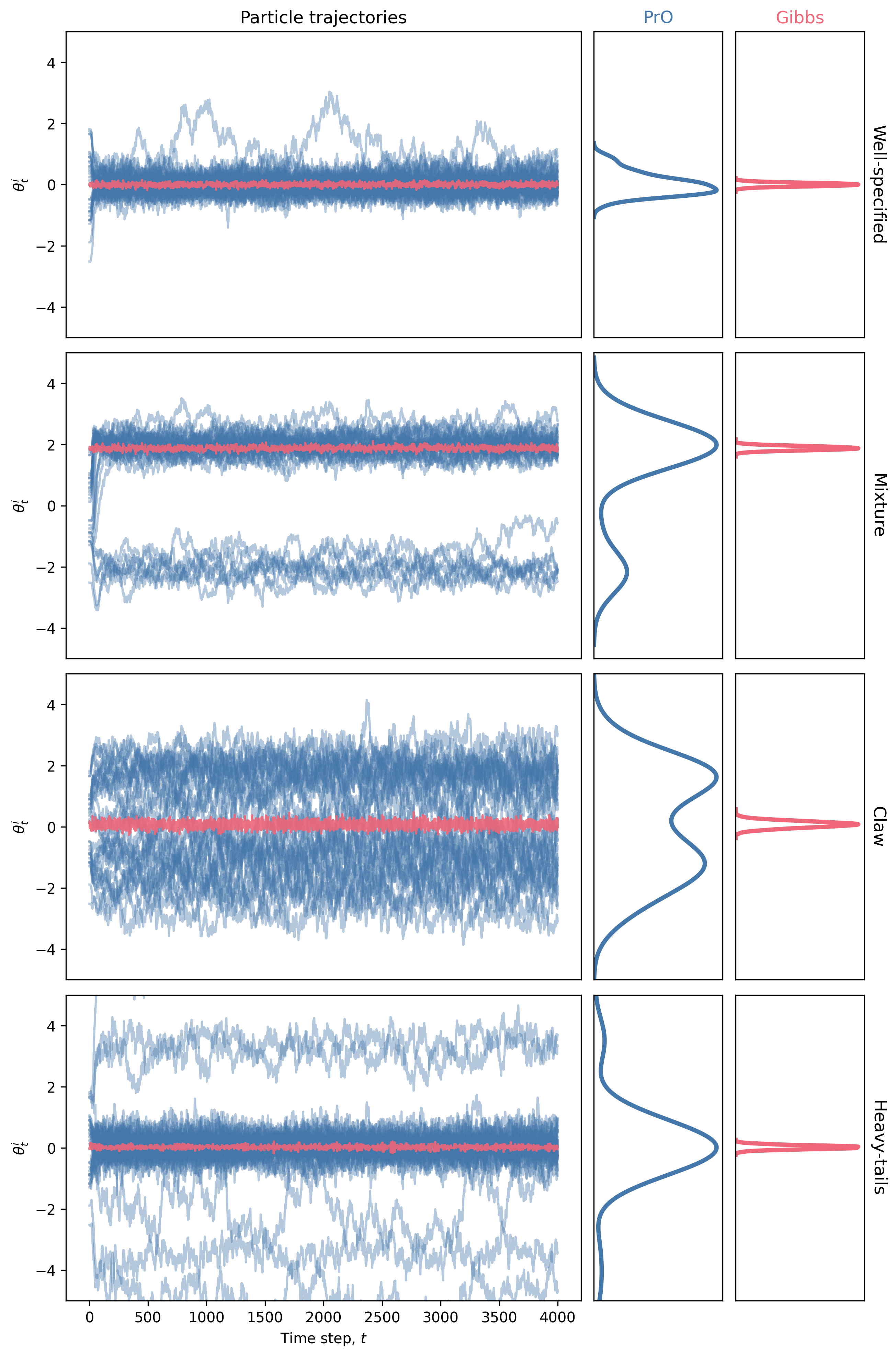}
    \caption{In the left-most column we show the particle trajectories from a Wasserstein gradient flow targeting the \textcolor{procolour}{\bfseries \ac{pro} posterior} compared to the path of a single MALA chain targeting the \textcolor{gibbscolour}{\bfseries Gibbs posterior}. In the two axes on the right we show the empirical posterior densities from these two trajectories.}
    \label{fig:initial-comparison-trajectories}
\end{figure}
\begin{figure}
    \centering
    \includegraphics[width=\linewidth]{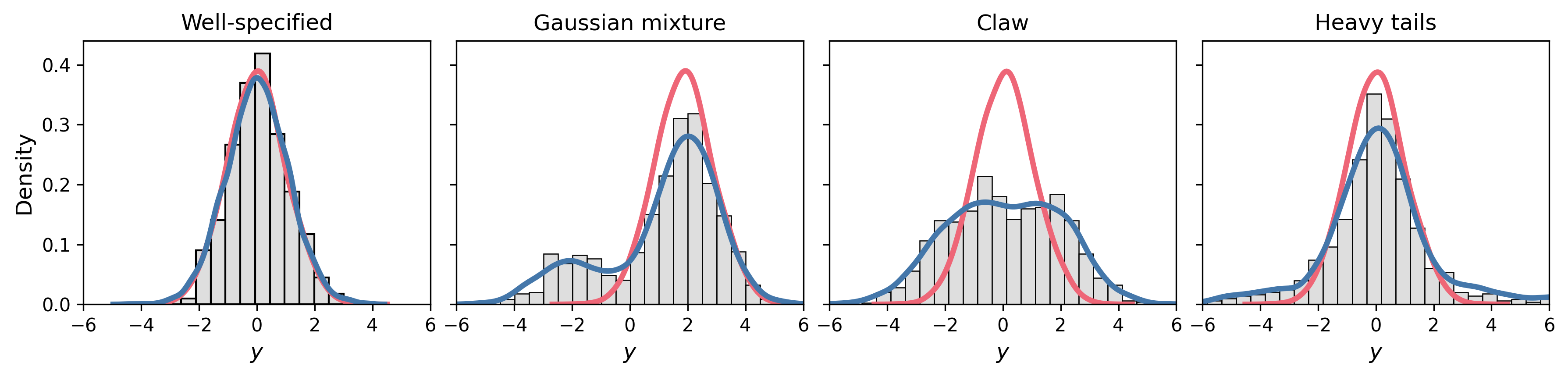}
    \caption{A comparison of \textcolor{procolour}{\bfseries \ac{pro}} and \textcolor{gibbscolour}{\bfseries Gibbs} posterior predictive distributions in a well-specified and three misspecified regimes. Grey bars indicate a histogram of the observed data.}
    \label{fig:mmd-illustrations}
\end{figure}
\subsection{Palmer penguins example}
\label{sec:penguins-supp}
The data containing $n = 342$ measurements of the penguins' bill lengths and depths in millimetres were imported from the \texttt{palmerpenguins} package in \textsf{Python}, and jointly centred to have zero-mean and scaled to have unit variance.
To these data we fit a Gaussian kernel density estimator implemented in \texttt{sklearn} \citep{scikit-learn}, where the bandwidth was chosen with cross-validation across a logarithmic grid.
We also fit for comparison a \ac{pro} posterior over the mean of an the isotropic bivariate Gaussian location model class, $\mathcal{M}_\Theta = \{\normal(\theta,\sigma^2I):\,\theta\in\mathbb{R}^2\}$, where $\sigma^2 = 0.2$ is fixed, and which we fit with the squared MMD score.
This same bandwidth was used for fitting the \ac{pro} posterior with the squared MMD as fitting the KDE.
We estimated the \ac{pro} posterior under the squared MMD loss approximated with $m=n$ samples, $\lambda_n = 10^{3}$, and the Wasserstein gradient flow was run with $p = 32$ particles for $10,000$ iterations.
\subsection{Regression with the MMD}
\label{appendix:reg-mmd}
Denote $d = \operatorname{dim}(\Theta)$.
In Sections~\ref{sec:misspecification-regimes},~\ref{sec:gmrf}, and~\ref{sec:redshift} we fit the linear models under the squared MMD loss.
Denote the model class as $\{\normal(y_i;\,x_i^\top\theta, \sigma^2):\,\theta\in\mathbb{R}^d\}$, for $i = 1,\ldots,n$ where $\sigma^2>0$ is treated as fixed and estimated from data:
\begin{equation*}
    \hat\sigma^2  = \sum_{i = 1}^n \frac{(y_i - \hat y_i)}{n - d - 1};\quad\hat y = X\hat\beta;\quad\hat\beta = (X^\top X)^{-1} X^\top y.
\end{equation*}
In the case of the Gibbs posterior, we use the efficient estimator proposed by \citet[Equation 6]{alquier_universal_2024}:
\begin{equation*}
    S_n(P_\theta) = \sum_{i=1}^n \ell^\dagger(\theta;\,X_i,y_i)
\end{equation*}
where for all $\theta\in\Theta,\,x\in\mathbb{R}^d,$ and $y\in\mathbb{R}$,
\begin{equation*}
    \ell^\dagger(\theta;\,X, y) = \mathbb{E}_{Y,Y^\prime\sim P_\theta(\cdot\mid x)}\{\kappa(Y, Y^\prime) - 2\kappa(Y, y)\}.
\end{equation*}
In practice we estimate both of these expectations with $m$ Monte Carlo draws:
\begin{equation*}
    \ell^\dagger(\theta;\,X, y) \approx \frac{1}{m (m - 1)}\sum_{i\ne j}\sum_{j=1}^m \kappa(Y_i, Y^\prime_j) - \frac{2}{m}\kappa(Y_i, y),
\end{equation*}
for $Y_i,Y^\prime_i \sim P_\theta(\cdot\mid x).$

In the case of the \ac{pro} posterior, we instead use the tandem loss
\begin{equation*}
    L_n(\vartheta, \theta) = \sum_{i=1}^n \ell(\vartheta,\theta;\,X_i,y_i)
\end{equation*}
where now
\begin{equation*}
    \ell(\vartheta,\theta;\,X,y) = \mathbb{E}_{\substack{Y\sim P_\vartheta(\cdot\mid x)\\Y^\prime\sim P_{\theta}(\cdot\mid x)}}\{\kappa(Y, Y^\prime) - \kappa(Y, y) - \kappa(Y^\prime, y)\}.
\end{equation*}
And again, each of these expectations can be approximated by Monte Carlo:
\begin{equation*}
    \ell(\vartheta,\theta;\,X, y) \approx \frac{1}{m^2}\sum_{i=1}^m\sum_{j=1}^m \kappa(Y_i, Y^\prime_j) - \frac{1}{m}\kappa(Y_i, y) - \frac{1}{m}\kappa(Y_i^\prime, y),
\end{equation*}
for $Y_i \sim P_\vartheta(\cdot\mid x)$ and $Y^\prime_i\sim P_{\theta}(\cdot\mid x)$.

In both cases, we use a Gaussian kernel where the length scale is chosen according to the median heuristic \citep{garreau_large_2018}.
\subsection{Linear regression example}
\label{sec:synthetic-regression-supp}
In Section~\ref{sec:misspecification-regimes} we generate $n = 500$ training data for each of the three misspecification regimes.
The log point-wise predictive density (lppd) is estimated from $S$ samples from the posterior $\{\theta^{(s)}\}_{s = 1}^S$ at a test point $\{X_i, y_i\}$ through
\begin{equation*}
    \operatorname{lppd}_i = \log \left\{\frac{1}{S} \sum_{s = 1}^S P_{\theta^{(s)}}(y_i\mid X_i)\right\}.
\end{equation*}
and similarly we compute the test point-wise $\operatorname{MMD}_i^2$ with posterior draws.
We subsequently compute the expected log predictive density (elpd) as $\operatorname{elpd}_{\operatorname{test}} = \sum_{i=1}^{n_{\mathrm{test}}} \operatorname{lppd}_i$, and likewise for $\operatorname{MMD}_{\operatorname{test}}^2$, for $n_{\mathrm{test}} = 1,000$.
In Figure~\ref{fig:linear-regression-pred-diff} we show the mean and standard error of the a normal approximation to the difference in these metrics inline with \citet[Equations 9 and 10]{sivula_uncertainty_2025}.
Denoting the point-wise difference as $\Delta_i = \operatorname{lppd}_i^{\text{\ac{pro}}} - \operatorname{lppd}_i^{\text{Gibbs}}$ in the case of the elpd, for example, the mean is computed as $\sum_{i = 1}^{n_{\mathrm{test}}} \Delta_i$, and the standard error as
\begin{equation*}
    \sqrt{\frac{n_{\mathrm{test}}}{n_{\mathrm{test}} - 1}\sum_{i=1}^{n_{\mathrm{test}}}\left\{\Delta_i - \frac{1}{n}\sum_{j=1}^{n_{\mathrm{test}}}\Delta_j\right\}^2}.
\end{equation*}
In Figure~\ref{fig:linear-regression-pred-diff} we rescale both metrics so that negative values indicate that the \ac{pro} posterior induces better predictive performance.
The continuously ranked probability score (CRPS) was approximated using posterior samples and an empirical cumulative distribution function using the \texttt{properscoring} package in Python.

The MMD was estimated using $m=32$ Monte Carlo samples, and the Wasserstein gradient flow was run using $p = 16$ particles and $\lambda_n = n^{1/2}$ for $4,000$ iterations, while MALA was run with $\dt t = 10^{-4}$ also for $4,000$ iterations after $96,000$ warmup iterations to ensure convergence.
\subsubsection{Sensitivity of the Wasserstein gradient flow}
\label{sec:sensitivity}
In Figure~\ref{fig:sensitivity} we show the sensitivity of the kernel gradient discrepancy (a diagnostic of the quality of the posterior samples) and two predictive scores to the choice of the number of particles and step size in a linear regression example.
We sample $n = 500$ data from the CR setting described in Section~\ref{sec:misspecification-regimes}, and then fit the \ac{pro} posterior with $p = \{2^1,2^2,\ldots,2^6\}$ particles and $\dt t$ between $0.1$ and $0.0001$.
In each case, we set $\lambda_n= n^{1/2}$ and run the WGF for $10,000$ iterations and estimate the MMD with $m = 32$ Monte Carlo samples.
In Figure~\ref{fig:sensitivity} we report the mean of the kernel gradient discrepancy (KGD), and predictive negative log likelihood (NLL) and MMD from $10$ fold cross-validation.
All metrics are improved as $\dt t$ is smaller, and as the number of particles grows.
We observed here, and in other experiments, that as long as the number of particles is sufficiently large to capture multi-modal behaviour in the data, the step size becomes important to tune appropriately.
And that the relationship relationship between predictive performance and step size is not always linear.
This observation led us to choose $p$ of the order $\approx2^5$, and optimise the step size in all experiments with the tuning-free FUSE schedule \citep{sharrock_tuning-free_2025}.
\begin{figure}[h!]
    \centering
    \begin{subfigure}{0.3\textwidth}
        \includegraphics[width=\linewidth]{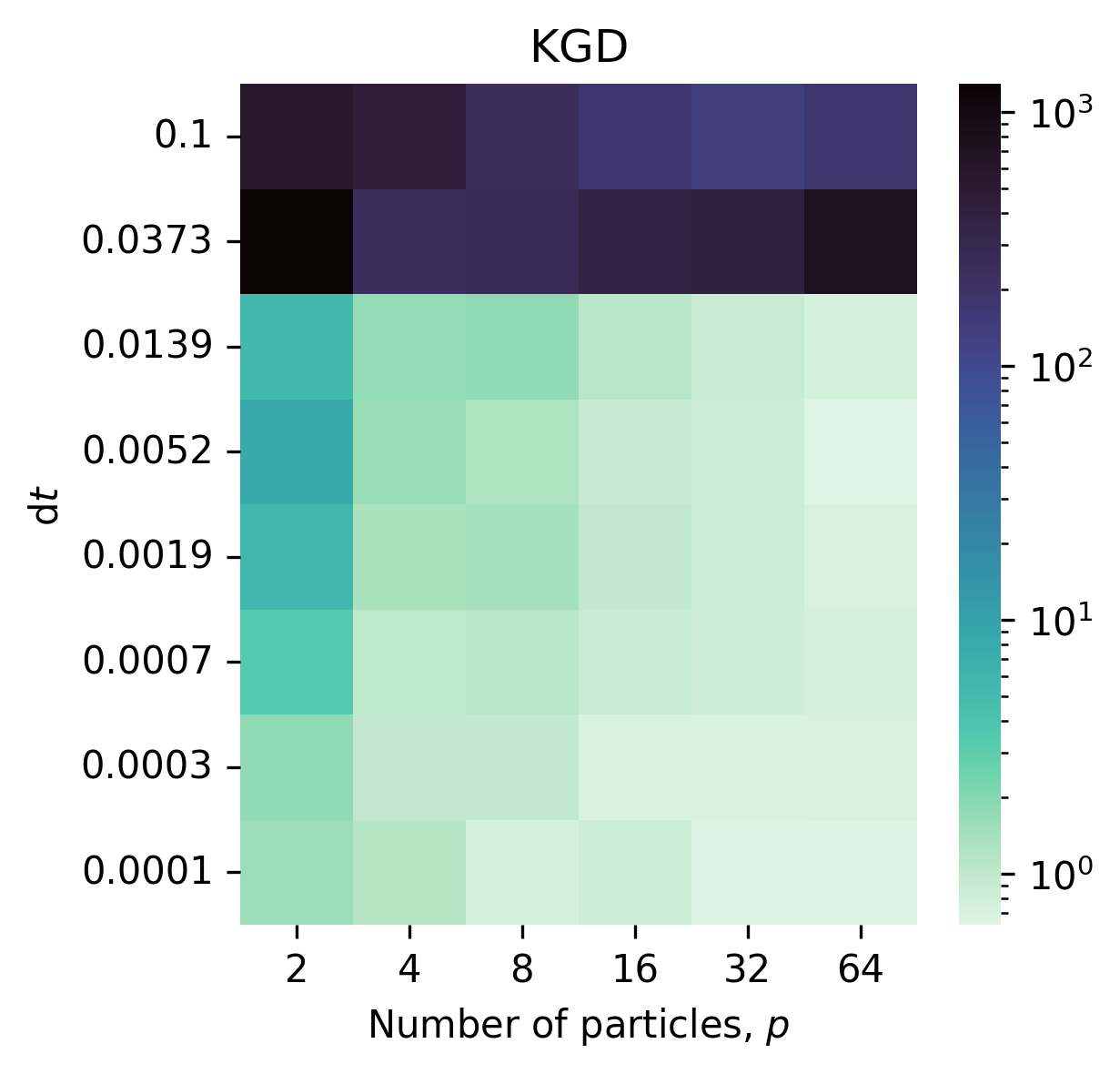}
        \caption{Posterior diagnostic.}
    \end{subfigure}
    \hfill
    \begin{subfigure}{0.6\textwidth}
        \includegraphics[width=\linewidth]{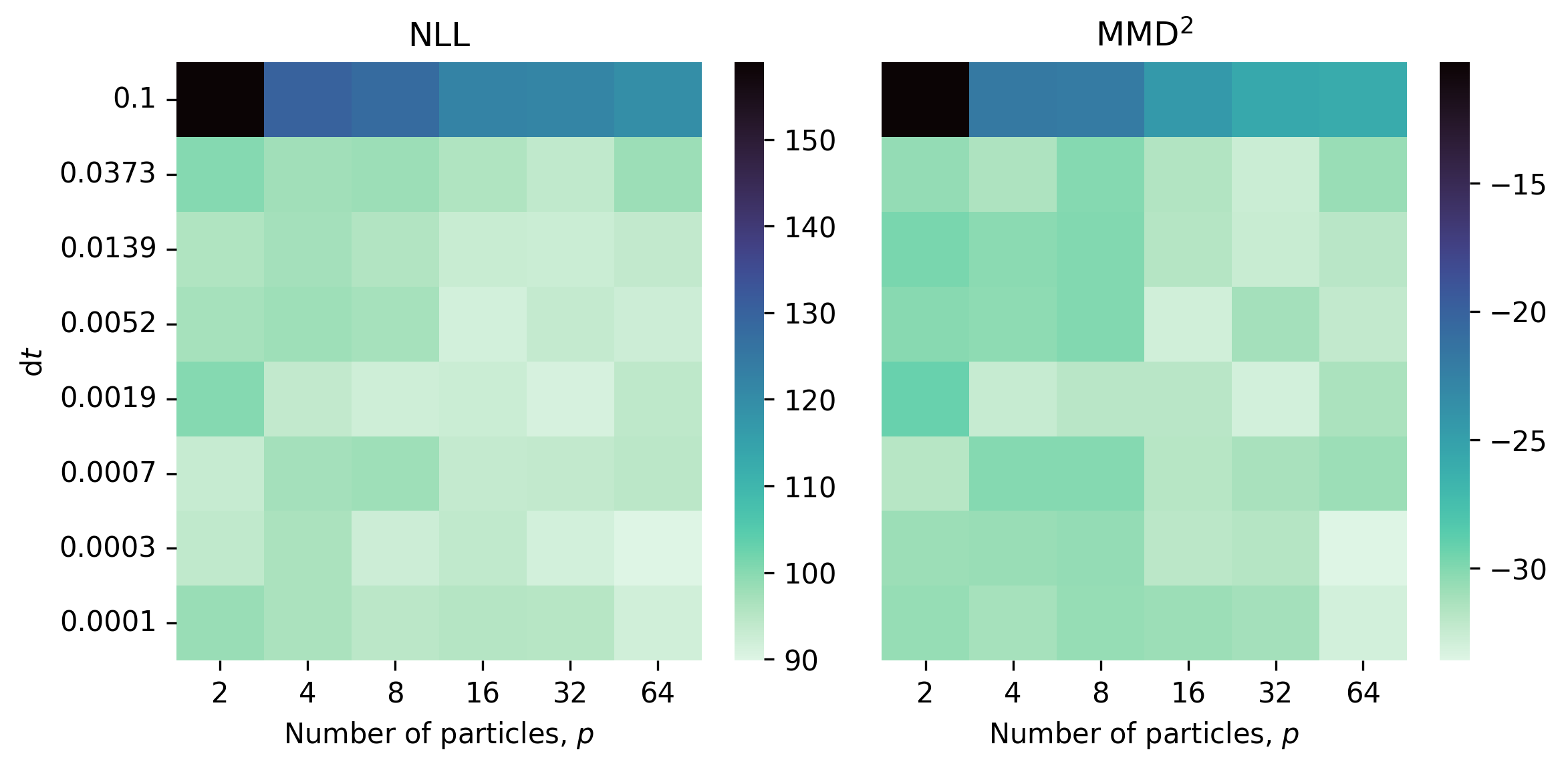}
        \caption{Predictive metrics.}
    \end{subfigure}
    \vspace*{-1em}
    \caption{\textit{Synthetic linear regression example}\/.}
    \label{fig:sensitivity}
\end{figure}
\subsection{Binary classification example}
We simulated $n = 1,000$ in the manner described in Section~\ref{sec:binary-classification}, and fit all models with $\lambda_n = n^{1/2}$.
The \ac{pro} posterior was run with $p = 64$ particles for $10,000$ iterations.
MALA was run for $4,000$ steps with $\dt t = 10^{-4}$.

\subsection{Boston housing data example}\label{app:supp-gmrf}
The \ac{pro} posterior under the MMD was run with $p = 16$ particles for $2,000$ iterations.
MALA was run for $4,000$ steps after $96,000$ warmup iterations, with $\dt t = 10^{-6}$.
The MMD was estimated with $m = 32$ Monte Carlo samples, and both models were fit with $\lambda_n = n^{1/2}$.
In the case of the log score, the Gibbs posterior is fit with the same hyper-parameters as described above, while the \ac{pro} posterior is now fit with $p = 128$ particles with $4,000$ iterations.
\begin{figure}[h!]
    \centering
    \includegraphics[width=0.666\linewidth]{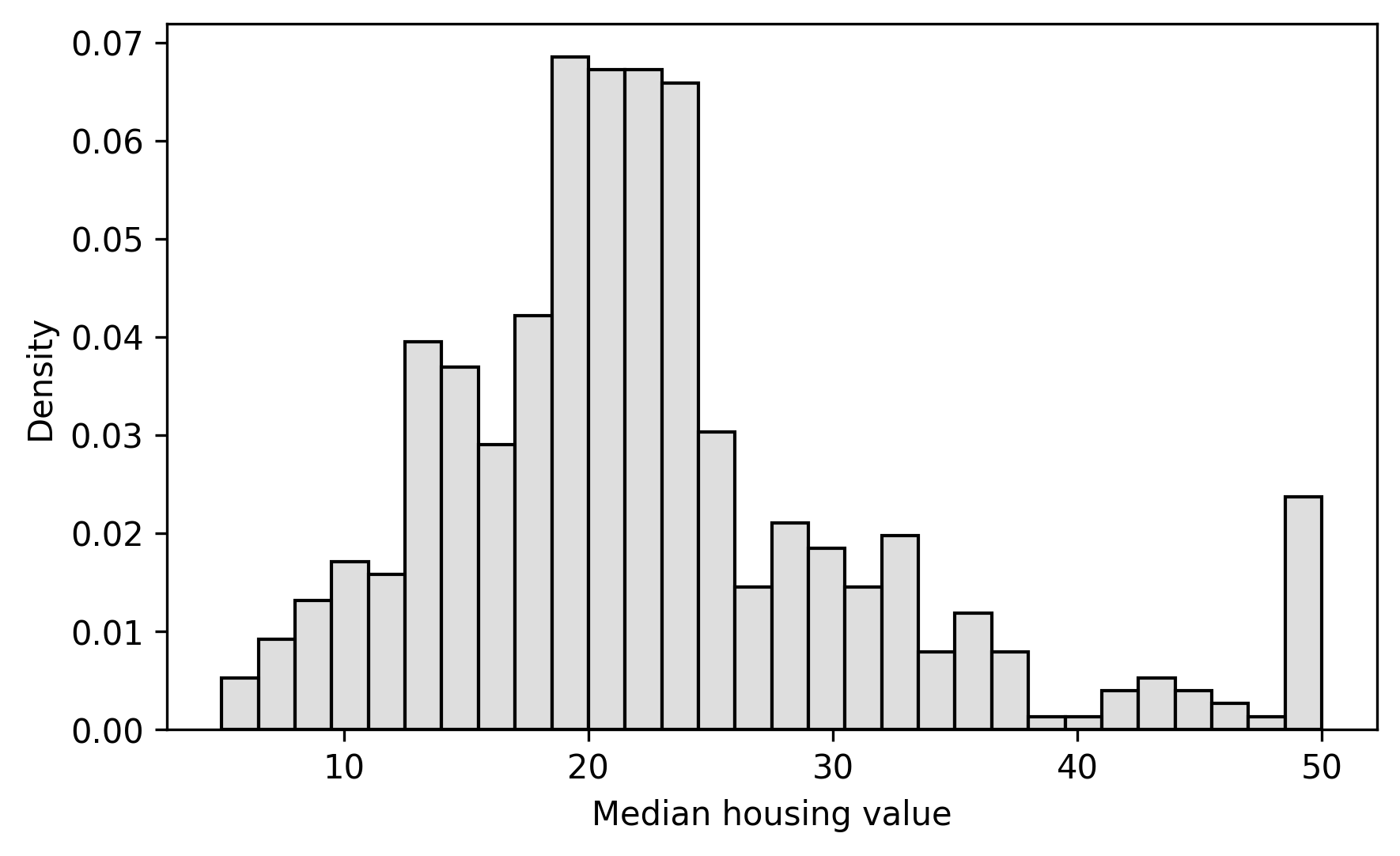}
    \caption{The empirical data distribution of the median housing value by census tract.}
    \label{fig:boston-data-distribution}
\end{figure}
\subsection{Redshift prediction example}\label{app:supp-sdss}
We use a subset of approximately $78{,}000$ observations from the seventeenth SDSS, retaining objects classified as either galaxies or quasars whose five filters all lie in the valid range $5 < \{u, g, r, i, z\} < 35$ and whose spectroscopic redshift satisfies $0 < z < 6$.
From the five filters we form the four colours $u-g$, $g-r$, $r-i$ and $i-z$, and keep the apparent $r$ magnitude as a measure of brightness inline with \citet{connolly1995slicing}.
The five inputs are standardised to zero mean and unit variance, expanded with degree-two polynomial features, and each resulting column is rescaled to have unit standard deviation.
This yields a total of $21$ parameters---an intercept, five linear terms, five squares and ten pairwise interactions---with the raw spectroscopic redshift used as the target.
The noise variance is fixed to the ordinary-least-squares residual standard deviation on the training fold.

The regression coefficients are given an independent standard Gaussian priors, and we use a learning rate $\lambda_n = 1/n$ so that we compare the \ac{pro} posterior with the standard Bayesian posterior.
The Bayes posterior is then sampled with MALA, using a step size of $\mathrm{d}t = 10^{-6}$, $96{,}000$ warm-up iterations and $4{,}000$ retained samples, initialised at the ordinary-least-squares parameter estimates.
The \ac{pro} posterior is obtained by simulating the Wasserstein gradient flow of the same score with $p=256$ particles for $4{,}000$ steps.
\section{Computation}
\subsection{Examples: MMD and logarithmic score}
\label{appendix:WGF-for-log-score-and-mmd}
Our above development recovers the algorithm of \citet{shenprediction2025} for the special case of squared MMD.
Here, elementary calculations show that for $\mathsf{L}_{\operatorname{MMD}}(\vartheta,\theta;x) = \langle\mu(P_\vartheta) - \mu(\updelta_{x}), 
\mu(P_{\theta}) - \mu(\updelta_x)\rangle_{\mathcal{H}}$ and $\nabla_1\mathsf{L}_{\operatorname{MMD}}(\vartheta,\theta;x)$ denoting its gradient with respect to $\vartheta$,
\begin{IEEEeqnarray}{rCl}
    \mathcal{W}(Q)[\vartheta] & = &
    \frac{1}{n}\sum_{i=1}^n \int \nabla_1\mathsf{L}_{\operatorname{MMD}}(\vartheta,\theta;x_i) \; \dt Q(\theta),
\nonumber
\end{IEEEeqnarray}
so that the evolution equations take the form
\begin{IEEEeqnarray}{rCl}
    \dt \vartheta_t^{(j)} = -\left\{\frac{ 1}{(p - 1)}\sum_{\substack{\ell = 1 \\ \ell \ne j}}^p
    \frac{ \lambda_n}{n}
    \sum_{i=1}^n 
    \nabla_1\mathsf{L}_{\operatorname{MMD}}(\vartheta_t^{(j)}, \vartheta_t^{(\ell)}; x_i)
    - \nabla_\vartheta \log \dt \Pi(\vartheta_t^{(j)})\right\} \dt t + \sqrt{2}\dt B_t^{(j)}.
    \nonumber
\end{IEEEeqnarray}

For the logarithmic scoring rule, one instead finds that 
\begin{IEEEeqnarray}{rCl}
     \mathcal{W}(q)[\vartheta] & = & -
     \frac{1}{n}\sum_{i=1}^n
     \frac{\nabla_\vartheta \,\dt P_{\vartheta}(x_i)}{
     \int \dt P_{\theta}(x_i) \dt Q(\theta) 
     }, 
\nonumber
\end{IEEEeqnarray}
so that the evolution equations are recovered as
\begin{IEEEeqnarray}{rCl}
    \dt \vartheta_t^{(j)} = -\left\{
    \frac{ \lambda_n}{n}
    \sum_{i=1}^n 
    \frac{
        \nabla_\vartheta \,\dt P_{\vartheta_t^{(j)}}(x_i)
    }{
        \frac{ 1}{(p - 1)}\sum_{\substack{\ell = 1 \\ \ell \ne j}}^p
        \dt P_{\vartheta_t^{(\ell)}}(x_i)
    }
    - \nabla_\vartheta \log \dt \Pi(\vartheta_t^{(j)})\right\} \dt t + \sqrt{2}\dt B_t^{(j)}.
    \nonumber
\end{IEEEeqnarray}

\subsection{Practicalities and Implementation}
\label{appendix:WGF-practicalities-and-implementation}

To provide intuition about how the sampling algorithm described in this section differs from more conventional Bayesian computation procedures, Figure~\ref{fig:initial-comparison-trajectories} compares it to the Metropolis adjusted Langevin algorithm \citep[MALA;][]{mala_1, mala_2} on several normal location examples.
While the $p=32$ particles which target the \ac{pro} posterior are initialised randomly from the prior, they soon repel each other, and lead to a multi-modal distribution.
In contrast, the single particle evolved according to MALA bounces around a singular mode.

While the interacting particle approach has the distinct advantage of being able to sample from any  \ac{pro} posterior, it introduces several hyper-parameters that may require careful tuning in practice.
The most important of these are the step sizes for the discretisation of the underlying SDE, the number of particles $p$, as well as the exact choice of $\lambda_n$.
Here, the choice of step size and the number of particles $p$ constitute a \textit{direct} trade-off between approximation quality and computational efficiency: 
choosing the step size too small or the number of particles too large will in principle increase accuracy, but do so at the cost of exploding computational cost and substantively slowed convergence.
Determining what constituted an appropriate step size in particular, is non-trivial and an active area of research.
Concurrently, \citet{sharrock_tuning-free_2025} developed an adaptive step size schedule for sampling form mean-field Langevin dynamics by adaptively tuning the step size, which they call FUSE.
We set the step size adaptively according to FUSE in all experiments.
While the theoretical result those authors showed do not extend to the \ac{pro} posterior, our empirical results give promise that similar results are attainable.
We leave this as an area of future research.
The relationship between the dimensionality of $\Theta$ and the number of particles required for consistent estimation has been investigated in the related family of Stein variational gradient descent algorithms by \citet[Theorems 1 and 4]{banerjee_improved_2025}.
While those authors found that the number of particles must scale polynomially or exponentially with the dimension of the problem, we observed that when the step size is chosen with FUSE that much fewer particles are required.
While little is known about choosing these hyper-parameters in our setting, recent work by \citet{chazal_computable_2025} may finally provide a principled way for diagnosing this (see, e.g., Appendix~\ref{sec:sensitivity} for a numerical demonstration).
In contrast, we have clear ideas on how $\lambda_n$ should be set from the theory developed in Section \ref{sec:theory-short}:
for the case covered in Theorem \ref{thm:DI-master} this is means setting $\lambda_n \asymp n^{1/2}/\log(n)$.

\end{document}